\documentclass[UTF-8,reqno]{amsart}
\usepackage{enumerate}
\usepackage{mhequ}
\usepackage{amssymb,url,color, booktabs,nccmath,amsmath}
\usepackage[left=2.7cm,right=2.7cm,top=3.5cm,bottom=3.5cm]{geometry}
\usepackage{longtable}
\usepackage{mathrsfs}
\usepackage{enumitem,dsfont,extpfeil}
\usepackage{relsize}
\usepackage{tikz}
\usepackage{bbm, dsfont}

\usepackage{color}
\usepackage[colorlinks=true]{hyperref}
\hypersetup{
	linkcolor=magenta,          
	citecolor= darkergreen,        
	filecolor=blue,      
	urlcolor=cyan
}

\definecolor{darkergreen}{rgb}{0.0, 0.5, 0.0}
\definecolor{myblue}{RGB}{0,0,139}

\numberwithin{equation}{section}
\def\theequation{\arabic{section}.\arabic{equation}}
\newcommand{\be}{\begin{eqnarray}}
	\newcommand{\ee}{\end{eqnarray}}
\newcommand{\ce}{\begin{eqnarray*}}
	\newcommand{\de}{\end{eqnarray*}}
\newtheorem{theorem}{Theorem}[section]
\newtheorem{proposition}[theorem]{Proposition}
\newtheorem{lemma}[theorem]{Lemma}
\newtheorem{assumption}[theorem]{Assumption}
\newtheorem{Examples}[theorem]{Example}
\newtheorem{corollary}[theorem]{Corollary}
\newtheorem{remark}[theorem]{Remark}

\newtheorem{definition}[theorem]{Definition}
\theoremstyle{definition}

\def\${|\!|\!|}

\DeclareMathOperator{\supp}{supp}

\def\Tr{\mathrm{Tr}}

\def\T{\mathbb{T}}

\def\eps{\varepsilon}

\def\<{{\langle}}
\def\>{{\rangle}}
\def\({{\Big(}}
\def\){{\Big)}}

\def\bx{{\mathbf{x}}}
\def\tr{\mathrm {tr}}
\def\W{{\mathcal W}}

\def\dif{{\mathord{{\rm d}}}}

\def\min{{\mathord{{\rm min}}}}

\def\={&\!\!=\!\!&}

\def\cB{{\mathcal B}}

\def\cE{{\mathcal E}}

\def\cH{{\mathcal H}}

\def\cK{{\mathcal K}}
\def\cL{{\mathcal L}}

\def\cN{{\mathcal N}}

\def\cT{{\mathcal T}}
\def\cU{{\mathcal U}}
\def\cV{{\mathcal V}}
\def\cW{{\mathcal W}}

\def\cZ{{\mathcal Z}}
\def\gH{{\mathfrak{H}}}
\def\gF{{\mathfrak{F}}}
\def\gS{{\mathfrak{S}}}

\def\1{{\mathbf{1}}}

\def\geq{\geqslant}
\def\leq{\leqslant}
\def\ge{\geqslant}

\def\eps{\varepsilon}

\def\bbH{\mathbb{H}}
\def\bbW{\mathbb{W}}
\def\<{{\langle}}
\def\>{{\rangle}}
\def\({{\Big(}}
\def\){{\Big)}}

\def\bx{{\mathbf{x}}}
\def\tr{\mathrm {Tr}}
\def\W{{\mathcal W}}

\def\dif{{\mathord{{\rm d}}}}

\def\min{{\mathord{{\rm min}}}}

\def\={&\!\!=\!\!&}
\def\bt{\begin{theorem}}
	\def\et{\end{theorem}}
\def\bl{\begin{lemma}}
	\def\el{\end{lemma}}
\def\br{\begin{remark}}
	\def\er{\end{remark}}
\def\bx{\begin{Examples}}
	\def\ex{\end{Examples}}
\def\bd{\begin{definition}}
	\def\ed{\end{definition}}
\def\bp{\begin{proposition}}
	\def\ep{\end{proposition}}
\def\bc{\begin{corollary}}
	\def\ec{\end{corollary}}

\def\geq{\geqslant}
\def\leq{\leqslant}
\def\ge{\geqslant}

\def\si{\phi}
 \def\R{\mathbb R}
 \def\R{\mathbb R}    
\def\N{\mathbb N}  
   
\def\<{\langle} \def\>{\rangle}

\makeatletter

\newcommand{\Rmnum}[1]{\expandafter\@slowromancap\romannumeral #1@}
\makeatother

\allowdisplaybreaks

\begin{document}
	\fontsize{10.0pt}{\baselineskip}\selectfont

	\title[Focusing $\Phi^6_1$ measure from many-body quantum mechanics]{Derivation of the focusing $\Phi^6_1$ measure in the optimal mass regime from many-body quantum Gibbs states}

	\author{Lin L\"u}  
	\address[L. L\"u]{School of Mathematics and Statistics, Beijing Institute of Technology, Beijing 100081, China}
	\email{lvlin0202@163.com}
	
	\author{Phan Th\`anh Nam}
	\address[P. T. Nam]{LMU Munich, Department of Mathematics, Theresienstrasse 39, 80333 Munich, Germany}
	\email{nam@math.lmu.de}
	
	\author{Rongchan Zhu}
	\address[R. Zhu]{School of Mathematics and Statistics, Beijing Institute of Technology, Beijing 100081, China}
	\email{zhurongchan@126.com}

	\begin{abstract}
		 We derive the focusing $\Phi^6_1$ measure on the torus $\mathbb{T}$ as the high-temperature/mean-field limit of many-body quantum Gibbs states with an attractive three-body interaction. The main difficulty in the focusing setting is to relate the classical mass cutoff to the quantum particle-number cutoff.
		Our result reaches the optimal mass threshold for the classical field identified by Oh, Sosoe, and Tolomeo \cite{OST22}, and thereby extends the earlier work of Rout and Sohinger \cite{RS25}.  At the critical threshold, the short-range interaction is allowed to shrink to a Dirac delta function on a logarithmic scale in the temperature parameter. Strictly below the threshold, the same convergence holds with a polynomial dependence on the temperature. Moreover, we establish a quantum-level phase transition at the same mass threshold. The proof develops the variational framework of Lewin, Nam, and Rougerie \cite{LNR15} in the focusing setting and relies on two new ingredients: a non-factorized trial state construction and a delicate tail estimate for the interacting lower symbol. These allow us to control the localization and relative entropy errors caused by the particle-number cutoff, as well as the contribution of the focusing exponential weight.
	\end{abstract}

	
	\maketitle
	\tableofcontents
	
	\section{Introduction}\label{sec:intro}
Nonlinear Gibbs measures provide a natural framework for the study of dispersive equations, constructive quantum field theory, and singular stochastic dynamics. Formally, a nonlinear Gibbs measure $\mu$ on $\Omega \subset \R^d$ is given by\footnote{The discussion in this introduction remains at a formal level, with rigorous definitions provided in Section~\ref{sec:problem formulation}.}
	\begin{align}\label{def:Gibbs:formal}
		\dif \mu(u)=\frac{1}{\mathscr{Z}} e^{-\cE(u)} \dif u,
	\end{align}
	where $\mathscr{Z}$ is the partition function, and the nonlinear Hamiltonian energy $\cE(u)$ is given by
	\begin{align}\label{def:Hamiltonian}
		\cE(u)=\frac12 \int_{\Omega} |\nabla u(x)|^2 \dif x \, \pm \, \frac{1}{p} \int_{\Omega} |u(x)|^p \dif x, \quad p\geq 2.
	\end{align}
  Measures of the form \eqref{def:Gibbs:formal} were first constructed in constructive quantum field theory; see, for example, \cite{GJ87, Nel73, Nel74, Sim74}. In the context of nonlinear Schr\"odinger equations (NLS), Bourgain showed in a series of seminal works that such measures are invariant under the NLS flow \cite{Bou94, Bou96, Bou97, Bou00}. Since then, this invariance has become a fundamental tool for constructing global solutions with rough random data; see, for instance, \cite{BB14a, BB14b, BDNY24, BS25, BT08a, BT08b, BTT13, Den12, DNY21, DNY24, OP16, Tzv08} and the references therein. These measures also arise as invariant measures and long-time limits of singular stochastic dynamics, thereby linking them to stochastic quantization and stochastic partial differential equations (SPDEs); see \cite{AR91, BG20, CC18, DD03, DHYZ25, DGR25, GH19, GH21,Hai14, Kup16, MW17a, MW17b, RZZ17} for further works in this direction.
	
	In this paper, we consider the mass-critical focusing case $p=6$ on the one-dimensional torus $\mathbb{T}=[0,1]$, where the nonlinear term in the Hamiltonian \eqref{def:Hamiltonian} carries a negative sign. The corresponding measure in \eqref{def:Gibbs:formal} is the focusing $\Phi^6_1$ measure, formally invariant under the 1D focusing quintic NLS. Since the Hamiltonian is unbounded from below, the construction of the focusing Gibbs measure requires suitable taming. In the seminal work \cite{LRS88}, Lebowitz, Rose, and Speer considered the focusing $\Phi^6_1$ measure with an $L^2$-cutoff:
	\begin{align}\label{phi61}
		\dif \mu^K(u):=\frac{1}{\mathscr{Z}^K} \mathrm{exp}\left(-  \int_{\mathbb{T}} \frac12 \left(|u'(x)|^2+|u(x)|^2 \right) \dif x+ \frac16 \int_{\mathbb{T}} |u(x)|^6 \dif x\right) \mathds{1}_{\{ \|u\|_{L^2(\mathbb{T}) } \leq K  \}}\dif u.
	\end{align}
We note that the general nonlinearities $|u|^p$ with $p\geq 3$ were also studied in \cite{LRS88}. For the focusing $\Phi^6_1$ measure in \eqref{phi61}, they identified the optimal threshold as the $L^2$-norm of the ground state $Q$ for the quintic NLS on $\R$, namely, the minimizer of the Gagliardo-Nirenberg-Sobolev inequality on $\R$ with $\|Q\|_{L^6(\R)}^6=3 \|Q'\|_{L^2(\R)}^2$. 
	It is expected that the following phase transition occurs:  $\mathscr{Z}^K<\infty$ if $K\leq  \|Q\|_{L^2(\R)}$, and $\mathscr{Z}^K=\infty$ if $K> \|Q\|_{L^2(\R)}$. The divergence in the supercritical regime was proved in \cite{LRS88}, while their argument for the normalizable regime $K< \|Q\|_{L^2(\R)}$ contains a minor gap. Subsequently, in \cite{Bou94}, Bourgain gave another proof, establishing normalizability for sufficiently small $K>0$. 
	
	More recently, in the breakthrough work \cite{OST22}, Oh, Sosoe, and Tolomeo proved	normalizability up to the optimal threshold $\|Q\|_{L^2(\R)}$. This sharp threshold is special to the one-dimensional mass-critical setting. In higher dimensions, the focusing problem is usually much more singular: for instance, the focusing $\Phi^4_2$ measure is not normalizable for any $K>0$; see \cite{BS96, OST24}. In such regimes, Hartree-type nonlocal interactions provide a useful substitute; see \cite{Bou97} in two and three dimensions and \cite{OOT24} for a refined three-dimensional analysis. For recent developments on the construction and phase transition of focusing Gibbs measures, as well as on the divergence regime and the associated dynamics, see also \cite{DRTW23, GLLOW25, HN25, LW22, RSTW25, TW23, Xia22}.
	
	Our goal is to derive the focusing $\Phi^6_1$ measure from many-body quantum Gibbs states and to show that the same critical threshold also appears at the quantum level. We consider the $n$-body Hamiltonian
	\begin{align}\label{n-body Schrodinger Hamiltonian}
		H_{\lambda,n}:=\frac12\sum_{i=1}^n \left( -\Delta_i +1 \right) -\frac{\lambda}{6} \sum_{1\leq j\neq k \neq \ell \leq n} w^\eps(x_j-x_k)w^\eps(x_j-x_\ell),
	\end{align}
	acting on the bosonic $n$-particle space $L^2_{\mathrm{sym}}(\mathbb{T}^n)$. Here, $-\Delta_i$ denotes the Laplacian acting on the $i$-th coordinate $x_i$, $w^\eps(\cdot)=\eps^{-1}w(\eps^{-1}\cdot)$ with $w\geq 0$, and $\lambda$ is the interaction strength. Note that the interaction term of \eqref{n-body Schrodinger Hamiltonian} is non-positive, so the construction of the quantum Gibbs state also requires a cutoff. Since the particle number $n$ varies, we work within the grand canonical ensemble. The grand canonical Gibbs state at temperature $\tau >0$ is
	\begin{align}\label{def:formal:Gibbs:state}
		\Gamma^{g_\eta}_\tau:=\frac{1}{\cZ^{g_\eta}_\tau} \bigoplus_{n=0}^{\infty}  e^{-H_{\lambda,n}/\tau} g_\eta \left( \frac{n}{\tau}\right),
	\end{align} 
	on the bosonic Fock space $\gF=\bigoplus_{n=0}^\infty L_{\mathrm{sym}}^2(\mathbb{T}^n)$. The scaling $ \lambda=\tau^{-2}$  is the natural choice in the mean-field/high-temperature regime: under the free Gibbs state at temperature $\tau$, the typical particle number is of order $\tau$. Hence, the one-body contribution has size $n \sim \tau$, whereas the three-body term has size $\lambda n^3\sim \lambda \tau^3$. Thus $\lambda=\tau^{-2}$ makes both terms contribute at order one in the Gibbs weight.
	
      From the many-body perspective, the sign in \eqref{n-body Schrodinger Hamiltonian} is crucial. Since $w^\eps\geq 0$, the three-body term is non-positive and becomes more negative when three particles are close on the $\eps$-scale. Thus, clustering lowers the energy, and the interaction is attractive. This is the many-body analogue of the focusing sign in \eqref{def:Hamiltonian}: the Gibbs weight favors concentration rather than penalizing it. With the opposite sign, configurations in which particles are close would instead increase the energy, corresponding to a repulsive (defocusing) model. 			
		At the quantum level, the sharp indicator lacks the regularity needed for functional calculus, so we replace it with a smooth family $g_\eta \to \mathds{1}_{[0, K^2] }$ as $\eta \to 0$, where $K\in ( 0,\|Q\|_{L^2(\R)}]$. Physically, the factor $g_\eta(n/\tau)$ restricts the grand canonical ensemble to the sector where the particle number per temperature stays below the critical mass threshold. In this sense, the cutoff $g_\eta(n/\tau)$ should be understood as the grand-canonical analogue of the classical $L^2$ cutoff in \eqref{phi61}: it prevents the attractive interaction from collapsing the partition function, while it retains the thermodynamically relevant states that converge to the classical focusing Gibbs measure.
	
The basic quantities are the interacting and free partition functions
	\begin{align*}
		\cZ^{g_\eta}_\tau =\sum_{n=0}^\infty \tr \left[ e^{-H_{\lambda,n}/\tau} g_\eta\left( \frac{n}{\tau}\right)  \right],
		\qquad 	\cZ_{\tau,0}=\sum_{n=0}^\infty \tr \left[ e^{-H_{0,n}/\tau} \right],
	\end{align*}
	where the trace is taken over $L^2_{\mathrm{sym}}(\mathbb{T}^n)$. When $w\in L^\infty(\mathbb{T})$, both $\cZ^{g_\eta}_\tau$ and $\cZ_{\tau,0}$ are finite for any fixed $\tau,\eps $ and $\eta $, but each diverges as $\tau \to \infty$. The meaningful object is therefore their ratio. Our first main result proves that, under an admissible relation among $\tau, \eps$ and $\eta$ whose precise form is described below, the relative partition function converges:
	\begin{align}\label{convergence:free:energy}
		\lim_{\tau \to \infty, \,\eps \to 0,\,  \eta \to 0} \frac{	\cZ^{g_\eta}_\tau}{\cZ_{\tau,0}} =\mathscr{Z}^K,
	\end{align}
	where $\mathscr{Z}^K$ is the partition function of the focusing $\Phi^6_1$ measure defined in \eqref{phi61}, for $K\in (0,\|Q\|_{L^2(\R)}]$. We further prove the convergence of the reduced density matrices of the grand-canonical Gibbs state, obtained by tracing out all but finitely many variables. Specifically, for any $K\in ( 0,\|Q\|_{L^2(\R)}]$ and every fixed $k\geq 1$,
	\begin{align}\label{convergence:density:matrices}
		\lim_{\tau \to \infty, \,\eps \to 0, \, \eta \to 0} \frac{1}{\tau^k \cZ^{g_\eta}_\tau} \sum_{n\geq k} \frac{n!}{(n-k)!} \tr_{k+1 \to n} \left[ e^{-H_{\lambda,n}/\tau} g_\eta \left( \frac{n}{\tau}\right)  \right] =\int_{L^2(\mathbb{T})} \big| u^{\otimes k} \big\rangle \big\langle u^{\otimes k} \big| \dif \mu^{K}(u),
	\end{align}
strongly in the trace-class. Here, the partial trace $\tr_{k+1 \to n} $ yields the $k$-body marginal of the $n$-body state and the sum
	over $n$ reflects the fluctuation of the particle number in the grand canonical ensemble. As in \cite[Sec. 2]{LNR14},  convergence of all these density matrices uniquely determines the limiting measure. Thus, \eqref{convergence:free:energy} and \eqref{convergence:density:matrices} give the optimal-mass
	derivation of the focusing $\Phi^6_1$ measure from many-body quantum Gibbs states.
	
A central feature of our result is that it distinguishes the critical and strictly subcritical mass regimes at the level of the microscopic interaction range. At the critical threshold $K=\|Q\|_{L^2(\R)}$, we obtain the convergence above while allowing the interaction range to shrink at the logarithmic scale; see Theorem~\ref{main:theorem} below:
	\begin{align*}
		\eps\gtrsim (\log \tau)^{-1/2}.
	\end{align*}
	If $K<\|Q\|_{L^2(\R)}$, the mass gap provides additional room in the quantum-to-Hartree estimates. Exploiting this room, we improve the admissible relation between $\eps$ and $\tau$ to the polynomial scale; see Theorem~\ref{main:theorem:sub} below:
	\begin{align*}
		\eps\gtrsim \tau^{-1/96}.
	\end{align*}
	Finally, we show that the relative partition function diverges once the cutoff exceeds the threshold, in agreement with the classical non-normalizability result of \cite{LRS88}. 
	More precisely, if $K>\|Q\|_{L^2(\R)}$ and $g\in C_c^\infty([0,\infty);\R_+)$ satisfies $g\equiv 1$ on $[0,K^2]$, then Theorem~\ref{blow:up:thm} below shows that
	\begin{align*}
		\lim_{\stackrel{\tau \to \infty,\, \eps \to 0,}{	\eps\gtrsim (\log \tau)^{-1/2}}}
		\frac{\sum_{n=0}^{\infty} \tr \!\left[e^{-H_{\lambda,n}/\tau} g\!\left( \frac{n}{\tau}\right)\right]}{\sum_{n=0}^{\infty} \tr \left[e^{-H_{0,n}/\tau}\right]}=\infty.
	\end{align*}
	Consequently, the many-body quantum model exhibits the same sharp phase transition as the classical focusing Gibbs measure \cite{OST22, LRS88}.
	
	The derivation of nonlinear Gibbs measures from many-body quantum Gibbs states was initiated by Lewin, Nam, and Rougerie \cite{LNR15}. Since then, this program has been extensively developed in various regimes \cite{CKRTG26, FKSS17, FKSS19, FKSS23, FKSS25, JR25, LNR15, LNR18, LNR21, NZZ26, Soh22}. In the defocusing case, they investigate the $\Phi^4_1$ theory \cite{FKSS17, FKSS19, LNR15, LNR18}, Hartree-type measures with nonlocal interactions in 2D and 3D \cite{FKSS17, FKSS23, LNR21, NZZ26, Soh22}, and the $\Phi^4_2$ theory in \cite{CKRTG26, FKSS25, JR25}. For more details on	the previously known literature, we refer the readers to the expository works \cite{FKSS20} and \cite{LNR19}.  In the most recent work \cite{NZZ25}, the second and third named authors, together with Xiangchan Zhu, rigorously derived the $\Phi^4_3$ measure from quantum Gibbs states by combining many-body methods with stochastic quantization.
	
For focusing models, Dinh and Rougerie \cite{DR26} recently studied one-dimensional attractive two-body interaction of Hartree type in the canonical ensemble. Since the particle number per temperature is fixed in this setting, rather than allowed to fluctuate as in the grand-canonical framework, the limiting Gibbs measure is supported on an $L^2$-sphere.
In the grand-canonical setting, the recent works of Rout and Sohinger \cite{RS23, RS25} treat the focusing $\Phi^4_1$ and $\Phi^6_1$ cases by complex-analytic methods based on a functional-integral representation of the interacting Gibbs state. In the quintic case, \eqref{convergence:free:energy} and \eqref{convergence:density:matrices} were partially established in \cite{RS25}, but only under certain restrictions:
	\begin{enumerate}[label=(\roman*)]
		\item the interaction range $\eps$ is not quantified in terms of the temperature $\tau$, and hence no explicit convergence rate for $\eps$ is obtained;
		\item the analysis only focuses on the regime where the mass cutoff $g$ in the focusing $\Phi^6_1$ measure is supported on sufficiently small intervals.
	\end{enumerate}
The present paper removes both restrictions by developing the variational approach of \cite{LNR15, LNR18,LNR21, NZZ25} in the focusing setting. It reaches the full classical normalizability range $K \in (0,\|Q\|_{L^2(\R)}]$ and proves the matching quantum blow-up above the threshold. More importantly, our result gives an explicit shrinking rate for $\eps$ in the critical and subcritical regimes. As a byproduct, our strategy also recovers and extends the focusing $\Phi^4_1$ result of \cite{RS23}, which corresponds to the standard two-body interaction potentials.

	\section{Problem formulation and main results}\label{sec:problem formulation}
	In this section, we introduce the notation, formulate the classical and quantum problems, and state the	main results together with an outline of the proof and the main difficulties.
	\subsection{Classical and quantum models}
In this subsection, we recall the definitions of the classical and
quantum models used throughout the paper.
	\subsubsection*{The general setup}
	Consider the one-dimensional torus $\mathbb{T}=[0,1]$. The
	one-particle space is $\mathfrak{H}=L^2(\mathbb{T;\mathbb{C}})$.
	The \textit{bosonic $n$-particle} space $\gH^{(n)}$ is the subspace of $\gH^{\otimes n}$ symmetric under permutations, defined by
	\begin{align}\label{def:n-particle:space}
		\gH^{(n)}=\{u\in \gH^{\otimes n}: u(x_1,...,x_n)=u(x_{\sigma(1)},...,x_{\sigma(n)}), \,  \forall (x_1,...,x_n) \in \mathbb{T}^n, \, \forall \sigma \in S_n \}.
	\end{align} 
	We define the one-body operator $h$ by
	\begin{align}\label{def:one-body operator}
		h=\frac12 (-\Delta+1),
	\end{align}
	which is a positive, self-adjoint, densely defined operator on $\mathfrak{H}$. The spectral decomposition of $h$ is given by
	\begin{align}\label{spectral expansion}
		h=\sum_{k\in \mathbb{Z}} \lambda_k  |u_k\rangle  \langle u_k|,
	\end{align}
	where the eigenvalues $\lambda_k$ and the corresponding normalized eigenfunctions $u_k$ are 
	\begin{align}\label{def:eigenvalues}
		\lambda_k=\frac12 ((2\pi k)^2+1), \qquad u_k=e^{2\pi i kx}.
	\end{align}
	In particular, it follows from \eqref{spectral expansion} and \eqref{def:eigenvalues} that 
	\begin{align}\label{fact:finite trace}
		\tr(h^{-1})=\sum_{k\in \mathbb{Z}} \frac{1}{\lambda_k}<\infty.
	\end{align} 
	
	\subsubsection*{Conventions}
	We denote by $\mathbb{N}:=\{1,2,3,...\}$ and $\mathbb{N}_0: =\mathbb{N} \cup \{0\}$. We write $a \lesssim  b$ if $a\leq cb$ for some constant $c>0$. We denote by $C$ a generic positive constant whose value may change from line to line. Let $\R_+:=[0,\infty)$. For an interval $I\subset \R$, we denote by $C_c^\infty(I;\R_+)$ the space of non-negative smooth functions with compact support in $I$. Let 
	\begin{align}\label{optimal:threshold}
		\cK_c:=\|Q\|_{L^2(\R)}
	\end{align} 
	denote the optimal threshold identified in \cite[Theorem 1.4]{OST22}.
	Let $\cH$ be a separable Hilbert space. We denote by $\mathds{1}$ the identity operator on $\cH$. For any subset $A\subset \cH$, let $\mathds{1}_A$ denote the characteristic function of $A$. Furthermore, $\cB(\cH)$ denotes the space of bounded linear operators on $\cH$ and $\cL(\cH)$ denotes the space of compact operators on $\cH$. Recall that the Schatten class $\mathfrak{S}^q(\cH)$ consists of operators $\cT \in \cL(\cH)$ with norm $\|\cT\|_{\mathfrak{S}^q(\cH)}<\infty$, where
	$$	 \|\cT\|_{\mathfrak{S}^q(\cH)}=
	\begin{cases}
		\left( \tr (|\cT|^q) \right)^{1/q} & q< \infty,\\
		\mathrm{sup}  \,  \mathrm{spec} |\cT| &  q= \infty.
	\end{cases}
	$$
	Here $|\cT|:=\sqrt{\cT^*\cT}$ and spec denotes the spectrum of an operator.	
	On the \textit{bosonic $n$-particle} space $\gH^{(n)}$, we define the symmetric tensor product as
	\begin{align*}
		f_1 \otimes_s \cdot \cdot \cdot \otimes_s f_n:= \frac{1}{\sqrt{n!}} \sum_{\sigma \in S_n}   f_{\sigma(1)} \otimes \cdot \cdot \cdot \otimes f_{\sigma(n)} ,
	\end{align*}
	for $f_1,...,f_n \in \gH$.
	If $\xi$ is a closed linear operator on $\gH^{(p)}$,  we can identify $\xi$ with its Schwartz kernel, denoted by $\xi(x_1,...,x_p; y_1,...,y_p)$; see for example \cite[Corollary V.4.4]{RS80}.

	\subsubsection*{The classical model}
	We begin with the non-interacting case. Consider the product probability space $(\mathbb{C}^{\N},\mathcal{G},\nu)$, where $\nu=\mathop  \bigotimes_{n= 1}^\infty \nu_n$ and $ \nu_n:= \frac{1}{\pi} e^{-|z|^2} \dif z$ denotes the standard complex Gaussian measure. Let $\{g_n\}_{n\in \N}$ be a sequence of independent standard complex-valued Gaussian random variables defined via the coordinate mapping $g_n(\omega):=\omega(n)$. We relabel the eigenpairs $\{(u_n,\lambda_n)\}_{n\in \N}$ of $h$, ordered so that $\lambda_n$ is nondecreasing, and define the classical free field $\phi$ via the series expansion
	\begin{align}\label{define:GFF}
		\phi= \sum_{n=1}^{\infty} \frac{g_n}{\sqrt{\lambda_n}} u_n.
	\end{align}
	Since $\tr(h^{-1})<\infty$, it is straightforward to verify that $\phi \in H^{\frac12-}(\mathbb{T})$, $\nu$-almost surely. The pushforward of the measure $\nu$ under the mapping $\phi$ is a complex Gaussian measure on $H^{\frac12-}(\mathbb{T})$ with covariance operator $h^{-1}$. This pushforward measure is referred to as the free Gibbs measure, denoted by $\mu_0$, and admits the formal density
	\begin{align}\label{def:free Gibbs measure}
		\dif \mu_0(u)=\frac{1}{\mathscr{Z}_0} \mathrm{exp}\left(- \int_{\mathbb{T}} \frac12 \left(|u'(x)|^2+|u(x)|^2 \right) \dif x\right) \dif u,
	\end{align}
	which must be interpreted in the sense of finite-dimensional projections. More precisely, the free Gibbs measure $\mu_0$ associated with $h$ is the unique probability measure (see, e.g. \cite[Sec. 3.1]{LNR15}), whose cylindrical projection onto the subspace $V_J=\mathrm{span}\{u_1,...,u_J \}$  for any $J\geq 1$ is given by
	\begin{align*}
		\dif \mu_{0,J}(u):=\mathop  \bigotimes \limits_{j = 1}^J \left(\frac{\lambda_j}{\pi} e^{-\lambda_j|\alpha_j|^2} \dif \alpha_j\right),
	\end{align*}
	where $\alpha_j=\langle u,u_j \rangle $ and $\dif \alpha_j =\dif \mathfrak{R} (\alpha_j) \dif \mathfrak{I}(\alpha_j) $ denotes the Lebesgue measure on $\mathbb{C}\simeq \R^2$. We remark that the measure $\mu_0$ is supported on $H^{\frac12-}(\mathbb{T})$, and hence there is no need to renormalize the mass in the interaction terms.
	
	We now turn to the interacting case and consider the classical potential energy $\cW$ defined by
	\begin{align}\label{def:Phi61:potential}
		\cW(u):= \frac16 \int_{\mathbb{T}} |u(x)|^6 \dif x.
	\end{align}
	Let $\mu_0$ and $\W$ be defined as in \eqref{def:free Gibbs measure} and \eqref{def:Phi61:potential}, and define the full focusing $\Phi^6_1$ measure by
	\begin{align}\label{def:Phi61:measure}
		\dif \mu^{\cK_c}(u): =\frac{1}{\mathscr{Z}^{\cK_c}} e^{\cW(u)}\mathds{1}_{\{\|u\|_{L^2} \leq \cK_c\}}\,\dif \mu_0(u),
	\end{align}
	where $\mathscr{Z}^{\cK_c}=\int_{\mathfrak{H}} e^{\cW(v)}\mathds{1}_{\{ \|v\|_{L^2} \leq \cK_c \}} \,\dif \mu_0(v)$ is the normalization constant (partition function). By \cite[Theorem 1.4]{OST22}, we have that $\mathscr{Z}^{\cK_c}<\infty$, and hence $\mu^{\cK_c}$ is well-defined. The potential energy $\cW$ can be interpreted as a local interaction corresponding to the Dirac delta function $\delta_0$. In this paper, we analyze the quantum model by first deriving a Hartree measure with a nonlocal interaction potential, and then
	proving its convergence to the focusing $\Phi^6_1$ measure \eqref{def:Phi61:measure}. To this end, we introduce the following nonlocal interaction with a regular kernel.
	
	\begin{assumption}\label{assum:on:w}
		Let $w\in L^\infty(\R)$ be an even, nonnegative function with compact support such that
		\begin{align*}
			\int_{\R} w(x)\dif x =\int_{\R} |w(x) |\dif x =1.
		\end{align*}
		Moreover, for $\eps>0$, we periodize $w$ and treat it as a function on $\mathbb{T}$ by setting
		\begin{align}\label{def:weps}
			w^\eps(x)=\sum_{k\in \mathbb{Z}} \eps^{-1} w(\eps^{-1}(x+k)). 
		\end{align}
	\end{assumption}
	
	Observe that Assumption~\ref{assum:on:w} implies that $w^\eps \in L^\infty(\T)$, $w^\eps $ converges weakly to the Dirac measure $\delta_0$ and $  \int_{\mathbb{T}} w^\eps(x)\dif x = \int_{\mathbb{T}} |w^\eps(x) |\dif x =1$. Let $w$ be as in Assumption~\ref{assum:on:w}. For $\eps>0$, we define the Hartree interaction $\mathcal{W}^\eps$ by
	\begin{align}\label{Hartree:interaction}
		\mathcal{W}^\eps(u):= \frac{1}{3!}\int w^\eps(x-y)w^\eps(x-z) |u(x)|^2 |u(y)|^2 |u(z)|^2 \dif x \dif y \dif z =\frac{1}{3!}  \int (w^\eps*|u|^2)^2 |u(x)|^2 \dif x.
	\end{align}
	Let $g\in C^\infty_c([0,\infty);\R_+)$ be supported in $[0,K^2]$ for some $K<\infty$. Since $w\in L^\infty(\mathbb{T})$, one verifies that for any fixed $\eps >0$, $	\mathscr{Z}^{\eps,g}:=\int_{\mathfrak{H}} e^{\cW^\eps(v)}g(\|v\|_{L^2}^2)\,\dif \mu_0(v) \leq \|g\|_{L^\infty}e^{\eps^{-2} \|w\|^2_{L^\infty} K^6}<\infty$. This allows us to define the following Hartree measure
	\begin{align}\label{def:Hartree:measure}
		\dif \mu^{\eps,g}(u): =\frac{1}{\mathscr{Z}^{\eps,g}} e^{\cW^\eps(u)}g(\|u\|_{L^2}^2)\,\dif \mu_0(u).
	\end{align}
	Moreover, the classical interaction $\mathcal{W}^\eps$ can be written as
	\begin{align*}
		\mathcal{W}^\eps=  \frac{1}{3!} \int \overline{u}(x_1) \overline{u}(x_2) \overline{u}(x_3)   W^\eps(x_1,x_2,x_3; y_1,y_2,y_3) u(y_1) u(y_2) u(y_3) \dif x_1 \dif x_2 \dif x_3 \dif y_1 \dif y_2 \dif y_3,
	\end{align*}
	where $W^\eps$ is the three-body operator on $\mathfrak{H}^{(3)}$ with kernel 
	\begin{equation}\label{def:operator W}
		\begin{aligned}
			W^\eps(x_1,x_2,x_3;y_1,y_2,y_3):=\frac{1}{3}&\big(w^{\varepsilon}(x_1-x_2)w^{\varepsilon}(x_1-x_3)+ w^{\varepsilon}(x_2-x_1) w^{\varepsilon}(x_2-x_3) 
			\\&+w^{\varepsilon}(x_3-x_1)w^{\varepsilon}(x_3-x_2)\big)
			\times \delta(x_1-y_1) \delta(x_2-y_2) \delta(x_3-y_3).
		\end{aligned}
	\end{equation}
	In particular, $W^\eps$ acts as a multiplication operator on $\gH^{(3)}$ given by
	\begin{align*}
		\frac13	\big(w^{\varepsilon}(x_1-x_2)w^{\varepsilon}(x_1-x_3)+ w^{\varepsilon}(x_2-x_1) w^{\varepsilon}(x_2-x_3) 
		+w^{\varepsilon}(x_3-x_1)w^{\varepsilon}(x_3-x_2) \big).
	\end{align*}
	Note that $W^\eps$ is a non-negative, self-adjoint, and bounded operator on $\gH^{(3)}$.

	\subsubsection*{The quantum model}
	The underlying many-body Hilbert space is the \textit{bosonic Fock space}
	\begin{align*}
		\mathfrak{F}=\mathfrak{F}(\mathfrak{H}):=\mathbb{C} \oplus \bigoplus_{p=1}^\infty\mathfrak{H}^{(p)},
	\end{align*}
	where $\mathfrak{H}^{(p)}$ denotes the $p$-particle symmetric subspace defined in \eqref{def:n-particle:space}. Since we work in the grand canonical setting, the number of particles is not fixed, and its expectation
	in a quantum state $\Gamma$ is given by $\tr[\cN \Gamma]$,  where the \textit{number operator} is
	\begin{align*}
		\cN=\bigoplus_{n=0}^\infty n \mathds{1}_{\gH^{(n)}}.
	\end{align*}
	If $\Gamma$ commutes with $\cN$, then it admits a diagonal decomposition of the form $\Gamma=\bigoplus_{k=0}^\infty  (\Gamma)_k$. The $n$-body \textit{reduced density matrices} of $\Gamma$ are defined by taking partial traces as follows:
	\begin{align}\label{def:density matrix}
		\Gamma^{(n)}:= \sum_{k\geq n} { k \choose n} \tr_{n+1 \to k} \big[ (\Gamma)_k\big] ,\qquad \forall n\geq 1.
	\end{align}
	We may interpret the reduced density matrices as the quantum analogue of marginal probability density functions. For every self-adjoint operator $A_n$ acting on $\gH^{(n)}$, one has the equivalent formulation
	\begin{align}\label{characterization:second:quan}
		\tr\big[ A_n \Gamma^{(n)} \big] =	\tr\big[ \mathbb{A}_n \Gamma \big] ,
	\end{align}
	where $\mathbb{A}_n$ denotes the \textit{second quantization} of $A_n$, defined by
	\begin{align*}
		\mathbb{A}_n=0\oplus \cdot \cdot \cdot \oplus 0\bigoplus_{k=n}^\infty \left( \sum_{1 \leq i_1 < ...< i_n \leq k} (A_n)_{i_1,...,i_n} \right) .
	\end{align*}
	Before proceeding with the quantum Hamiltonian, we recall the definitions of the creation and annihilation operators on $\gF$. For $\Psi= (\Psi^{(p)})_{p\in \N}\in \gF$ and $f\in \gH$, they are defined by
	\begin{align*}
		(a^\dagger(f)\Psi)^{(p)}(x_1,...,x_p)&=\frac{1}{\sqrt{p}} \sum_{i=1}^p f(x_i)  \Psi^{(p-1)}(x_1,...,x_{i-1},x_{i+1},...,x_p),
		\\(a(f)\Psi)^{(p)}(x_1,...,x_p)&=\sqrt{p+1} \int \overline{f(x)}\Psi^{(p+1)}(x,x_1,...,x_p).
	\end{align*}
	In particular, for each $k\in \mathbb{Z}$, we denote $a^\dagger_k=a^\dagger(u_k)$ and $a_k=a(u_k)$. These operators are also related to the operator-valued distributions $a^\dagger(x)$ and $a(x)$ via
	\begin{align*}
		a^\dagger(f)=\int_{\mathbb{T}} a^\dagger(x)f(x) \dif x, \qquad a(f)=\int_{\mathbb{T}}  a(x) \overline{ f(x)}\dif x.
	\end{align*}
	They satisfy the canonical commutation relations for all $f,g\in \gH$ :
	\begin{align}\label{CCR}
		[a(f),a^\dagger(g)]= \langle f,g\rangle \mathds{1}_{\gF(\gH)}, \qquad [a^\dagger(f),a^\dagger(g)]=[a(f),a(g)]=0.
	\end{align}
	Moreover, the $n$-particle density matrix of a state $\Gamma$, defined in \eqref{def:density matrix}, can also be expressed in terms of the creation and annihilation operators
	\begin{align}\label{characterize of density matrix}
		\big\langle f_1 \otimes_s \cdot \cdot \cdot \otimes_s f_n, \Gamma^{(n)}  g_1 \otimes_s \cdot \cdot \cdot \otimes_s g_n \big\rangle =\tr \left(  a^\dagger(g_1) \cdot \cdot \cdot a^\dagger(g_n) a(f_1) \cdot \cdot \cdot a(f_n) \Gamma  \right).
	\end{align}
	
	The creation and annihilation operators allow one to rewrite the second quantization of any $n$-body operator with a specified orthonormal basis for $\gH$.
	The number operator is defined as  
	\begin{align*}
		\mathcal{N}= \dif \Gamma (\mathds{1}) := \sum_{j\geq 1}  a^\dagger_j a_j .
	\end{align*} 
	For the one-body operator $h$ as in \eqref{def:one-body operator}, the \textit{quantum free Hamiltonian} is given by
	\begin{align*}
		\mathbb{H}_0=\dif \Gamma(h):= 0 \oplus \bigoplus_{n=1}^\infty   \left(\sum_{j=1}^n h_j\right)= \sum_{i,j\geq 1} \langle hu_i,u_j \rangle a^\dagger_i a_j =\int_{\mathbb{T}} a^\dagger(x)(ha)(x) \dif x.
	\end{align*}
	For $w$ as in Assumption~\ref{assum:on:w} and $W^\eps$ as in \eqref{def:operator W}, the \textit{quantum interaction} is defined by
	\begin{equation}\label{Hartree:quantum:interaction}
		\begin{aligned}
			\mathbb{W} :=& \, 0\oplus 0 \oplus 0 \bigoplus_{k=3}^\infty \left( \sum_{1 \leq i_1 < i_2< i_3 \leq k} (W^\eps)_{i_1,i_2,i_3} \right) 
			\\ =& \frac{1}{3!}\int a^\dagger(x)a^\dagger(y) a^\dagger(z)w^\eps(x-y)w^\eps(x-z)a(x)a(y)a(z)  \dif x \dif y \dif z.
		\end{aligned}
	\end{equation}
	Here, $\mathbb{W}$ is the second quantization of the three-body operator $W^\eps$, and therefore defines a non-negative operator on the Fock space $\gF$.	
	The many-particle interacting quantum system is then characterized by the \textit{quantum interacting Hamiltonian}
	\begin{align}\label{quantum interacting Hamiltonian}
		\mathbb{H}_{\tau}= \mathbb{H}_{\tau,0} - \mathbb{W}_{\tau}:= \frac{1}{\tau}\mathbb{H}_{0} -\frac{1}{\tau^3}\mathbb{W}.
	\end{align}
	The free Gibbs state is defined by 
	\begin{align}\label{def:free:Gibbs:state}
		\Gamma_{\tau,0}:= \frac{1}{ \cZ_{\tau,0}}e^{-\mathbb{H}_{\tau,0}}, \qquad \cZ_{\tau,0}:= \tr_{\gF(\gH)} \left( e^{-\mathbb{H}_{\tau,0}} \right).
	\end{align}
	We remark that the free Gibbs state $	\Gamma_{\tau,0}$ is well defined, since $\tr_{\gF(\gH)} [e^{-\bbH_{\tau,0}}]<\infty$ for any $\tau>0$. 
	As in the classical setting, a truncation of the particle number operator $\cN$ is required to ensure that the interacting Gibbs state is well-defined. Let $g\in C^\infty_c([0,\infty);\R_+ )$ and define the interacting Gibbs state by
	\begin{align}\label{def:interacting:gibbs_state}
		\Gamma^g_\tau:= \frac{1}{\mathcal{Z}^g_{\tau}}  e^{-\mathbb{H}_{\tau}} g\left(\frac{\mathcal{N}}{\tau}\right), \qquad \mathcal{Z}^g_{\tau}:=\Tr_{\gF(\gH)} \left(e^{-\mathbb{H}_\tau} g\left(\frac{\mathcal{N}}{\tau}\right)\right),
	\end{align}
	where $\mathcal{Z}^g_{\tau}$ is the partition function. Indeed, for any fixed $\tau>0$ and $\eps>0$, the partition function $\mathcal{Z}^g_{\tau}$ is finite. Consequently, $\Gamma^g_\tau$ defines a quantum state, i.e., a nonnegative trace-class operator on $\gF(\gH)$ satisfying $\tr_{\gF(\gH)}[\Gamma^g_\tau]=1$. Our goal is to recover the full focusing
	$\Phi^6_1$ measure with a sharp indicator cutoff. Such cutoffs, however, are not suitable at the quantum level, as they lack the regularity
	required for several essential estimates. We therefore work with smooth, compactly supported cutoff functions that approximate the sharp indicator cutoff. This necessitates additional conditions on the admissible
	cutoffs at the quantum level.

	\begin{assumption}\label{assum:cutoff}
		For any $K>0$ and $\eta \in\big(0,\frac12 K^2 \big)$, let $f_\eta \in C^\infty_c([0,\infty);[0,1])$ be a non-trivial function such that $f_\eta(s)=1$ for $s\leq K^2-\eta$ and $f_\eta(s)=0$  for $s> K^2$. Moreover, we assume that $\|f_\eta^{(j)}\|_{L^\infty}\lesssim \eta^{-j} $ for any $j\in \N$ and that $f_\eta \to 1_{[0,K^2)}$ pointwise as $\eta \to 0$.
	\end{assumption}
	\noindent
	Such functions can be constructed by mollification. Let $\{\vartheta_\ell\}_{\ell >0}$ be a family of standard mollifiers supported in $(-\ell,\ell)$, and define $g_\ell(x)=1_{(-\infty,K^2-\ell ]}(x)$. We then set $f_\eta(x)=(g_{\eta/2}*\vartheta_{\eta/2})(x)$, which satisfies Assumption~\ref{assum:cutoff}.
	In the sequel, the cutoff parameter $K$ may vary according to the required support of the	cutoff functions.

	\subsection{Main results}
We state the main results in the order in which they are proved. Theorem~\ref{main:theorem} gives the optimal-mass derivation of the focusing $\Phi^6_1$ measure, including the endpoint $\cK_c$. Theorem~\ref{main:theorem:sub} proves that, strictly below the threshold, the convergence rate of the interaction range $\eps$ can be improved to a polynomial scale. Theorem~\ref{Rmk:Hartree} isolates the quantum-to-Hartree convergence at fixed $\eps>0$, where no normalizability threshold is needed. Theorem~\ref{blow:up:thm} proves the complementary supercritical divergence and hence identifies the same
sharp threshold at the quantum level.
	\begin{theorem}[\bf{Derivation of the focusing $\Phi^6_1$ measure in the optimal mass regime}]\label{main:theorem}
		Let $w$ satisfy Assumption~\ref{assum:on:w}, and let $\cK_c=\|Q\|_{L^2(\R)}$ be the optimal threshold defined in \eqref{optimal:threshold}. Let $\cK \in(0,\cK_c]$ be arbitrary. For all sufficiently large $\tau>0$, let $\eta \in \big[ \tau^{-\frac{1}{64}} ,\frac{1}{2}\cK^2  \big)$, and let the cutoff $f_\eta$ satisfy Assumption~\ref{assum:cutoff} with cutoff parameter $\cK$. Assume $1>\eps \geq M \left(\log \tau \right)^{-1/2}$, where $M=(2\|w\|^2_{L^\infty}+1)(\cK^2+4)^3$. Denote by $\Gamma^{f_\eta}_\tau$ the interacting Gibbs state as in \eqref{def:interacting:gibbs_state}, with cutoff $f_\eta$, and by $\mathcal{Z}^{f_\eta}_{\tau}$ the associated partition function. Then, as $\tau \to \infty$ and $\eps, \eta \to 0$, the relative partition function converges:
		\begin{align}\label{main:result:partition}
			\left|	 \frac{\cZ^{f_\eta}_{\tau}}{\cZ_{\tau,0}} -  \int_{\mathfrak{H}} e^{\cW(u)} \mathds{1}_{\{\|u\|_{L^2} \leq \cK \}} \dif \mu_0(u) \right| \to 0,
		\end{align}
		where $\cW=\frac16 \int_{\mathbb{T}} |u(x)|^6 \dif x$ is defined in \eqref{def:Phi61:potential}.
		Moreover, as $\tau \to \infty$ and $\eps, \eta \to 0$, we have trace-class convergence of the density matrices: for all $k\geq 1$,
		\begin{align}\label{main:result:density matrices}
		 \left\| \frac{k!}{\tau^k}(\Gamma^{f_\eta}_{\tau})^{(k)} -\int_{\gH}|u^{\otimes k}\rangle\langle u^{\otimes k}|\,\dif \mu^{\cK}(u) \right\|_{\gS^1(\gH^{(k)})}	\to 0,
		\end{align}		
		where $\dif \mu^{\cK}(u)=\frac{1}{\mathscr{Z}^\cK} e^{\cW(u)}\mathds{1}_{\{\|u\|_{L^2} \leq \cK \}}  \,\dif \mu_0$ is the focusing $\Phi^6_1$ measure.
	\end{theorem}
 Thus, Theorem~\ref{main:theorem} identifies both the limiting free energy and the limiting density matrices. In particular, for $\cK\leq \cK_c$, the quantum Gibbs state converges to the focusing $\Phi^6_1$ measure itself, rather than merely at the level of partition functions. This is the first result to cover the optimal range of the mass cutoff, including the threshold case. When $\cK>\cK_c$, the situation is completely different; we return to this case later in Theorem~\ref{blow:up:thm}. 
 
 This result strengthens \cite[Theorem 1.10]{RS25} in two independent directions. First, it reaches the full range of classical normalizability $\cK\in(0,\cK_c]$, including both sharp $L^2$-cutoffs $\mathds{1}_{\{\|u\|_{L^2}\leq \cK\}}$ for any $\cK\in(0,\cK_c]$, and smooth mass cutoffs $g(\|u\|_{L^2}^2)$ with $g\in C_c^\infty([0,\cK_c^2];\R_+)$. Second, unlike in \cite{RS25}, where the dependence between $\eps$ and $\tau$ was not quantified and hence no precise convergence rate was available, Theorem~\ref{main:theorem} makes this dependence quantitative. In particular, it allows $\eps$ to converge within the logarithmic admissible regime $\eps \gtrsim (\log \tau)^{-\frac12}$. Moreover, in the mass-subcritical regime $\cK<\cK_c$, this logarithmic dependence can be further improved to a polynomial dependence on $\tau$, as shown in the next main theorem.
		\begin{theorem}[\bf{Derivation of the focusing $\Phi^6_1$ measure in the mass-subcritical regime}]\label{main:theorem:sub}
		Let $w$ satisfy Assumption~\ref{assum:on:w}, and let $\cK \in(0,\cK_c)$ be arbitrary. For all sufficiently large $\tau>0$, let $\eps \in [\tau^{-\frac{1}{96}},1)$, $\eta \in \big[ \tau^{-\frac{1}{64}} ,\frac{1}{2}\cK^2  \big)$, and let the cutoff $f_\eta$ satisfy Assumption~\ref{assum:cutoff} with cutoff parameter $\cK$. Let $\Gamma^{f_\eta}_\tau$ be the interacting Gibbs state as in \eqref{def:interacting:gibbs_state}, let $\mathcal{Z}^{f_\eta}_{\tau}$ denote the corresponding partition function, and let $\cW$ and $\mu^{\cK}$ be as in Theorem~\ref{main:theorem}. Then, as $\tau \to \infty$ and $\eps, \eta \to 0$, the relative partition function converges:
		\begin{align}\label{main:result:partition:sub}
			\left|	 \frac{\cZ^{f_\eta}_{\tau}}{\cZ_{\tau,0}} -  \int_{\mathfrak{H}} e^{\cW(u)} \mathds{1}_{\{\|u\|_{L^2} \leq \cK \}} \dif \mu_0(u) \right| \to 0.
		\end{align}
	Moreover, the corresponding density matrices converge in trace class: for all $k\geq 1$,
		\begin{align}\label{main:result:density matrices:sub}
			\left\| \frac{k!}{\tau^k}(\Gamma^{f_\eta}_{\tau})^{(k)} -\int_{\gH}|u^{\otimes k}\rangle\langle u^{\otimes k}|\,\dif \mu^{\cK}(u) \right\|_{\gS^1(\gH^{(k)})}	\to 0.
		\end{align}		
	\end{theorem}
	\begin{remark}
		The exponent $1/96$ in the admissible scale $\eps \geq \tau^{-1/96}$
		is not expected to be optimal. 
		A major open challenge is to reach the physically relevant scale $\eps \ll \tau^{-1}$, suggested by the diluteness condition, which requires the interaction range to be much smaller than the typical inter-particle distance (see, e.g., \cite{Rou20} for a general discussion on different scaling regimes). Achieving this scale seems beyond the reach of the current methods.
	\end{remark}
	
	\begin{remark}
	The same argument also applies to the cubic case with a two-body interaction potential. In this setting, the focusing $\Phi^4_1$ measure is well defined under an arbitrary $L^2$-cutoff $ \mathds{1}_{\{\|u\|_{L^2} \leq K \}}$, $K<\infty$, as summarized in \cite[Theorem 1.1]{OST22}. Consequently, the conclusion of Theorem~\ref{main:theorem:sub} extends directly to the focusing $\Phi^4_1$ measure with  $L^2$-cutoff $ \mathds{1}_{\{\|u\|_{L^2} \leq K \}}$ for any $K>0$, with the same polynomial dependence on $\eps$. This recovers and extends the focusing $\Phi^4_1$ result of \cite{RS23}.
	\end{remark}
We emphasize that the threshold phenomenon does not appear at the Hartree level when $\eps>0$ is fixed. The interaction is bounded on every compact mass set, so the quantum model converges to the Hartree measure for any compactly supported smooth cutoff $g$. The support of $g$ is required to lie below the threshold only when passing from the Hartree measure to the focusing $\Phi^6_1$ measure. We thus obtain the following Hartree-level convergence result.
	\begin{theorem}[\bf{Derivation of the Hartree measure}]\label{Rmk:Hartree}
		Let $w$ satisfy Assumption~\ref{assum:on:w}, fix $\eps >0$, and let $g\in C^\infty_c([0,\infty);\R_+)$. Let $\Gamma^g_\tau$ be the interacting Gibbs state defined in \eqref{def:interacting:gibbs_state} with cutoff $g$, and let $\mathcal{Z}^g_{\tau}$ be the associated partition function. Let $\cW^{\eps}$ be the Hartree interaction defined in \eqref{Hartree:interaction} and let $\mu^{\eps,g}$ be the Hartree measure defined in \eqref{def:Hartree:measure} with cutoff $g$. Then, as $\tau \to \infty$,
		\begin{align*}
			\left|	 \frac{\mathcal{Z}^g_{\tau}}{\cZ_{\tau,0}} -  \int_{\mathfrak{H}} e^{\cW^\eps(u)} g \left( \|u\|_{L^2}^2 \right) \dif \mu_0(u) \right| \to 0.
		\end{align*}
		Moreover, we have the convergence of the density matrices for all $k\geq 1$,
		\begin{align*}
		 \left\| \frac{k!}{\tau^k}(\Gamma_\tau^g)^{(k)} -\int_{\gH}|u^{\otimes k}\rangle\langle u^{\otimes k}|\,\dif \mu^{\eps,g}(u) \right\|_{\gS^1(\gH^{(k)})}	\xrightarrow{\tau \to \infty } 0.
		\end{align*}
	\end{theorem}
	
	Here, fixing $\eps >0$ is essential: it makes $\int_{\mathfrak{H}} e^{\cW^\eps(u)} g \left( \|u\|_{L^2}^2 \right) \dif \mu_0(u)$ finite for every $g\in C^\infty_c([0,\infty) ;\R_+)$.
Theorem~\ref{Rmk:Hartree} therefore extends \cite[Theorem 1.8]{RS25}, where the cutoff is required to have sufficiently small support.	
We now turn to the supercritical regime $\cK>\cK_c$. At the classical level, this is precisely the regime in which the focusing $\Phi^6_1$ measure is no longer normalizable; see \cite[Theorem 2.2(b)]{LRS88}. The following theorem shows that the quantum model captures the same transition: once the cutoff exceeds $\cK_c$, the relative partition function diverges.
	\begin{theorem}[\bf{Blow-up above the mass threshold}]\label{blow:up:thm}
		Let $w$ satisfy Assumption~\ref{assum:on:w}. For any $\cK>\cK_c$, let $g\in C_c^\infty([0,\infty);\R_+)$ satisfy $g \equiv 1$ on $[0,\cK^2]$. Assume $1>\eps \geq M \left(\log \tau \right)^{-\frac{1}{2}}$, where $M=(2\|w\|^2_{L^\infty}+1)(\cK^2+4)^3$. Let the quantum Hamiltonians $\mathbb{H}_\tau$ and $\mathbb{H}_{\tau,0}$ be defined as in \eqref{quantum interacting Hamiltonian}. Then, as $\tau \to \infty$ and $\eps \to 0$, the relative partition function diverges:
		\begin{align*}
			\frac{\tr_{\gF(\gH)}\left(e^{-\mathbb{H}_\tau} g(\cN/\tau) \right)}{\tr_{\gF(\gH)}
				\left(e^{-\mathbb{H}_{\tau,0}}\right)} \to \infty.
		\end{align*}
	\end{theorem}
	
	Taken together, Theorems~\ref{main:theorem} and \ref{blow:up:thm} describe a sharp quantum phase transition at the cutoff $\cK_c$. Below and at the threshold $\cK_c$, the normalized partition function converges to a finite classical value, and the many-body system has a well-defined focusing $\Phi^6_1$ limit. Above the threshold, the same normalized quantity diverges. This matches the classical normalizability picture proved in \cite{LRS88,OST22}.

	\begin{remark}
		Although the preceding discussion concerns the time-independent case, our analysis can be extended to the time-dependent setting considered in \cite{FKSS19, RS25}. In particular, the convergence of the time-dependent correlation functions established in \cite[Theorem 1.14]{RS25} can be extended to the full focusing $\Phi^6_1$ measure. The proof combines the Schwinger-Dyson expansion developed in \cite{FKSS19}, the time-independent convergence established in Theorem~\ref{main:theorem} and a stability analysis for the quintic NLS. As this argument is essentially parallel to those of \cite{FKSS19, RS25} and not directly related to the variational framework of the present paper, we do not pursue this direction further.
	\end{remark}
	
	\subsection{Outline of variational proof and difficulties}\label{Outline of variational proof and difficulties}
	The proof of Theorem~\ref{main:theorem} has two layers. First, we derive a Hartree-type measure from the quantum Gibbs state by adapting the variational framework of \cite{LNR15, LNR18, LNR21, NZZ25} to the present focusing model with a particle-number cutoff. Second, we pass from the Hartree interaction to the local $\Phi^6_1$ interaction. The first layer is summarized by the following general-cutoff statement.
	
	\begin{proposition}\label{Thm:Hartree}
		Let $w$ satisfy Assumption~\ref{assum:on:w}, and let $\cK>0$ be arbitrary. For all sufficiently large $\tau>0$, let $\eta \in \big[ \tau^{-\frac{1}{64}} ,\frac{1}{2}\cK^2  \big)$, and let the cutoff $f_\eta$ satisfy Assumption~\ref{assum:cutoff} with cutoff parameter $\cK$. Assume $1>\eps \geq M\left(\log \tau \right)^{-\frac12}$, where $M=(2\|w\|^2_{L^\infty}+1)(\cK^2+4)^3$.
	Denote by $\cZ^{f_\eta}_{\tau}$ and $\mathcal{Z}^{f_\eta}_{\tau,0}$ the associated partition functions. Then, as $\tau \to \infty$,
		\begin{align}\label{key:bound1}
			\left| \frac{\cZ^{f_\eta}_{\tau}}{\cZ^{f_\eta}_{\tau,0}} - \frac{\int_{\gH} e^{ \cW^\eps(u) } f_\eta(\|u\|_{L^2}^2) \dif \mu_0(u)}{\int_{\gH} f_\eta(\|u\|_{L^2}^2) \dif \mu_0(u) } \right| \to 0,
		\end{align}
		where $	\cW^\eps$ is the Hartree interaction defined in \eqref{Hartree:interaction}. Moreover, for all $k\geq 1$,
		\begin{align}\label{key:bound2}
				\left\| \frac{k!}{\tau^k}(\Gamma_{\tau}^{f_\eta})^{(k)} - \int_{\gH} |u^{\otimes k}\rangle  \langle u^{\otimes k} | \dif \mu^{\eps,f_\eta}(u) \right\| _{\gS^1(\gH^{(k)})}  \to 0,
		\end{align}
		where $	\dif  \mu^{\eps,{f_\eta}}(u):=\frac{e^{\cW^\eps(u)} f_\eta(\|u\|_{L^2}^2) \dif   \mu_{0}(u)}{\int e^{\cW^\eps(v)} f_\eta(\|v\|_{L^2}^2) \dif   \mu_{0}(v)}$ denotes the Hartree measure on $\gH$.
	\end{proposition}
	Proposition~\ref{Thm:Hartree} is proved in Section~\ref{sec3:lower:upper:free-energy}. For the quantum-to-Hartree argument itself, the same proof applies to any fixed $g\in C_c^\infty([0,\infty);\R_+)$. Once the cutoff is fixed, $\eta$ can also be treated as fixed, and it suffices to keep quantitative control of the errors in $\tau$ and $\eps$. This substantially simplifies the argument. The additional assumptions on $f_\eta$ are needed only in the proof of Theorem~\ref{main:theorem}, where the cutoff approaches the critical regime as $\eta \to 0$, and the resulting cutoff errors must be balanced in a $\tau$-dependent way using Proposition~\ref{lemma:non-interacting} below. 
	After Proposition~\ref{Thm:Hartree}, the normalizability result of \cite[Theorem 1.4]{OST22}, and a classical approximation argument identify the local focusing $\Phi^6_1$ limit. It is only in this second layer of the argument that $\supp f_\eta \subset [0,\cK_c^2]$ is required.

We now outline the variational strategy for the proof of Proposition~\ref{Thm:Hartree}. The starting point is the pair of quantum and classical Gibbs variational principles, which reduce the problem to comparing two minimization problems with different reference states or measures. Recall that, by \cite[(1.9)]{LNR15}, the interacting Gibbs state $\Gamma^{f_\eta}_\tau$, defined in \eqref{def:interacting:gibbs_state}, is the $\textit{unique minimizer}$ of the variational problem\footnote{We remark that applying the variational principle to the free Gibbs state $\Gamma_{\tau,0}$ leads to an extra term $\tr\left( \Gamma \log f_\eta(\cN/\tau)\right)$ on the right-hand side of \eqref{variational principle quantum}. Similarly, its application to the free Gibbs measure $\mu_0$ leads to an extra term $\int_{\gH} \log f_\eta(\|u\|_{L^2}^2) \dif \nu(u)$ on the right-hand side of \eqref{variational principle classical}. These terms pose significant difficulties for proving upper and lower bounds. For this reason, we employ truncated free Gibbs states and measures here.}:
	\begin{align}\label{variational principle quantum}
		-\log  \frac{\cZ^{f_\eta}_{\tau}}{\cZ^{f_\eta}_{\tau,0}}= \inf_{\substack{\Gamma\geq 0\\ \tr_{\gF}\Gamma=1}} \Big\{ \cH(\Gamma,\Gamma^{f_\eta}_{\tau,0}) -  \Tr[ \bbW_\tau \Gamma]  \Big\},
	\end{align}
	where $\cH(\Gamma,\Gamma')= \tr\big[ \Gamma(\log \Gamma -\log \Gamma') \big]\geq 0$ denotes the quantum relative entropy and
	$\Gamma^{f_\eta}_{\tau,0}$ is the truncated free Gibbs state with cutoff $f_\eta$
	\begin{align}\label{def:free:gibbs_state}
		 \Gamma^{f_\eta}_{\tau,0} := \frac{1}{\mathcal{Z}^{f_\eta}_{\tau,0}}  e^{-\mathbb{H}_{\tau,0}} f_\eta\left(\frac{\mathcal{N}}{\tau}\right), \qquad \mathcal{Z}^{f_\eta}_{\tau,0}:=\Tr_{\gF(\gH)} \left(e^{-\mathbb{H}_{\tau,0}} f_\eta \left(\frac{\mathcal{N}}{\tau}\right)\right),
	\end{align} 
	On the classical side, the Hartree measure $\mu^{\eps,f_\eta}$ is the $\textit{unique minimizer}$ of the classical variational problem
	\begin{align}\label{variational principle classical}
		-\log \Big( \int_{\gH} e^{\cW^\eps(u) } \dif \mu^{f_\eta}_0 (u)\Big) = \inf_{\nu \ll \mu^{f_\eta}_0} \Big\{ \mathcal{H}_{\mathrm{cl}} (\nu,\mu^{f_\eta}_0) - \int_{\gH} \cW^\eps(u)  \, \dif \nu(u)  \Big\},
	\end{align}
	where $\cH_{\mathrm{cl}}$ is the classical relative entropy between two probability measures and $\mu^{f_\eta}_0$ is the truncated free Gibbs measure defined by 
	\begin{align}\label{def:truncated:Gaussian:measure}
		\dif \mu^{f_\eta}_0(u):=\frac{f_\eta(\|u\|_{L^2}^2)\,\dif \mu_0(u)}{\int_{\mathfrak{H}}f_\eta(\|v\|_{L^2}^2)\,\dif \mu_0(v)}.
	\end{align}
	
The comparison is performed after a finite-dimensional localization. We introduce the low-frequency projection $P=\mathds{1}(h\leq \Lambda_e)$ and split the free energy into low- and high-momentum components. The low modes form a finite-dimensional problem, where semiclassical tools and the quantitative quantum de Finetti theorem are available. The high modes are treated as an error term with a uniformly controlled
contribution. In this way, the quantum variational problem is connected to the classical variational problem in \eqref{variational principle classical}. However, obtaining relative free-energy upper and lower bounds with errors explicit in $\Lambda_e$, $\eta$, $\eps$ and $\tau$ poses additional difficulties.
	
For the upper bound, the localization is implemented by choosing a suitable trial state. In the earlier defocusing works \cite{LNR15, LNR18, LNR21}, a factorized ansatz was sufficient: one used the tensor product of the interacting Gibbs state on the low modes and the free Gibbs state on the high modes. In the present setting, however, this choice becomes problematic. The particle-number cutoff $f_\eta\left(\cN/\tau\right)$ couples the low and high modes, and is therefore not compatible with a tensor-product decomposition.  Using a factorized cutoff in the localized reduction would produce relative-entropy errors involving $\tr\big[\Gamma_{\mathrm{ts}}(\log f_\eta(\cN/\tau) -\log (f_\eta(\cN_P/\tau) \otimes f_\eta(\cN_Q/\tau)  ) )\big]$, which are difficult to control effectively. We therefore keep the full particle-number cutoff and use the non-factorized trial state
\begin{align*}
\Gamma_{\mathrm{ts}}:=\frac{e^{-\bbH_{\tau,0}+\bbW_{\tau,P}}f_\eta\left(\cN/\tau\right)}{\tr\Big(e^{-\bbH_{\tau,0}+\bbW_{\tau,P}}f_\eta\left(\cN/\tau\right)\Big)}.
\end{align*}
This allows the interaction and relative-entropy localization errors to be estimated directly; see Lemmas~\ref{errors:localization_interaction} and \ref{lemma:relative entropy} below.
It remains to derive a classical lower bound for the localized quantum Gibbs weight. To this end, we use the coherent states to express $\tr_{\gF(P\gH)}( e^{-\bbH_{\tau, P}} f_\eta(\cN_P/\tau))$ as an integral over the classical field $u\in P\gH$.  Then, decomposing the coherent state into $n$-particle sectors identifies the integrand with a Poisson expectation. This is the mechanism that converts the quantum cutoff $f_\eta(\cN_P/\tau)$ into the classical cutoff $f_\eta(\|Pu\|_{L^2}^2)$, and the quantum Hamiltonian $\bbH_{\tau,P}$ into the classical Hartree functional $\cW^{\eps}(Pu)- \langle Pu,hPu \rangle$, up to controlled errors. After these substitutions, the principal term can be bounded from below and yields the desired quantity $\int_{P\gH} e^{\cW^\eps(Pu)-\langle Pu,hPu\rangle} f_\eta(\|Pu\|_{L^2}^2) \dif u$.

The lower bound follows the same localization scheme up to the finite-dimensional reduction. The subsequent semiclassical comparison, however, is more delicate. One is naturally led to compare the finite-dimensional partition functions $	\int_{P\gH} e^{\cW^\eps} \dif  \mu_{P,0}^{\tau^{-1}}$ and $	\int_{P\gH}  e^{\cW^\eps} \dif   \mu^{f_\eta}_{0,P}$, where $\mu_{P,0}^{\tau^{-1}}$ is the lower symbol associated with the truncated free Gibbs state, and $\mu^{f_\eta}_{0,P}$ is the projected truncated Gaussian measure. The $L^1$-estimate $\|\mu_{P,0}^{\tau^{-1}}-\mu^{f_\eta}_{0,P} \|_{L^1(P\gH)}$ is not sufficient by itself to control these integrals, because the focusing weight $e^{\cW^\eps }$ may be large. On the other hand, a pointwise comparison of the two densities, as in \cite{NZZ25}, is technically cumbersome in the presence of cutoffs. Our key observation is to insert an interacting lower-symbol tail estimate before comparing the two partition functions. More precisely, for any $R$ strictly larger than the support parameter $\cK^2$ of the cutoff, the interacting lower-symbol $\mu_{P,\tau}^{\tau^{-1}}$ has negligible contribution on the region $\{\|Pu\|_{L^2}^2>R\}$; namely, by Lemma~\ref{lemma:interacting:tail} below, 
 \begin{align*}
 	\int_{\{\|Pu\|_{L^2}^2>R\}} \cW^\eps \dif \mu_{P,\tau}^{\tau^{-1}}\ll 1.
 \end{align*} 
It follows that the variational lower bound can be reduced to an integral in which the exponential interaction weight is effectively restricted to a bounded-mass region. Thus, it remains only to compare
\begin{align*}
	\int_{P\gH} e^{\cW^\eps \mathds{1}_{\{\|Pu\|^2_{L^2}\leq R\}}} \dif  \mu_{P,0}^{\tau^{-1}}, \qquad \mathrm{and} \qquad 	\int_{P\gH} e^{\cW^\eps } \dif  \mu^{f_\eta}_{0,P}.
\end{align*}
On this restricted region, the $L^1$ comparison $\|\mu_{P,0}^{\tau^{-1}}-\mu^{f_\eta}_{0,P} \|_{L^1(P\gH)}$ is sufficient, and yields the matching projected lower bound. Combining this with the refined estimate on the classical partition function and a Cauchy argument then removes the projection and yields \eqref{key:bound1}.

This completes the argument for general cutoffs, but it also reveals a quantitative limitation. At the endpoint $\cK_c$, the tail estimate is only available above the cutoff level, namely for $R>\cK^2_c$. The remaining bounded-mass region must therefore be controlled by the rough pointwise exponential bound $e^{C_R\eps^{-2}}$, which leads to the logarithmic admissible scale for $\eps$ in terms of  $\tau$. In the strictly subcritical regime, one can instead choose an intermediate mass $\cK_s$ with $\cK<\cK_s<\cK_c$ and apply the tail estimate at $R=\cK_s^2$. On the complementary region $\|Pu\|_{L^2}\leq \cK_s $, the mass remains strictly below the critical threshold, so the Oh-Sosoe-Tolomeo moment bound in \cite[Sec. 4]{OST22} controls the focusing exponential weight. Thus, Lemma~\ref{lemma:interacting:tail} removes the large-mass contribution, while the bounded-mass part is handled by subcritical exponential integrability. This eliminates the pointwise exponential loss in the free energy lower bound.
The same mechanism is used in the proof of the free energy upper bound. In the general-cutoff argument, both the relative-entropy localization error and the final semiclassical error are controlled by rough pointwise exponential bounds of the form $e^{O(\eps^{-2})}$. In the subcritical regime, the intermediate mass level $\cK_s$, Lemma~\ref{lemma:interacting:tail}, and the Oh-Sosoe-Tolomeo estimate make it possible to replace these pointwise bounds with uniform estimates on truncated interacting partition functions. This replacement improves the admissible scale from logarithmic to polynomial and is precisely where the subcritical mass gap enters the error analysis.

		\subsubsection*{Organization of the paper}
	This paper is organized as follows. In Section~\ref{sec3:lower:upper:free-energy} we analyze the variational problem with a general cutoff. We derive upper and lower bounds on the free energy and obtain the Hartree measure from the quantum model. Section~\ref{sec:mass:subcritical:regime} treats the same problem in the mass-subcritical regime, where the admissible range of the interaction scale $\eps$ is improved. In Section~\ref{sec:hartree to phi61}, we recover the focusing $\Phi^6_1$ measure from the Hartree measure. Section~\ref{Proof of main results} completes the proofs of the main theorems. To make the paper self-contained, we also include in Appendix~\ref{sec:appendix A} a proof of the quantitative convergence estimate for the truncated free Gibbs state, and in Appendix~\ref{sec:appendix:B} a uniform projected Oh-Sosoe-Tolomeo bound.

	\section{From quantum model to Hartree measure: a general cutoff regime}\label{sec3:lower:upper:free-energy}		
In this section, we prove the convergence from the quantum model to the Hartree measure under a general cutoff condition, in which the cutoff is not assumed to be supported below the critical threshold. As explained in Subsection~\ref{Outline of variational proof and difficulties}, the argument is variational, but it can be carried out only after localizing the problem to a finite-dimensional subspace. More precisely, in Subsection~\ref{subsec:fock}, we recall the localization on Fock space and the associated lower symbols.
These tools allow us to compare the quantum and classical variational problems in finite dimensions. Subsections~\ref{sec:upper bound} and \ref{sec:lower bound} then establish matching upper and lower bounds for the localized free energy. Subsection~\ref{subsec:proof:hartree} is devoted to removing the finite-dimensional projection, which completes the proof of Proposition~\ref{Thm:Hartree} and, as an immediate consequence, yields Theorem~\ref{Rmk:Hartree}.
We begin by recalling the localization of the problem on Fock space.
	\subsection{Fock-space localization and de Finetti measure}\label{subsec:fock}
	Let $P$ be a finite-dimensional orthogonal projection on $\mathfrak{H}$ and let $Q=\mathds{1}-P$. Since $\mathfrak{H}=(P\mathfrak{H})\oplus (Q\mathfrak{H})$, we have the unitary equivalence
	\begin{align*}
		\mathfrak{F} =\mathfrak{F} ((P\mathfrak{H})\oplus (Q\mathfrak{H}))   \approx  \mathfrak{F}(P\mathfrak{H})\otimes 	\mathfrak{F}(Q\mathfrak{H}),
	\end{align*}
	namely there is a unitary
	\begin{align}\label{def:unitary}
		\cU : \gF ( P\gH \oplus Q\gH) \mapsto \gF(P\gH) \otimes \gF (Q\gH),
	\end{align}
	satisfying $\cU\cU^* = \mathds{1}$ and acting on creation operators  as
	\begin{align}\label{property:unitary}
		\cU a ^\dagger( f ) \cU ^*  = a ^\dagger(Pf) \otimes  \mathds{1} +  \mathds{1} \otimes a^\dagger(Q f) ,
	\end{align}
	and a similar formula for annihilation operators. Consequently, for any state $\Gamma$ on $\gF$ and any orthogonal
	projector $P$, we define its localization $\Gamma_P$ as a state on $\gF(P\gH)$ by taking the partial trace over
	$\gF(Q\gH)$:
	\begin{align}\label{def:P-localization}
		\Gamma_P:= \tr_{\gF(Q\gH)}\left[ \cU \Gamma \cU^* \right],
	\end{align}
	or, equivalently,
	\begin{align}\label{equiv:P-localization}
		\tr_{\gF(P\gH)}\left[ A \Gamma_P  \right] =	\tr_{\gF(P\gH)\otimes \gF(Q\gH)}\left[ (A \otimes \mathds{1}_{\gF(Q\gH)}) \cU \Gamma \cU^*  \right] ,
	\end{align}
	for any bounded operator $A$ on $\gF(P\gH)$.	The density matrices of $\Gamma_P$ can be shown to satisfy
	\begin{align*}
		(\Gamma_P)^{(k)}=P^{\otimes k}\Gamma^{(k)}P^{\otimes k},\qquad \forall k\geq1.
	\end{align*}
	We recall that a coherent state is a Weyl-rotation of the vacuum $|0\rangle$  in the Fock space $\gF$:
	\begin{align*}
		\xi(u):=W(u) |0\rangle := \exp(a^\dagger(u)-a(u))  |0\rangle = e^{-\|u\|^2/2} \exp\left(a^\dagger(u)\right)  |0\rangle = e^{-\|u\|^2/2} \bigoplus_{n \ge 0} \frac{1}{\sqrt{n!}} u^{\otimes n},
	\end{align*}
	where the Weyl operator $W(u)$ is a unitary operator translating creation and annihilation operators
	\begin{equation}
		\label{eq:Weyl-action}
		W(f)^* a^\dagger(g) W(f)=a^\dagger(g) + \langle f,g \rangle, \quad W(f)^* a(g) W(f)=a(g) + \langle g,f \rangle.
	\end{equation}
	We have the resolution of the identity (see, e.g.~\cite[(6.3)]{LNR15})
	\begin{equation}
		\pi^{-\Tr(P)} \int_{P\gH} |\xi(u)\rangle \langle \xi (u)| \dif u = \pi^{-\Tr(P)} \left( \int_{P\gH} e^{-\|u\|^2} \dif u \right)   \mathds{1}_{\gF(P\gH)} =  \mathds{1}_{\gF(P\gH)}.
		\label{eq:resolution_coherent2}
	\end{equation}
	\noindent
	We next recall the lower symbols and a quantitative bound from \cite[Lemma 6.2 and Remark 6.4]{LNR15}: 
	\begin{lemma}\label{Quantitative quantum de Finetti}
		For any state $\Gamma$ on $\gF$ and any scale $\varsigma>0$, we define the \emph{lower symbol} (or {\em Husimi function}) of $\Gamma$ on $P\gH$ at scale $\varsigma$ by
		\begin{equation} \label{eq:Husimi}
			\dif \mu_{P,\Gamma}^{\varsigma}(u):=(\varsigma \pi)^{-\Tr(P)}\langle \xi(u/\sqrt{\varsigma}),\Gamma_P \xi(u/\sqrt{\varsigma})\rangle_{\gF(P\gH)} \dif u.
		\end{equation}
		Here $\dif u$ is the usual Lebesgue measure on $P\gH \simeq \mathbb{C}^{\Tr(P)}$.  
		We have, for all $k\in\N$,
		\begin{equation}\label{eq:Chiribella}
			\int_{P\gH}|u^{\otimes k}\rangle \langle u^{\otimes k}| \, \dif \mu^\varsigma_{P,\Gamma}(u) = k!\varsigma^k\Gamma^{(k)}_P + k! \varsigma ^k \sum_{\ell = 0} ^{k-1} {k \choose \ell} \Gamma_P^{(\ell)} \otimes_s \mathds{1}_{\otimes_s ^{k-\ell} P\gH}.
		\end{equation}
		Thus, with $J=\Tr [P]$, 
		\begin{align}\label{de finetti estimate}
			\Tr \left| k!\varsigma^k\Gamma^{(k)}_P-\int_{P\gH}|u^{\otimes k}\rangle \langle  u^{\otimes k}| \, \dif \mu^\varsigma_{P,\Gamma}(u) \right| \leq \varsigma^k \sum_{\ell=0}^{k-1}{k\choose \ell}^2  \frac{(k-\ell +J-1)!}{(J-1)!}\tr_{\gF(P\gH)} \left[ \cN_P^{\ell}\Gamma_P\right]. 
		\end{align}
	\end{lemma}
	\noindent
	Finally, we recall a Berezin-Lieb type inequality from \cite[Theorem 7.1]{LNR15}, which links the relative entropy of two quantum states to the classical entropy of their lower symbols.
	\begin{lemma}\label{thm:rel-entropy}
		Let $\Gamma$ and $\Gamma'$  be two states on the Fock space $\gF$.  Let $\mu_{P,\Gamma}^{\varsigma}$ and $\mu_{P,\Gamma'}^{\varsigma}$ be the corresponding lower symbol defined as in \eqref{eq:Husimi}. Then we have
		\begin{equation}
			\cH(\Gamma,\Gamma') \geq \cH(\Gamma_P,\Gamma'_P) \geq \cH_{\mathrm{cl}}(\mu_{P,\Gamma}^{\varsigma}, \mu_{P,\Gamma'}^{\varsigma}).
			\label{eq:Berezin-Lieb}
		\end{equation}
	\end{lemma}

	\subsection{Free energy upper bound}\label{sec:upper bound}
	In this subsection, we establish the upper bound on the localized free energy by applying the quantum variational principle with a suitable trial state. In contrast to \cite{LNR15, LNR18, LNR21, NZZ25}, the main new difficulty is that the cutoff $f_\eta(\cN/\tau)$ destroys the simple tensor-product structure separating low and high modes. Consequently, the standard factorized trial state is no longer directly compatible with the relative entropy term.
The overall strategy is as follows. First, we localize the interaction to $P\gH$ and introduce a trial state whose cutoff still depends on the full number operator. This is the key device that allows us to keep the entropy contribution under control. Second, we show that the error incurred by replacing the full interaction with its $P$-localized version, together with the corresponding error in the entropy, is negligible. This is the most delicate part of the argument and is carried out in Lemmas~\ref{errors:localization_interaction} and~\ref{lemma:relative entropy}. Third, after the problem has been reduced to a finite-dimensional variational problem, the semi-classical analysis, combined with the convergence of the truncated free Gibbs state established in Proposition~\ref{lemma:non-interacting} below, identifies the limiting classical quantity.
	\begin{proposition}\label{prop:upper bound}
		Let $w$ satisfy Assumption~\ref{assum:on:w}, and let $\cK >0$ be arbitrary. Let $\Lambda_e>0$ be sufficiently large, and let $\eta \in \big(0 ,\frac{1}{2}\cK^2  \big)$ and $\eps \in (0,1)$ satisfy $\eta^{-1} e^{M\eps^{-2}} \leq \Lambda_e^{\frac18}$, where $M=(2\|w\|^2_{L^\infty}+1)(\cK^2+4)^3$. Define $P:=\mathds{1}(h\leq \Lambda_e)$ to be the orthogonal projection on $ \gH$. With the three-body operator $W^\eps$ defined in \eqref{def:operator W}, we introduce the projected interaction potential
		\begin{align}\label{def:widetilde:WP}
			\cW^\eps_P(u)=\frac{1}{3!} \big\langle  P^{\otimes 3}  u^{\otimes 3}, W^\eps P^{\otimes 3} u^{\otimes 3}\big\rangle.
		\end{align}
		Let the cutoff $f_\eta$ satisfy Assumption~\ref{assum:cutoff} with cutoff parameter $\cK$. Denote by $\cZ^{f_\eta}_{\tau}$ and $\cZ^{f_\eta}_{\tau,0}$ the partition functions as in \eqref{def:interacting:gibbs_state} and \eqref{def:free:gibbs_state}, respectively, associated with the cutoff $f_\eta$. Then, for all sufficiently large $ \tau$ and $\Lambda_e$ satisfying $\Lambda_e \leq \tau^{\frac14}$, we have
		\begin{align}\label{ineq:upper bound}
			-\log  \frac{\cZ^{f_\eta}_{\tau}}{\cZ^{f_\eta}_{\tau,0}} \leq -\log \left( \int_{P \gH} e^{\cW_P^\eps(u) } \dif \mu^{f_\eta}_{0,P} (u)\right) + O\left( \Lambda_e^{-\frac18}\right).
		\end{align}
		Here, $\mu^{f_\eta}_{0,P}$ is the projected Gaussian measure with cutoff $f_\eta$ on $P\gH$, defined by
		\begin{align}\label{truncated Gaussian measure}
			\dif \mu^{f_\eta}_{0,P}(u):=\frac{1}{\int_{P \gH} f_\eta\left( \|Pu\|_{L^2}^2 \right)\dif  \mu_{0,P}(u) }f_\eta\left( \|Pu\|_{L^2}^2 \right) \dif \mu_{0,P}(u).
		\end{align}
	\end{proposition}
	\begin{proof}
For clarity, we organize the proof so that each type of error is estimated once and then reused throughout the argument. The proof is divided into four steps. Step~1 reduces the variational principle to the projected space. Step~2 controls the localization and entropy errors arising from the non-factorized trial state. Step~3 establishes the finite-dimensional semiclassical bound. Step~4 rewrites this projected estimate in the normalized Gaussian form stated in the proposition.
		
\textbf{Step 1: Reduction to a finite-dimensional estimate.}
		We start by defining the interacting Gibbs state on $\gF(P\gH)$:
		\begin{align}\label{interacting Gibbs in trial}
		\Gamma^{f_\eta}_{\tau,P}:=\frac{e^{-\bbH_{\tau,P}}f_\eta\left(\cN_{P}/\tau\right)}{\tr_{\gF(P\gH)}\Big(e^{-\bbH_{\tau,P}}f_\eta\left(\cN_{P}/\tau\right)\Big)},
		\end{align}
		with $\cN_{P}= \dif \Gamma(P)$ and
		\begin{align}\label{regularized Hamiltonian}
			\bbH_{\tau,P}:=\bbH_{\tau,0,P}-\bbW_{\tau,P}:=\frac{1}{\tau}\dif \Gamma(PhP)-\frac{1}{\tau^3} \mathbb{W}_{P}.
		\end{align}
		Here, $\mathbb{W}_P$ is the second quantization of the operator $W_P^\eps:=P^{\otimes 3} W^\eps P^{\otimes 3}$. With the same conventions, we define the free Gibbs state on $\gF(P\gH)$ as
		\begin{align}\label{non-interacting Gibbs in trial}
			\Gamma^{f_\eta}_{\tau,0,P}:=\frac{e^{-\bbH_{\tau,0,P}}f_\eta\left(\cN_{P}/\tau \right)}{\tr_{\gF(P\gH)} \Big(e^{-\bbH_{\tau,0,P}}f_\eta\left(\cN_{P}/\tau\right)\Big)},\qquad 
			\Gamma_{\tau,0,P}:=\frac{e^{-\bbH_{\tau,0,P}}}{\tr_{\gF(P\gH)} \Big(e^{-\bbH_{\tau,0,P}}\Big)} .
		\end{align}
		Similarly, we define $	\Gamma^{f_\eta}_{\tau,0,Q}$ and $	\Gamma_{\tau,0,Q}$
		by replacing $P$ with $Q=\mathds{1}-P$ in the above definitions. Now we explain the choice of the trial state, which is a key difference from the previous works \cite{LNR15, LNR18, LNR21}.
		Following the strategy in those papers, one would naturally consider the trial state $$\Gamma_{\mathrm{ts}}= \cU^* (	 \Gamma^{f_\eta}_{\tau,P}\otimes 	\Gamma^{f_\eta}_{\tau,0,Q} ) \cU.$$ 
		However, in the present setting, the truncated free Gibbs state $\Gamma^{f_\eta}_{\tau,0}$ does not factorize. After applying the variational principle (see \eqref{upper bound to finite dimension} below), the term $\tr\big[\Gamma_{\mathrm{ts}}(\log f_\eta(\cN/\tau) -\log (f_\eta(\cN_P/\tau) \otimes f_\eta(\cN_Q/\tau)  ) )\big]$ appears in the discrepancy between the relative entropy $\cH(\Gamma_{\mathrm{ts}}, \Gamma^{f_\eta}_{\tau,0})$ and its $P$-localized counterpart. This term is difficult to control directly. We therefore introduce the following trial state on $\gF(\gH)$, with a cutoff given by the full number operator:
		\begin{align}\label{def:trial:state}
\Gamma_{\mathrm{ts}}:=\frac{e^{-\bbH_{\tau,0}+\cU^* (  \bbW_{\tau,P} \otimes \mathds{1}_{\gF(Q\gH)} ) \cU}f_\eta\left(\cN/\tau\right)}{\tr_{\gF(\gH)}\Big( e^{-\bbH_{\tau,0}+ \cU^* (  \bbW_{\tau,P} \otimes \mathds{1}_{\gF(Q\gH)} ) \cU}f_\eta\left(\cN/\tau\right) \Big)},
		\end{align}
		where the associated partition function is denoted by $\cZ_{\mathrm{ts}}$.
		Using the variational principle \eqref{variational principle quantum}, it follows that
		\begin{align}\label{upper bound to finite dimension}
			- &\log  \frac{\cZ^{f_\eta}_{\tau}}{\cZ^{f_\eta}_{\tau,0}}  \leq \cH\left( 	\Gamma_{\mathrm{ts}} ,	\Gamma^{f_\eta}_{\tau,0} \right) - \tr_{\gF(\gH)} \left( \bbW_\tau	\Gamma_{\mathrm{ts}} \right) 
			\\&= \cH\left(\Gamma^{f_\eta}_{\tau,P}, \Gamma^{f_\eta}_{\tau,0,P} \right) - \tr_{\gF(P\gH)} \left( \bbW_{\tau,P} \Gamma^{f_\eta}_{\tau,P} \right) - \tr_{\gF(\gH)} \left(  \bbW_{\tau}	\Gamma_{\mathrm{ts}} \right) + \tr_{\gF(P\gH)} \left(  \bbW_{\tau,P} \Gamma^{f_\eta}_{\tau,P} \right) \nonumber 
			\\ & \qquad   \qquad   \qquad  \qquad  \quad +\cH\left( 	\Gamma_{\mathrm{ts}},	\Gamma^{f_\eta}_{\tau,0} \right)- \cH\left( \Gamma^{f_\eta}_{\tau,P},  \Gamma^{f_\eta}_{\tau,0,P} \right) \nonumber
			\\ & = -\log \frac{\cZ^{f_\eta}_{\tau,P}}{\cZ^{f_\eta}_{\tau,0,P}} - \left[  \tr_{\gF(\gH)} \left(  \bbW_{\tau}	\Gamma_{\mathrm{ts}} \right) -\tr_{\gF(P\gH)} \left(  \bbW_{\tau,P} \Gamma^{f_\eta}_{\tau,P} \right) \right]
			+\left[\cH\left( 	\Gamma_{\mathrm{ts}},	\Gamma^{f_\eta}_{\tau,0} \right)- \cH\left(\Gamma^{f_\eta}_{\tau,P}, \Gamma^{f_\eta}_{\tau,0,P} \right) \right], \nonumber
		\end{align}
		where we denote $		\cZ^{f_\eta}_{\tau,P}=\tr_{\gF(P\gH)} \left( e^{-\bbH_{\tau,P}}f_\eta\left(\cN_{P}/\tau\right) \right) $ and $	\cZ^{f_\eta}_{\tau,0,P}=\tr_{\gF(P\gH)} \left( e^{-\bbH_{\tau,0,P}}f_\eta\left(\cN_{P}/\tau\right) \right)$. 
		Although the trial state $\Gamma_{\mathrm{ts}}$ does not factorize, we next show that the error terms in \eqref{upper bound to finite dimension} are sufficiently small.
		
\textbf{Step 2: Error estimates for the trial state.}
		First, using the characterization of second quantization \eqref{characterization:second:quan}, the unitary equivalence $\cU$ and \eqref{equiv:P-localization}, we have
		\begin{equation}\label{rewrite:first}
			\begin{aligned}
				\tr_{\gF(\gH)}   \left(\cU^*\big( \bbW_{\tau,P}	\otimes \mathds{1}_{\gF(Q\gH)} \big) \cU \Gamma_{\mathrm{ts}} \right)  &= 	\tr_{\gF(P\gH) \otimes \gF(Q\gH) }  \left(\big( \bbW_{\tau,P}	\otimes \mathds{1}_{\gF(Q\gH)} \big) \cU \Gamma_{\mathrm{ts}}\cU^* \right)  
				\\ &= \tr_{\gF(P\gH)  }  \left(  \bbW_{\tau,P}	\big( \Gamma_{\mathrm{ts}}\big)_P \right)  =\frac{1}{\tau^3} \tr \big(W_P^\eps \big((\Gamma_{\mathrm{ts}}\big)_P\big)^{(3)} \big) .
			\end{aligned}
		\end{equation}
		Then, we use \eqref{rewrite:first} to rewrite the second term on the right-hand side of \eqref{upper bound to finite dimension} as 
		\begin{equation}
			\begin{aligned}\label{error:terms_in_upper_bound}
				\tr_{\gF(\gH)} & \left(  \bbW_{\tau}	\Gamma_{\mathrm{ts}} \right) -\tr_{\gF(P\gH)}  \left(  \bbW_{\tau,P} \Gamma^{f_\eta}_{\tau,P} \right)
	\\ &=	\tr_{\gF(\gH)}  \left(  \left(\bbW_{\tau} -\cU^*\big( \bbW_{\tau,P}	\otimes \mathds{1}_{\gF(Q\gH)} \big) \cU \right)	\Gamma_{\mathrm{ts}} \right)  
				\\ &\qquad  \qquad +  \tr_{\gF(\gH)} \left( \cU^* \big( \bbW_{\tau,P}	\otimes \mathds{1}_{\gF(Q\gH)}\big) \cU \Gamma_{\mathrm{ts}} \right) -\tr_{\gF(P\gH)} \left(\bbW_{\tau,P}	 \Gamma^{f_\eta}_{\tau,P}  \right) 
				\\&= \frac{1}{\tau^3} \tr \left( \big(W^\eps- W^\eps_P\big)\Gamma_{\mathrm{ts}}^{(3)} \right)  
			+\left[ \tr_{\gF(\gH)} \left( \cU^* \big( \bbW_{\tau,P}	\otimes \mathds{1}_{\gF(Q\gH)}\big) \cU \Gamma_{\mathrm{ts}} \right) -\tr_{\gF(P\gH)} \left(\bbW_{\tau,P}	 \Gamma^{f_\eta}_{\tau,P}  \right) \right].
			\end{aligned}
		\end{equation}
	The next lemma isolates the interaction part of the upper-bound error. Its content is that once the high-mode mass of the trial state is small, replacing $W^\eps$ by $W^\eps_P$ costs only a negligible amount. We now control the first term on the right-hand side. 
		\begin{lemma}\label{errors:localization_interaction}
			For all sufficiently large $\tau$ and $\Lambda_e$, and all $ \eta\in \big(0,\frac{1}{2}\cK^2\big)$, $\eps \in(0,1)$ satisfying $\eps^{-3} \leq \Lambda_e^{\frac18}  $, there exists a constant $C>0$ depending only on $w$ and $\cK$ such that
			\begin{align*}\label{errors:localization_interaction1}
				\frac{1}{\tau^3} \left|  \tr \left( \big( W^\eps_{P}- W^\eps \big) \Gamma_{\mathrm{ts}}^{(3)}\right)\right| \leq C \Lambda_e^{-\frac18}.
			\end{align*}
		\end{lemma}
		\begin{proof}
		The proof is divided into two parts. We first reduce the three-body localization errors to the one-body high-frequency mass of the trial state. We then control this high-mode mass by means of a free tail term together with a relative-entropy contribution.
			
			\textit{\underline{Step 1}: reduce the localization error to the high-mode one-body mass.} 	Observe that $W^\eps$ is self-adjoint and satisfies
			$$W^\eps- W^\eps_{P} = P^{\otimes 3} W^\eps  \big(	\mathds{1}^{\otimes 3} -P^{\otimes 3}\big) +  \big(	\mathds{1}^{\otimes 3} -P^{\otimes 3}\big)  W^\eps P^{\otimes 3}  + \big(	\mathds{1}^{\otimes 3} -P^{\otimes 3}\big)  W^\eps \big(	\mathds{1}^{\otimes 3} -P^{\otimes 3}\big).$$
			We begin by estimating the first two terms on the right-hand side.
			By the Cauchy-Schwarz inequality (see, e.g., \cite[Sec. 1]{Sim79}), for any $\varsigma>0$, we have
			\begin{equation}\label{CS:ineq}
				\begin{aligned}
					& \pm \big( P^{\otimes 3} W^\eps  (	\mathds{1}^{\otimes 3} -P^{\otimes 3})  + (	\mathds{1}^{\otimes 3} -P^{\otimes 3}) W^\eps P^{\otimes 3}   \big)
					\\ & \quad \leq  \varsigma  P^{\otimes 3}|W^\eps| P^{\otimes 3}  +\frac{1}{\varsigma} (	\mathds{1}^{\otimes 3} -P^{\otimes 3}) |W^\eps| (	\mathds{1}^{\otimes 3} -P^{\otimes 3})  ,
				\end{aligned}
			\end{equation}
			where $|W^\eps|:= \sqrt{(W^\eps)^*W^\eps}$. Since $\Gamma_{\mathrm{ts}}^{(3)} \geq 0$, combining this with \eqref{CS:ineq} yields, for any $\varsigma>0$,
			\begin{equation}\label{CS:ineq2}
				\begin{aligned}
					\frac{1}{\tau^3} & \left|  \tr \left(\left(P^{\otimes 3} W^\eps  (	\mathds{1}^{\otimes 3} -P^{\otimes 3})  + (	\mathds{1}^{\otimes 3} -P^{\otimes 3}) W^\eps P^{\otimes 3}   \right) \Gamma_{\mathrm{ts}}^{(3)}\right)\right| 
					\\ & \leq \frac{\varsigma}{\tau^3} \tr \left( P^{\otimes 3}|W^\eps| P^{\otimes 3}  \Gamma_{\mathrm{ts}}^{(3)} \right) +  \frac{1}{ \varsigma \tau^3}  \tr \left((	\mathds{1}^{\otimes 3} -P^{\otimes 3}) |W^\eps| (	\mathds{1}^{\otimes 3} -P^{\otimes 3}) \Gamma_{\mathrm{ts}}^{(3)} \right)
					\\ &=\frac{\varsigma}{\tau^3} \tr \left( |W^\eps| P^{\otimes 3}  \Gamma_{\mathrm{ts}}^{(3)} P^{\otimes 3} \right) +  \frac{1}{ \varsigma \tau^3}  \tr \left(|W^\eps| (	\mathds{1}^{\otimes 3} -P^{\otimes 3}) \Gamma_{\mathrm{ts}}^{(3)} (	\mathds{1}^{\otimes 3} -P^{\otimes 3})  \right) ,
				\end{aligned}
			\end{equation}
			where we used the cyclicity of the trace in the last line.
			Optimizing over $\varsigma>0$, we choose $$\varsigma =  \sqrt{\tr \big(|W^\eps| (	\mathds{1}^{\otimes 3} -P^{\otimes 3}) \Gamma_{\mathrm{ts}}^{(3)}(	\mathds{1}^{\otimes 3} -P^{\otimes 3})  \big)} \sqrt{ \tr \big( |W^\eps| P^{\otimes 3}  \Gamma_{\mathrm{ts}}^{(3)} P^{\otimes 3} \big)^{-1}},$$ 
			and substituting into \eqref{CS:ineq2}, we obtain
				\begin{align}\label{CS:ineq3}
					\frac{1}{\tau^3} & \left|  \tr \big(\left(P^{\otimes 3} W^\eps  (	\mathds{1}^{\otimes 3} -P^{\otimes 3})  + (	\mathds{1}^{\otimes 3} -P^{\otimes 3}) W^\eps P^{\otimes 3}   \right)  \nonumber \Gamma_{\mathrm{ts}}^{(3)}\big)\right| 
					\\ & \leq  \frac{1}{\tau^3} \sqrt{\tr \big( |W^\eps| P^{\otimes 3}  \Gamma_{\mathrm{ts}}^{(3)} P^{\otimes 3} \big)} \sqrt{\tr \big( |W^\eps| (	\mathds{1}^{\otimes 3} -P^{\otimes 3}) \Gamma_{\mathrm{ts}}^{(3)} (	\mathds{1}^{\otimes 3} -P^{\otimes 3}) \big)}  \nonumber
					\\ & \leq \frac{\eps^{-2}\|w\|^2_{L^\infty}}{\tau^3}  \sqrt{ \tr \big(  \Gamma_{\mathrm{ts}}^{(3)}  \big)} \sqrt{\tr \big( (	\mathds{1}^{\otimes 3} -P^{\otimes 3}) \Gamma_{\mathrm{ts}}^{(3)}  \big)}
					\\ & \leq  \eps^{-2}\|w\|^2_{L^\infty}   \sqrt{  \tr_{\gF(\gH)}  \big( (\cN/\tau)^3 \Gamma_{\mathrm{ts}}  \big)}
					\sqrt{  \frac{1}{\tau^3} \tr \big( \big(	\mathds{1}^{\otimes 3} -P^{\otimes 3}\big) \Gamma_{\mathrm{ts}}^{(3)}  \big)}  \nonumber
					\\ & \leq  \eps^{-2}\|w\|^2_{L^\infty}  \cK^3   \sqrt{  \frac{1}{\tau^3} \tr \big( \big(	\mathds{1}^{\otimes 3} -P^{\otimes 3}\big) \Gamma_{\mathrm{ts}}^{(3)}  \big)}, \nonumber
				\end{align}
			where in the second inequality we used $0\leq P\leq \mathds{1}_{\gH} $ and the bound $\big\| |W^\eps | \big\|_{\mathrm{op}} \leq \eps^{-2}\|w\|^2_{L^\infty} $. Similarly, we also obtain
			\begin{align}\label{CS:ineq4}
				\frac{1}{\tau^3} & \left|  \tr \left( (	\mathds{1}^{\otimes 3} -P^{\otimes 3}) W^\eps (	\mathds{1}^{\otimes 3} -P^{\otimes 3})  \Gamma_{\mathrm{ts}}^{(3)}\right)\right| \leq \frac{1}{\tau^3} \eps^{-2}\|w\|^2_{L^\infty}  \tr \left( (	\mathds{1}^{\otimes 3} -P^{\otimes 3}) \Gamma_{\mathrm{ts}}^{(3)}  \right).
			\end{align}
			
			Next, we estimate $\tr \big( (	\mathds{1}^{\otimes 3} -P^{\otimes 3}) \Gamma_{\mathrm{ts}}^{(3)}  \big)$. Using the elementary inequality $\mathds{1}^{\otimes 3}-P^{\otimes 3}\leq \sum_{j=1}^3 \mathds{1}^{\otimes(j-1)}\otimes Q\otimes \mathds{1}^{\otimes(3-j)}$, we find that
			\begin{equation}\label{WP-WGammatau}
				\begin{aligned}
					\frac{1}{\tau^3}  \tr \left( \big(	\mathds{1}^{\otimes 3} -P^{\otimes 3}\big) \Gamma_{\mathrm{ts}}^{(3)}  \right) &\leq     	\frac{1}{\tau^3} \tr \left( (Q \otimes \mathds{1} \otimes  \mathds{1} + \mathds{1}   \otimes Q\otimes  \mathds{1}+ \mathds{1} \otimes  \mathds{1} \otimes Q) \Gamma_{\mathrm{ts}}^{(3)}\right) 
					\\&	\leq \frac{C}{\tau} \tr_{\gF(\gH)} \left( \cU^* \big(\mathds{1}_{\gF(P\gH)} \otimes  \cN_Q \big) \cU \left( \frac{\cN}{\tau} \right)^2 \Gamma_{\mathrm{ts}}  \right)
					\\ &\leq \frac{C}{\tau}  \cK^4   \tr_{\gF(\gH)}  \left( \cU^* \big(\mathds{1}_{\gF(P\gH)} \otimes  \cN_Q \big) \cU \Gamma_{\mathrm{ts}}  \right)
					\\ &=  \frac{ C }{\tau}  \cK^4 \tr \left( Q \Gamma_{\mathrm{ts}}^{(1)}  \right) 
					\leq \frac{C}{\tau}  \left| \tr \left( Q \big(\Gamma_{\mathrm{ts}}^{(1)} - \Gamma_{\tau,0}^{(1)} \big) \right) \right|+ \frac{C}{\tau}   \tr \left( Q\Gamma_{\tau,0}^{(1)} \right) ,
				\end{aligned}
			\end{equation}
			where in the third inequality we used that $\cU^* (\mathds{1}_{\gF(P\gH)} \otimes  \cN_Q ) \cU \geq0$ and that it commutes with $\cN$. 
			Collecting the bounds \eqref{CS:ineq3}--\eqref{WP-WGammatau}, we conclude that
			\begin{equation}\label{Wp-W:main}
				\begin{aligned}
					\frac{1}{\tau^3}  \left|  \tr \left( \left( W^\eps_{P}- W^\eps \right) \Gamma_{\mathrm{ts}}^{(3)}\right)\right|
				& \leq  C \eps^{-2}\|w\|^2_{L^\infty}  \cK^3  \left( \frac{1}{\tau}  \left| \tr \left( Q \big(\Gamma_{\mathrm{ts}}^{(1)} - \Gamma_{\tau,0}^{(1)} \big) \right) \right|+ \frac{1}{\tau}   \tr \left( Q\Gamma_{\tau,0}^{(1)} \right) \right)^{\frac12}
					\\&  \quad + C \eps^{-2}\|w\|^2_{L^\infty}   \left( \frac{1}{\tau}  \left| \tr \left( Q \big(\Gamma_{\mathrm{ts}}^{(1)} - \Gamma_{\tau,0}^{(1)} \big) \right) \right|+ \frac{1}{\tau}   \tr \left( Q\Gamma_{\tau,0}^{(1)} \right) \right).
				\end{aligned}
			\end{equation}
			\textit{\underline{Step 2:} control the high-mode one-body mass.}
			At this point, the problem has been reduced to the free high-frequency tail and to the deviation of the trial state's one-body density matrix from the free one. The first is deterministic, while the second is handled through the relative entropy with respect to $\Gamma_{\tau,0}$.
			The second term on the right-hand side of \eqref{WP-WGammatau} can be estimated as
			\begin{align}\label{Wp-W:second}
				\frac{1}{\tau}\tr \left( Q\Gamma_{\tau,0}^{(1)}  \right)= \tr \left( \frac{Q/\tau}{e^{h/\tau}-1} \right) \leq \tr(Qh^{-1})\leq C \Lambda_e^{-\frac12}  ,
			\end{align}
			for any $\tau >0$, where $C>0$ is independent of $\tau$. 
			
			Turning to the first term on the right-hand side of \eqref{WP-WGammatau}, we apply H\"older's inequality in Schatten space (see, e.g., \cite[Theorem 2.8]{Sim79}) together with \cite[Theorem 6.1]{LNR21} to obtain
			\begin{equation}\label{control_relative_entropy:1}
				\begin{aligned}
					\frac{1}{\tau}	\left| \tr \left( Q \big(\Gamma_{\mathrm{ts}}^{(1)} - \Gamma_{\tau,0}^{(1)} \big) \right) \right| & = \frac{1}{\tau} \left| \tr \left( Qh^{-1} \sqrt{h} \big(\Gamma_{\mathrm{ts}}^{(1)} - \Gamma_{\tau,0}^{(1)} \big) \sqrt{h} \right) \right| 
					\\ &\leq  \frac{1}{\tau} \|Qh^{-1}\|_{\mathfrak{S}^2}\left\| \sqrt{h} \big(\Gamma_{\mathrm{ts}}^{(1)} - \Gamma_{\tau,0}^{(1)} \big) \sqrt{h}  \right\|_{\mathfrak{S}^2} 
					\\ &\leq 2\|Qh^{-1}\|_{\mathfrak{S}^2}  \sqrt{ \cH\big(\Gamma_{\mathrm{ts}},\Gamma_{\tau,0}\big)}\left(\sqrt{2}+\sqrt{ \cH\big(\Gamma_{\mathrm{ts}} ,\Gamma_{\tau,0}\big)}\right).
				\end{aligned}
			\end{equation}
			To control the relative entropy $	\cH \big( \Gamma_{\mathrm{ts}} , \Gamma_{\tau,0} \big)$, we first apply the variational principle \eqref{variational principle quantum} to obtain
			\begin{align}	\label{control_relative_entropy:2}
				& \cH  \big( \Gamma_{\mathrm{ts}} , \Gamma_{\tau,0} \big) = -\log\left( \frac{\cZ_{\mathrm{ts}}}{\cZ_{\tau,0}}\right) + \tr_{\gF(\gH)} \left( \cU^* \big(  \bbW_{\tau,P} \otimes \mathds{1}_{\gF(Q\gH)} \big) \cU \Gamma_{\mathrm{ts}}\right) + \tr_{\gF(\gH)}  \big( \Gamma_{\mathrm{ts}} \log f_\eta(\cN/\tau) \big) \nonumber
				\\ & \quad \leq -\log\left( \frac{\cZ_{\mathrm{ts}}}{\cZ^{f_\eta}_{\tau,0}}\right) -\log\left( \frac{\cZ^{f_\eta}_{\tau,0}}{\cZ_{\tau,0}}\right) 
				+\tr_{\gF(\gH)} \left( \cU^* \big(  \bbW_{\tau,P} \otimes \mathds{1}_{\gF(Q\gH)} \big) \cU \Gamma_{\mathrm{ts}}\right) 
				\\ &\quad \leq -\tr_{\gF(\gH)} \left(\cU^* \big(  \bbW_{\tau,P} \otimes \mathds{1}_{\gF(Q\gH)} \big) \cU \Gamma^{f_\eta}_{\tau,0}\right)-\log\left( \frac{\cZ^{f_\eta}_{\tau,0}}{\cZ_{\tau,0}}\right)+ \tr_{\gF(\gH)} \left( \cU^* \big(  \bbW_{\tau,P} \otimes \mathds{1}_{\gF(Q\gH)} \big) \cU \Gamma_{\mathrm{ts}}\right), \nonumber
			\end{align}
			where we used $\tr_{\gF(\gH)}  \big( \Gamma_{\mathrm{ts}} \log f_\eta(\cN/\tau) \big)\leq 0$ in the second line.
			We next estimate the three terms on the right-hand side of \eqref{control_relative_entropy:2} separately.
			Recall the following bounds for the second quantization $\mathbb{W}$, proved in \cite[(3.91)]{Kno09}:
			\begin{align*}
				\frac{1}{\tau^3} \big\| \mathbb{W} \big|_{\gH^{(n)}}  \big\|\leq \left(\frac{n}{\tau}\right)^3 \eps^{-2} \|w\|_{L^\infty}^2,
			\end{align*}
			which in particular imply
			\begin{align}\label{bound:number operator}
				\left\|  	\frac{1}{\tau^3} \mathbb{W} \, \mathds{1}_{\{\cN/\tau \leq \cK^2 \}} \right\| \leq  \eps^{-2} \|w\|_{L^\infty}^2\cK^6.
			\end{align}
			By property \eqref{property:unitary}, we have $\cU \cN \cU^*=\cN_P  \otimes \mathds{1}_{\gF (Q\gH)} +  \mathds{1}_{\gF (P \gH)} \otimes \cN_Q $. Combining this with the bound \eqref{bound:number operator}, the first term on the right-hand side of \eqref{control_relative_entropy:2} satisfies
			\begin{align}\label{control_relative_entropy:22}
				& \left|  \tr_{\gF(\gH)} \left(\cU^* \big(  \bbW_{\tau,P} \otimes \mathds{1}_{\gF(Q\gH)} \big) \cU \Gamma^{f_\eta}_{\tau,0}\right)\right|  \nonumber
				\\ & \qquad = 	\left| 	\tr_{\gF(\gH)}  \left( \cU^* \big(  \bbW_{\tau,P} \otimes \mathds{1}_{\gF(Q\gH)} \big) \cU \mathds{1}_{\{\cN/\tau \leq \cK^2 \}} \Gamma^{f_\eta}_{\tau,0}   \right) \right| \nonumber
				\\ & \qquad  \leq \left\| \cU^* \big(  \bbW_{\tau,P} \otimes \mathds{1}_{\gF(Q\gH)} \big) \cU  \mathds{1}_{\{\cU^* ( \cN_P  \otimes \mathds{1}_{\gF (Q\gH)} +  \mathds{1}_{\gF (P \gH)} \otimes \cN_Q ) \cU \leq \cK^2\tau \}} \right\|  \nonumber
				\\ & \qquad  \leq  \left\|  \frac{1}{\tau^3} \mathbb{W}_P \otimes \mathds{1}_{\gF(Q\gH)} \mathds{1}_{\{ \cN_P  \otimes \mathds{1}_{\gF (Q\gH)} +  \mathds{1}_{\gF (P \gH)} \otimes \cN_Q  \leq \cK^2\tau \}} \right\|   
			\leq  \eps^{-2} \|w\|_{L^\infty}^2\cK^6. 
			\end{align}
			The third term on the right-hand side of \eqref{control_relative_entropy:2} is estimated in the same way and satisfies the same bound. It remains to control the second term. By  \eqref{convergence:for:feta}, for every $\eta$ satisfying $\frac{\cK^2}{2}>\eta \geq \tau^{-\frac{1}{64}}$, we have
			\begin{equation}\label{limit:on:freegibbs}
				\begin{aligned}
					\lim_{\tau \to \infty}  \tr_{\gF(\gH)} \left( \Gamma_{\tau,0} f_\eta \left(\frac{\cN}{\tau}\right) \right)=\int_{\gH}  f_\eta(\|u\|^2_{L^2}) \dif \mu_0(u)
				> \frac12 \int_{\gH} \mathds{1}_{\big\{ \|u\|_{L^2}^2 \leq \frac{\cK^2}{4}\big\}} \dif \mu_0(u)=: C(\cK).
				\end{aligned}
			\end{equation}			
			Consequently, there exists $T_0>0$ such that, for all $\tau >T_0$ and $ \eta\in \big[\tau^{-\frac{1}{64}},\frac{1}{2}\cK^2\big)$,
			\begin{align}\label{upper_bound_Ztau0}
				-\log\left[ \tr_{\gF(\gH)} \left( \Gamma_{\tau,0} f_\eta \left(\frac{\cN}{\tau}\right) \right)  \right] \leq -\log\left( \frac12 \int_{\gH} \mathds{1}_{\big\{ \|u\|_{L^2}^2 \leq \frac{\cK^2}{4}\big\}} \dif \mu_0(u) \right) =\log\left(\frac{1}{C(\cK)}\right). 
			\end{align}
			Inserting the bounds \eqref{control_relative_entropy:2}, \eqref{control_relative_entropy:22} and \eqref{upper_bound_Ztau0} into \eqref{control_relative_entropy:1}, we obtain, for any $\tau> T_0$,
				\begin{align}\label{control_relative_entropy:11}
					\frac{1}{\tau}\left| \tr \left( Q \big(\Gamma_{\mathrm{ts}}^{(1)} - \Gamma_{\tau,0}^{(1)} \big) \right) \right| & \leq 2\left(\sqrt{2}+\sqrt{\log\big(1/C(\cK)\big)+2\eps^{-2} \|w\|_{L^\infty}^2\cK^6}  \right)^2 \|Qh^{-1}\|_{\mathfrak{S}^2} \nonumber
					\\& \leq C(w,\cK) \eps^{-2} \|Qh^{-1}\|_{\mathfrak{S}^2} \lesssim_{w,\cK}  \eps^{-2} \Lambda_e^{-\frac12}.
				\end{align}
			
			Once \eqref{Wp-W:second} and \eqref{control_relative_entropy:11} are in hand, the claimed localization estimate follows directly from \eqref{Wp-W:main}.
			Substituting \eqref{Wp-W:second} and \eqref{control_relative_entropy:11} into \eqref{Wp-W:main}, we obtain
			\begin{align*}
				\frac{1}{\tau^3} &	 \left|  \tr \left( \left( W^\eps_{P}- W^\eps \right) \Gamma_{\mathrm{ts}}^{(3)}\right)\right|
			\lesssim_{w,\cK}  \eps^{-2} (\eps^{-2} \Lambda_e^{-\frac12}+   \Lambda_e^{-\frac12} )^{\frac12}+  \eps^{-2} ( \eps^{-2} \Lambda_e^{-\frac12}+   \Lambda_e^{-\frac12} ) \lesssim_{w,\cK} \Lambda_e^{-\frac18},
			\end{align*}
			where the last inequality follows from $\eps^{-3} \leq  \Lambda_e^{\frac18}$. This completes the proof of Lemma~\ref{errors:localization_interaction}. 
		\end{proof}
		
   Next, we control the second term on the right-hand side of \eqref{error:terms_in_upper_bound} as well as the difference of the relative entropies in the last term of \eqref{upper bound to finite dimension}. The purpose is to recover, from the full quantum trial state, the entropy term that naturally appears in the finite-dimensional classical variational problem. The difficulty lies precisely in the lack of factorization induced by the cutoff: after localization, the $P$- and $Q$-modes are still coupled through $f_\eta(\cN/\tau)$, hence \cite[Lemma 10.3]{LNR21} is not applicable. We therefore rewrite the entropy difference into interaction contributions and logarithms of truncated partition functions, and then bound these pieces separately using the localization estimates already obtained, together with a uniform positive lower bound on the truncated free partition function.
		\begin{lemma}\label{lemma:relative entropy}
			For all sufficiently large $\tau$ and $\Lambda_e$, and all $ \eta\in \big(0,\frac{1}{2}\cK^2\big)$, $\eps \in(0,1)$ satisfying $\eta^{-1}e^{M\eps^{-2}} \leq \Lambda_e^{\frac18}  \leq  \tau^{\frac{1}{32}}$, there exists a constant $C>0$ depending only on $w$ and $\cK$ such that
			\begin{align*}
				\left| \cH(\Gamma_{\mathrm{ts}}, \Gamma^{f_\eta}_{\tau,0}) -\cH( \Gamma^{f_\eta}_{\tau,P}, \Gamma^{f_\eta}_{\tau,0,P})\right| \leq  C \Lambda_e^{-\frac{1}{4}}.
			\end{align*}
		\end{lemma}
		\begin{proof}
For notational convenience, set
			\begin{align*}
				\bbW_{\tau,P,\cU}&:=\cU^* \big( \bbW_{\tau,P}	\otimes \mathds{1}_{\gF(Q\gH)}\big) \cU, 
				\\
				\cN_{P,\cU} &:= \cU^* \big( \cN_{P}	\otimes \mathds{1}_{\gF(Q\gH)}\big) \cU,  &	\cN_{Q,\cU} &:= \cU^* \big( \mathds{1}_{\gF(P \gH)} \otimes \cN_{Q}	 \big) \cU,
				\\ 
				\textbf{A}_{\tau,P}&:=\tr_{\gF(P\gH)}\Big(e^{-\bbH_{\tau,P}}f_\eta(\cN_P/\tau)\bbW_{\tau,P}\Big),
				& 
				\textbf{A}_\tau&:=\cZ_{\tau,0,Q}^{-1}\tr_{\gF(\gH)}\Big(e^{-\bbH_{\tau,0}+\bbW_{\tau,P,\cU}}f_\eta(\cN/\tau)\bbW_{\tau,P,\cU}\Big),
				\\
				\textbf{B}_{\tau,P}	&:= 
				\tr_{\gF(P\gH)}\Big(e^{-\bbH_{\tau,P}}f_\eta(\cN_P/\tau)\Big),
				& \textbf{B}_\tau&:=\cZ_{\tau,0,Q}^{-1}\tr_{\gF(\gH)}\Big(e^{-\bbH_{\tau,0}+\bbW_{\tau,P,\cU}}f_\eta(\cN/\tau)\Big),
				\\
				\textbf{C}_{\tau,P}&:=\tr_{\gF(P\gH)}\Big(e^{-\bbH_{\tau,0,P}}f_\eta(\cN_P/\tau)\Big),
				& 	\textbf{C}_\tau&:=\cZ_{\tau,0,Q}^{-1}\tr_{\gF(\gH)}\Big(e^{-\bbH_{\tau,0}}f_\eta(\cN/\tau)\Big).
			\end{align*}
			Then a direct computation gives
			\begin{align}\label{difference:relative entropy}
				\tr_{\gF(\gH)} \left( \cU^* \big( \bbW_{\tau,P}	\otimes \mathds{1}_{\gF(Q\gH)}\big) \cU \Gamma_{\mathrm{ts}} \right) -\tr_{\gF(P\gH)} \left(\bbW_{\tau,P}	 \Gamma^{f_\eta}_{\tau,P}  \right)  &= \frac{ \textbf{A}_\tau}{ \textbf{B}_\tau} -\frac{ \textbf{A}_{\tau,P}}{ \textbf{B}_{\tau,P}} , \nonumber
				\\
				\cH (\Gamma_{\mathrm{ts}}, \Gamma^{f_\eta}_{\tau,0}) -\cH(\Gamma^{f_\eta}_{\tau,P}, \Gamma^{f_\eta}_{\tau,0,P})
				=\left(\frac{ \textbf{A}_\tau}{ \textbf{B}_\tau} -  \frac{ \textbf{A}_{\tau,P}}{ \textbf{B}_{\tau,P}} \right) &-\log\!\left(\frac{ \textbf{B}_\tau}{ \textbf{B}_{\tau,P}}\right)+\log\!\left(\frac{ \textbf{C}_\tau}{ \textbf{C}_{\tau,P}}\right).
			\end{align}
			We then estimate the three terms on the right-hand side of \eqref{difference:relative entropy} separately. The first step compares the interacting traces before and after localization. The second step provides uniform lower bounds for the relevant denominators. The third step inserts these estimates into the entropy identity and converts them into the claimed bound.
		
			\textit{\underline{Step 1:} comparison of the interacting partition functions $\left| \textbf{A}_{\tau}- \textbf{A}_{\tau,P}  \right| $ and $\left| \textbf{B}_{\tau}- \textbf{B}_{\tau,P}  \right| $.} Since $ \mathds{1}_{\gF(P\gH)} \otimes \Gamma_{\tau,0,Q}$ commutes with $e^{-\bbH_{\tau,P}}f_\eta\big(\cN_P/\tau\big) \bbW_{\tau,P} \otimes \mathds{1}_{\gF(Q\gH)} $, we first use the unitary equivalence \eqref{def:unitary} to rewrite $\textbf{A}_{\tau,P}$ as
			\begin{align*}
				&	\tr_{\gF(P\gH)}  \Big( e^{-\bbH_{\tau,P}}  f_\eta\big(\cN_P/\tau\big) \bbW_{\tau,P} \Big)
				\\ & =	\tr_{\gF(P\gH)\otimes \gF(Q\gH)}\Big( e^{-\bbH_{\tau,P}}  f_\eta\big(\cN_P/\tau\big) \bbW_{\tau,P} \otimes \Gamma_{\tau,0,Q} \Big) 
				\\ &= \cZ_{\tau,0,Q}^{-1}\tr_{\gF(\gH)}\Big( \cU^*  \big( e^{-\bbH_{\tau,P}} \otimes \mathds{1}_{\gF(Q\gH)}  \big) \big(  \mathds{1}_{\gF(P\gH)} \otimes e^{-\bbH_{\tau,0,Q}}  \big) \big( f_\eta\big(\cN_P/\tau\big)  \otimes \mathds{1}_{\gF(Q\gH)}  \big) \big( \bbW_{\tau,P}  \otimes \mathds{1}_{\gF(Q\gH)}  \big) \cU   \Big)
				\\ & = \cZ_{\tau,0,Q}^{-1}\tr_{\gF(\gH)}\Big(  e^{-\cU^*( \bbH_{\tau,P}\otimes \mathds{1}_{\gF(Q\gH)} )\cU }    e^{-\cU^* (\mathds{1}_{\gF(P\gH)} \otimes \bbH_{\tau,0,Q} )\cU}   f_\eta\left( \cN_{P,\cU} /\tau\right) \bbW_{\tau,P,\cU} \Big)
				\\ &= \cZ_{\tau,0,Q}^{-1}\tr_{\gF(\gH)}\Big(  e^{-\cU^*( \bbH_{\tau,0,P}\otimes \mathds{1}_{\gF(Q\gH)} +\mathds{1}_{\gF(P\gH)} \otimes \bbH_{\tau,0,Q}-\bbW_{\tau,P}	\otimes \mathds{1}_{\gF(Q\gH)})\cU }  f_\eta\left( \cN_{P,\cU} /\tau\right) \bbW_{\tau,P,\cU}  \Big)
				\\ &= \cZ_{\tau,0,Q}^{-1}\tr_{\gF(\gH)}\Big(  e^{- \bbH_{\tau,0}+\bbW_{\tau,P,\cU} }  f_\eta\left( \cN_{P,\cU} /\tau\right) \bbW_{\tau,P,\cU}  \Big),
			\end{align*}
			where in the last equality we used $\cU^* \bbH_{\tau,0}\cU =\bbH_{\tau,0,P}\otimes \mathds{1}_{\gF(Q\gH)} +\mathds{1}_{\gF(P\gH)} \otimes \bbH_{\tau,0,Q}$. By the mean-value formula and $\mathds{1}_{\{\cN /\tau\leq \cK^2 \}}  \mathds{1}_{\{\cN_{P,\cU} /\tau\leq \cK^2 \}} = \mathds{1}_{\{\cN /\tau\leq \cK^2 \}}  $, we obtain
			\begin{align}\label{first part in Gamma_tau,P-Gamma_1}
			&	\cZ_{\tau,0,P}^{-1} \left| \textbf{A}_{\tau}- \textbf{A}_{\tau,P}  \right|  \nonumber
				\\ &=\cZ_{\tau,0,P}^{-1}   \cZ_{\tau,0,Q}^{-1} \left|  \tr_{\gF(\gH)} \Big(  e^{-\bbH_{\tau,0}+\bbW_{\tau,P,\cU} } \big( f_\eta\left(\cN/\tau\right)-f_\eta\left(\cN_{P,\cU} /\tau\right) \big) \bbW_{\tau,P,\cU} \Big) \right|    \nonumber
				\\ & = \cZ_{\tau,0,P}^{-1}   \cZ_{\tau,0,Q}^{-1} \left|  \tr_{\gF(\gH)} \Big(  e^{-\bbH_{\tau,0}+\bbW_{\tau,P,\cU} } \big( f_\eta\left(\cN/\tau\right)-f_\eta\left(\cN_{P,\cU} /\tau\right) \big) \mathds{1}_{\{\cN_{P,\cU} /\tau\leq \cK^2 \}} \bbW_{\tau,P,\cU} \Big) \right|   \nonumber
				\\				& =  \cZ_{\tau,0,P}^{-1}   \cZ_{\tau,0,Q}^{-1}\left| \tr_{\gF(\gH)}  \Big( e^{-\bbH_{\tau,0}+\bbW_{\tau,P,\cU} } \int_0^1 f_\eta'\Big( \frac{\cN_{P,\cU} +s  \cN_{Q,\cU} }{\tau}  \Big) \frac{\cN_{Q,\cU}}{\tau} \dif s \mathds{1}_{\{\cN_{P,\cU} /\tau\leq \cK^2 \}} \bbW_{\tau,P,\cU} \Big)  \right| . 
			\end{align}
			Since the operators $\mathds{1}_{\{\cN_{P,\cU} /\tau\leq \cK^2 \}} $, $e^{-\bbH_{\tau,0}+ \bbW_{\tau,P,\cU}}  $ and $  \cN_{Q,\cU}   $ are non-negative and mutually commute, their product is also non-negative. Invoking again \eqref{bound:number operator} and \eqref{control_relative_entropy:22}, we deduce that the right-hand side of \eqref{first part in Gamma_tau,P-Gamma_1} is bounded by
			\begin{align}\label{first part in Gamma_tau,P-Gamma_1.4}
				&\|f_\eta'\|_{L^\infty}	\cZ_{\tau,0,P}^{-1}\cZ_{\tau,0,Q}^{-1}  \tr_{\gF(\gH)}  \Big(  e^{-\bbH_{\tau,0}+\bbW_{\tau,P,\cU} }  \frac{\cN_{Q,\cU}}{\tau}  \mathds{1}_{\{\cN_{P,\cU}  /\tau \leq \cK^2 \}}   \Big) 
				  \left\|  \mathds{1}_{\{\cN_{P,\cU}  /\tau \leq \cK^2 \}}   \bbW_{\tau,P,\cU} \right\| 
				\\&	\leq \eta^{-1}\eps^{-2}\|w\|_{L^\infty}^2   \cK^6	\cZ_{\tau,0,P}^{-1} \cZ_{\tau,0,Q}^{-1}  \tr_{\gF(\gH)}  \left(  e^{-\bbH_{\tau,0}+\bbW_{\tau,P,\cU} } \frac{\cN_{Q,\cU}}{\tau} \mathds{1}_{\{ \cN_{P,\cU}/\tau  \leq \cK^2\}}  \right). \nonumber
			\end{align}
			Since $\cN_{P,\cU}$ commutes with $\bbH_{\tau,0}-\bbW_{\tau,P,\cU}$ and $\mathds{1}_{\{ \cN_{P,\cU} /\tau\leq \cK^2\}}$ is a projection, the following identity holds:
				\begin{align}\label{op:identity}
						\mathds{1}_{\{ \cN_{P,\cU}/ \tau \leq \cK^2\}} e^{-\bbH_{\tau,0}+\bbW_{\tau,P,\cU} }  =	\mathds{1}_{\{ \cN_{P,\cU}/ \tau  \leq \cK^2\}}  e^{-(\bbH_{\tau,0}-\bbW_{\tau,P,\cU})		\mathds{1}_{\{\cN_{P,\cU}/ \tau  \leq \cK^2 \}}  }.  
			\end{align}
			We recall the Golden-Thompson inequality (see, e.g., \cite[Sec. 8.1]{Sim79}): for self-adjoint operators $X,Y,Z$ satisfying $Z\geq 0$ and $[X,Z]=[Y,Z]=0$, one has
			\begin{align}\label{trace:ineq}
				\tr\big(Ze^{X+Y}\big) \leq 	\tr\big(Ze^Xe^Y\big).
			\end{align}
			Moreover, since $  \mathds{1}_{\{ \cN_{P,\cU}/ \tau  \leq \cK^2 \}} $ commutes with $\bbW_{\tau,P,\cU}	\mathds{1}_{\{ \cN_{P,\cU}/ \tau  \leq \cK^2 \}} $ and $\bbH_{\tau,0}	\mathds{1}_{\{ \cN_{P,\cU}/ \tau  \leq \cK^2 \}} $, it follows from \eqref{bound:number operator}, \eqref{control_relative_entropy:22}, \eqref{op:identity} and \eqref{trace:ineq} that
			\begin{align}	\label{first part in Gamma_tau,P-Gamma_1.6}
				  \tr_{\gF(\gH)} &  \Big(e^{-\bbH_{\tau,0}+\bbW_{\tau,P,\cU} } \frac{\cN_{Q,\cU}}{\tau} \mathds{1}_{\{ \cN_{P,\cU}/ \tau  \leq \cK^2  \}}  \Big)  \nonumber
				\\ &= \tr_{\gF(\gH)}  \Big(\frac{\cN_{Q,\cU}}{\tau} \mathds{1}_{\{ \cN_{P,\cU}/ \tau  \leq \cK^2  \}}  e^{-(\bbH_{\tau,0}-\bbW_{\tau,P,\cU})		\mathds{1}_{\{\cN_{P,\cU}/ \tau  \leq \cK^2  \}} } \Big) \nonumber
				\\ &\leq  \tr_{\gF(\gH)}  \Big( \frac{\cN_{Q,\cU}}{\tau} \mathds{1}_{\{ \cN_{P,\cU}/ \tau  \leq \cK^2  \}} e^{-\bbH_{\tau,0}	\mathds{1}_{\{ \cN_{P,\cU}/ \tau  \leq \cK^2  \}} }  e^{\bbW_{\tau,P,\cU}	\mathds{1}_{\{ \cN_{P,\cU}/ \tau  \leq \cK^2  \}} }  \Big) \nonumber
				\\ & \leq  \tr_{\gF(\gH)}  \Big( \frac{\cN_{Q,\cU}}{\tau}\mathds{1}_{\{ \cN_{P,\cU}/ \tau  \leq \cK^2  \}} e^{-\bbH_{\tau,0}	\mathds{1}_{\{ \cN_{P,\cU}/ \tau  \leq \cK^2  \}} }   \Big) \Big\|
			 e^{\bbW_{\tau,P,\cU}	\mathds{1}_{\{ \cN_{P,\cU}/ \tau  \leq \cK^2 \}} }  \Big\|   \nonumber
				\\ &\leq e^{ \eps^{-2} \|w\|_{L^\infty}^2  \cK^6}    \tr_{\gF(\gH)}  \Big( \frac{\cN_{Q,\cU}}{\tau}e^{-\bbH_{\tau,0}}   \Big) .
			\end{align}
		By the factorized property $\cU e^{-\bbH_{\tau,0}} \cU^* = e^{-\bbH_{\tau,0,P}} \otimes e^{-\bbH_{\tau,0,Q}}  $, we obtain
				\begin{align}\label{first part in Gamma_tau,P-Gamma_1.7}
					\cZ_{\tau,0,P}^{-1} & \cZ_{\tau,0,Q}^{-1}     \tr_{\gF(\gH)}  \Big( \frac{\cN_{Q,\cU}}{\tau}e^{-\bbH_{\tau,0}}   \Big)  \nonumber
					\\ &=  \cZ_{\tau,0,P}^{-1}  \cZ_{\tau,0,Q}^{-1} \tr_{\gF(P\gH)\otimes \gF(Q\gH)}\Big( \big( \mathds{1}_{\gF(P\gH)} \otimes  \frac{\cN_Q}{\tau} \big) \big(e^{-\bbH_{\tau,0,P}} \otimes e^{-\bbH_{\tau,0,Q}} \big)\Big) \nonumber 
				\\ &=   \tr_{\gF(Q\gH)} \left(  \frac{\cN_Q}{\tau}  \Gamma_{\tau,0,Q}    \right) \leq C  \Lambda_e^{-\frac12}, 
				\end{align}
			where the last inequality follows from \eqref{Wp-W:second}. Collecting the bounds \eqref{first part in Gamma_tau,P-Gamma_1.4}, \eqref{first part in Gamma_tau,P-Gamma_1.6} and \eqref{first part in Gamma_tau,P-Gamma_1.7}, we obtain
			\begin{equation}\label{first part in Gamma_tau,P-Gamma_1.5}
				\begin{aligned}
					\cZ_{\tau,0,P}^{-1} &\left| \textbf{A}_{\tau}- \textbf{A}_{\tau,P}  \right| 
				\leq 
					C \eta^{-1}\eps^{-2}\|w\|_{L^\infty}^2   \cK^6e^{ \eps^{-2}\|w\|_{L^\infty}^2   \cK^6} \Lambda_e^{-\frac12} \leq C \eta^{-1}\eps^{-2} e^{ \eps^{-2} \|w\|_{L^\infty}^2  \cK^6} \Lambda_e^{-\frac12},
				\end{aligned}
			\end{equation}
			for some constant $C$ depending only on $w$ and $\cK$.
			Replacing $ \bbW_{\tau,P,\cU} $ and $ \bbW_{\tau,P} $ in \eqref{first part in Gamma_tau,P-Gamma_1} by the identity operators $\mathds{1}_{\gF(\gH)}$ and $\mathds{1}_{\gF(P\gH)}$, respectively, and repeating the above arguments, we also control
				\begin{align}\label{first part in Gamma_tau,P-Gamma_1.3}
					\cZ_{\tau,0,P}^{-1}\left| \textbf{B}_{\tau}- \textbf{B}_{\tau,P}  \right| 
					\leq C \eta^{-1}e^{ \eps^{-2} \|w\|_{L^\infty}^2  \cK^6} \Lambda_e^{-\frac12}. 
				\end{align}

			\textit{\underline{Step~2:} lower bounds for the partition function $\textbf{B}_{\tau,P}$.}
		We first apply Golden-Thompson inequality \eqref{trace:ineq} again to obtain
			\begin{align}\label{lower:bound:for:denominator}
				\tr_{\gF(P\gH)} & \Big( f_\eta(\cN_P/\tau)  e^{-\bbH_{\tau,0,P} }   \Big)  
				 = \tr_{\gF(P\gH)}\Big(  f_\eta(\cN_P/\tau)  e^{-\bbH_{\tau,P} - \bbW_{\tau,P} }   \Big)  \nonumber
				\\ &  \leq  \tr_{\gF(P\gH)} \Big(  f_\eta(\cN_P/\tau)  e^{-\bbH_{\tau,P}  } e^{-\bbW_{\tau,P}	 }   \Big) 
	  \leq  \tr_{\gF(P\gH)}\Big(  f_\eta(\cN_P/\tau)  e^{-\bbH_{\tau,P} } \Big)
               =\textbf{B}_{\tau,P}, 
			\end{align}
			where we used $\bbW_{\tau,P}\geq 0$ in the last line. It therefore remains to derive a lower bound on the free partition function. We use \eqref{convergence on same index2:new} and \eqref{int:f(u)-f(Pu)} to obtain that for $\eta^{-1}\leq \Lambda_e^{\frac18}\leq \tau^{\frac{1}{32}}$,
			\begin{align}\label{diff:trf(NP)-intf(u)}
					\Big|  \tr_{\gF(P\gH)}  & \Big( f_\eta(\cN_P/\tau)    \Gamma_{\tau,0,P} \Big)  -	\int_{\gH} f_\eta(\|u\|_{L^2}^2) \dif \mu_0(u) \Big| \nonumber
				\\& \leq  \Big| \tr_{\gF(P\gH)}  \Big( f_\eta(\cN_P/\tau)    \Gamma_{\tau,0,P} \Big)   -	\int_{\gH} f_\eta(\|Pu\|_{L^2}^2) \dif \mu_0(u) \Big| \nonumber
	 \\ & \qquad +  \Big| 	\int_{\gH} f_\eta(\|Pu\|_{L^2}^2) \dif \mu_0(u) -\int_{\gH} f_\eta(\|u\|_{L^2}^2) \dif \mu_0(u)  \Big| \nonumber
				\\ &\leq C\tau^{-\frac14}+ C\Lambda_e^{-\frac18} \leq C\Lambda_e^{-\frac18}.
			\end{align}
			Invoking \eqref{limit:on:freegibbs} and \eqref{upper_bound_Ztau0} once more, together with \eqref{diff:trf(NP)-intf(u)}, we conclude that for all sufficiently large $\tau>0$ and $\Lambda_e>0$ satisfying $\eta^{-1}\leq \Lambda_e^{\frac18}\leq \tau^{\frac{1}{32}}$,
			\begin{equation}\label{lower:bound:for:denominator1}
				\begin{aligned}
					&\tr_{\gF(P\gH)}  \Big( f_\eta(\cN_P/\tau)    \Gamma_{\tau,0,P} \Big)  \geq    	\int_{\gH} f_\eta(\|u\|_{L^2}^2) \dif \mu_0(u) - C\Lambda_e^{-\frac18} 
					\\ & \geq \int_{\gH} \mathds{1}_{\big\{ \|u\|_{L^2}^2 \leq \frac{\cK^2}{4}\big\}} \dif \mu_0(u) -C\Lambda_e^{-\frac18} \geq \frac12 \int_{\gH} \mathds{1}_{\big\{ \|u\|_{L^2}^2 \leq \frac{\cK^2}{4}\big\}} \dif \mu_0(u)=C (\cK).
				\end{aligned}
			\end{equation}
			Combining the bounds \eqref{lower:bound:for:denominator} and \eqref{lower:bound:for:denominator1}, we deduce that for all sufficiently large $\tau>0$ and $\Lambda_e>0$ satisfying $\eta^{-1}\leq \Lambda_e^{\frac18}\leq \tau^{\frac{1}{32}}$,
			\begin{align}\label{denominator3}
				\cZ_{\tau,0,P}^{-1}  \textbf{B}_{\tau,P}=	\cZ_{\tau,0,P}^{-1} \tr_{\gF(P\gH)}  \big(  e^{-\bbH_{\tau,P}} f_\eta(\cN_P/\tau)  \big) \geq  C(\cK) .
			\end{align}
			
			\textit{\underline{Step~3:} conclusion of the bound on \eqref{difference:relative entropy}.}
			Now, we combine \eqref{bound:number operator}, \eqref{control_relative_entropy:22}, \eqref{first part in Gamma_tau,P-Gamma_1.5}, \eqref{first part in Gamma_tau,P-Gamma_1.3} and \eqref{denominator3} to deduce that, for all sufficiently large $\tau>0$ and $\Lambda_e>0$ satisfying $\eta^{-1} e^{M\eps^{-2}}  \leq \Lambda_e^{\frac18}\leq \tau^{\frac{1}{32}}$, the first term in \eqref{difference:relative entropy} is controlled by
			\begin{align}\label{Gamma_tau,P-Gamma_12}
				& \left| \frac{ \textbf{A}_{\tau,P}}{ \textbf{B}_{\tau,P}}-\frac{ \textbf{A}_\tau}{ \textbf{B}_\tau}\right| \leq \frac{\cZ_{\tau,0,P}^{-1}| \textbf{A}_{\tau,P}- \textbf{A}_{\tau}| }{\cZ_{\tau,0,P}^{-1}  \textbf{B}_{\tau,P} } + \frac{\cZ_{\tau,0,P}^{-1}| \textbf{B}_{\tau,P}- \textbf{B}_{\tau}| \textbf{A}_{\tau} }{\cZ_{\tau,0,P}^{-1}  \textbf{B}_{\tau,P}   \textbf{B}_{\tau}} 
				 \nonumber
				\\  & \stackrel{\eqref{bound:number operator}, \, \eqref{control_relative_entropy:22}}{\leq }  \frac{\cZ_{\tau,0,P}^{-1}| \textbf{A}_{\tau,P}- \textbf{A}_{\tau}| }{\cZ_{\tau,0,P}^{-1}  \textbf{B}_{\tau,P} } + \frac{\cZ_{\tau,0,P}^{-1}   \eps^{-2}\|w\|_{L^\infty}^2   \cK^6\textbf{B}_{\tau}  | \textbf{B}_{\tau,P}- \textbf{B}_{\tau}|  }{\cZ_{\tau,0,P}^{-1}  \textbf{B}_{\tau,P}   \textbf{B}_{\tau}} 
				\nonumber
				\\ &\stackrel{ \eqref{first part in Gamma_tau,P-Gamma_1.5},\, \eqref{first part in Gamma_tau,P-Gamma_1.3}, \, \eqref{denominator3}}{\leq} C \eta^{-1} \eps^{-2}  e^{ \|w\|_{L^\infty}^2  \cK^6\eps^{-2} } \Lambda_e^{-\frac12}  \leq  C	\eta^{-1}	 e^{M\eps^{-2}}\Lambda_e^{-\frac12}    \leq C \Lambda_e^{-\frac14},
			\end{align} 
			for some constant $C$ depending only on $w$ and $\cK$.

			For the second term in \eqref{difference:relative entropy}, combining the bounds \eqref{first part in Gamma_tau,P-Gamma_1.3}, \eqref{denominator3}, and \eqref{Gamma_tau,P-Gamma_12}, and using that $ \log(1+t) = O(t) $ for small $|t|$, we obtain, for all  large $\tau>0$, $\Lambda_e>0$ satisfying $\eta^{-1} e^{M\eps^{-2}}  \leq \Lambda_e^{\frac18}\leq \tau^{\frac{1}{32}}$,
			\begin{align*}
				 \left| \log \left(\frac{ \textbf{B}_\tau}{ \textbf{B}_{\tau,P}}\right) \right| =  \left| \log\!\left(1+ \frac{ \cZ_{\tau,0,P}^{-1}  (\textbf{B}_\tau -\textbf{B}_{\tau,P}) }{ \cZ_{\tau,0,P}^{-1}  \textbf{B}_{\tau,P}}\right) \right| 
			 \leq \frac{ \cZ_{\tau,0,P}^{-1} |\textbf{B}_\tau -\textbf{B}_{\tau,P} | }{ \cZ_{\tau,0,P}^{-1}  \textbf{B}_{\tau,P}}
				\leq C\Lambda_e^{-\frac{1}{4}}.
			\end{align*}
		For the last term in \eqref{difference:relative entropy}, we use $ \log(1+t) = O(t) $ for $|t|$ sufficiently small, \eqref{lower:bound:for:denominator1} and \eqref{diff:Gibbs:cutoff} to obtain
			\begin{align}\label{third:free:state:control}
				&	\left| \log \left( \frac{ \textbf{C}_\tau }{ \textbf{C}_{\tau,P}} \right) \right|
			  \leq   \frac{ \left| \tr_{\gF(P\gH)}  \Big( f_\eta(\cN_P/\tau)    \Gamma_{\tau,0,P} \Big)  -\tr_{\gF(\gH)}  \Big( f_\eta(\cN/\tau)    \Gamma_{\tau,0} \Big)  \right| }{ \tr_{\gF(P\gH)}  \Big( f_\eta(\cN_P/\tau)    \Gamma_{\tau,0,P} \Big) } \leq C\Lambda_e^{-\frac14} .
			\end{align}
			This completes the proof of Lemma~\ref{lemma:relative entropy}.
		\end{proof} 
	Combining Lemmas~\ref{errors:localization_interaction} and \ref{lemma:relative entropy}, we obtain that, for all sufficiently large $\tau$ and $\Lambda_e$, and all $ \eta\in \big(0,\frac{1}{2}\cK^2\big)$, $\eps \in(0,1)$ satisfying $\eta^{-1}e^{M\eps^{-2}} \leq \Lambda_e^{\frac18}  \leq \tau^{\frac{1}{32}}$, 
		\begin{equation}\label{all:error_in_upper_bound}
			\begin{aligned}
				&\left|    \tr_{\gF(P\gH)} \left(  \bbW_{\tau,P} \Gamma^{f_\eta}_{\tau,P} \right) -  \tr_{\gF(\gH)} \left(  \bbW_{\tau}	\Gamma_{\mathrm{ts}} \right)  +\cH\left( 	\Gamma_{\mathrm{ts}},	\Gamma^{f_\eta}_{\tau,0} \right)- \cH\left(\Gamma^{f_\eta}_{\tau,P},  \Gamma^{f_\eta}_{\tau,0,P} \right)\right| \leq O \big( \Lambda_e^{-\frac18} \big).
			\end{aligned}
		\end{equation}
		
		\begin{remark}\label{rmk:upper:exponential}
			 Lemma ~\ref{lemma:relative entropy} is the first place in the upper-bound argument where $\tau$ becomes exponentially dependent on $\eps^{-2}$. This loss comes from the rough treatment of the trace term $ \tr_{\gF(\gH)}  (e^{-\bbH_{\tau,0}+\bbW_{\tau,P,\cU} } \cN_{Q,\cU}/\tau$ $ \mathds{1}_{\{ \cN_{P,\cU}/ \tau  \leq \cK^2  \}}  )$, where the interaction weight is taken out of the trace and estimated by its operator norm on the support of the cutoff. In the subcritical regime, this step will be refined by estimating the trace directly, which removes this exponential loss in $\eps^{-2}$.
		\end{remark}
		
		\textbf{Step 3: Semi-classical analysis in projected space.}
		Let $J:= \mathrm{dim}(P\gH)= \tr(P) \leq \sqrt{\Lambda_e} $.
	We now prove the following projected semiclassical estimate. For all sufficiently large $\tau$ and $\Lambda_e$, and all $ \eta\in \big(0,\frac{1}{2}\cK^2\big)$, $\eps\in(0,1)$ satisfying $\eta^{-1}e^{M\eps^{-2}}\leq \Lambda_e^{\frac18}\leq \tau^{\frac{1}{32}}$, we have
		\begin{align}\label{upper bound:end1}
			\frac{\tr _{\gF(P\gH)}\left(e^{-\bbH_{\tau,P}}f_\eta(\mathcal{N}_P/\tau) \right) }{\tr_{\gF(P\gH)}(e^{-\bbH_{\tau,0,P}})}
			\geq \int_{P\gH}  e^{\cW_P^\eps(u) }  f_\eta\left( \|Pu\|_{L^2}^2 \right) 
			\dif \mu_{0,P}(u) - C\Lambda_e^{-\frac12}.
		\end{align}
	where $C>0$ depends only on $w$ and $\cK$. The proof proceeds naturally in two steps: first, we perform a coherent-state decomposition and separate the principal term from the cutoff errors; second, we bound the principal term from below and remove the cutoff error.
		
		\textit{\underline{Step 3a:} Coherent-state decomposition and separation of principal and error terms.} Since $[\bbH_{\tau,P}, \cN_P]= 0$, the coherent-state trace can be decomposed into particle-number sectors.
	More precisely, by the resolution of the identity in terms of coherent states \eqref{eq:resolution_coherent2}, we obtain
			\begin{align}\label{conherent-state:identity}
				\tr_{\gF(P\gH)}&\left(e^{-\bbH_{\tau,P}}f_\eta(\mathcal{N}_P/\tau) \right)
				=\pi^{-J} \int_{P\gH} \tr_{\gF(P\gH)} \left(  e^{-\bbH_{\tau,P}} f_\eta(\mathcal{N}_P/\tau) \big|\xi(u) \big\rangle \big \langle  \xi(u)\big| \right) \dif u \nonumber
				\\ &=\frac{\tau^J}{\pi^{J}} \int_{P\gH} \left\langle  \xi(\sqrt{\tau}u), e^{-\bbH_{\tau,P}} f_\eta(\mathcal{N}_P/\tau) \xi(\sqrt{\tau}u) \right\rangle \dif u
				\\ &\geq \frac{\tau^J}{\pi^{J}} \int_{P\gH} \mathds{1}_{\{\|Pu\|_{L^2}\leq \cK\}} \left\langle  \xi(\sqrt{\tau}u), e^{-\bbH_{\tau,P}} f_\eta(\mathcal{N}_P/\tau) \xi(\sqrt{\tau}u) \right\rangle \dif u \nonumber
				\\ & =\frac{\tau^J}{\pi^{J}} \int_{P\gH} \mathds{1}_{\{\|Pu\|_{L^2}\leq \cK\}} e^{-\tau \|Pu\|_{L^2}^2} \sum_{n=0}^\infty \frac{(\tau\|Pu\|_{L^2}^2 )^n}{n!} f_\eta \left(\frac{n}{\tau}\right) \left\langle \frac{(Pu)^{\otimes n} }{\|Pu\|_{L^2}^n} ,e^{-\bbH_{\tau,P}^{(n)}}  \frac{(Pu)^{\otimes n} }{\|Pu\|_{L^2}^n}\right\rangle \dif u .\nonumber
			\end{align}
		Here, the $n$-body Hamiltonian is defined by $\bbH_{\tau,P}^{(n)}= \frac{1}{\tau}\sum_{j=1}^n (PhP)_{j}- \frac{1}{\tau^3}\sum_{1 \leq i_1 < i_2< i_3 \leq n} (W_P^\eps)_{i_1,i_2,i_3}$, $n\geq 3$, whereas for $n=1,2$, only the kinetic part remains. In the formulas above we first assume $\|Pu\|_{L^2}>0$. The resulting estimate extends to the case $\|Pu\|_{L^2}=0$ since the corresponding coherent-state Poisson distribution has mean zero and only the sector $n=0$ contributes. By the Peierls-Bogoliubov inequality $\langle x,e^A x \rangle \geq e^{\langle x,Ax \rangle}$ for normalized vector $x$, we obtain
		\begin{align*}
			\left\langle \frac{(Pu)^{\otimes n} }{\|Pu\|_{L^2}^n} ,e^{-\bbH_{\tau,P}^{(n)}}  \frac{(Pu)^{\otimes n} }{\|Pu\|_{L^2}^n}\right\rangle  \geq \mathrm{exp} \left(  - 	\left\langle \frac{(Pu)^{\otimes n} }{\|Pu\|_{L^2}^n} ,\bbH_{\tau,P}^{(n)}  \frac{(Pu)^{\otimes n} }{\|Pu\|_{L^2}^n}\right\rangle \right) =: e^{-\cE_{n,\tau} (Pu)}.
		\end{align*}
	For $\|Pu\|_{L^2}>0$, this exponent is explicitly
		\begin{align*}
			\cE_{n,\tau} (Pu) = \begin{cases}
				\frac{n}{\tau\|Pu\|_{L^2}^2 }\langle u,PhPu \rangle -  \frac{n(n-1)(n-2)}{\tau^3\|Pu\|_{L^2}^6 } \cW^\eps_P,& n\geq 3,
				\\ 	 \frac{n}{\tau\|Pu\|_{L^2}^2 }\langle u,PhPu \rangle, & n=1,2,
				\\ 0, & n=0.
			\end{cases}
		\end{align*}
Consequently, inserting this expression into \eqref{conherent-state:identity} yields
		\begin{align}\label{tr:lower_two_part}
				\tr_{\gF(P\gH)} \left( e^{-\bbH_{\tau,P}}f_\eta(\mathcal{N}_P/\tau) \right) & \geq \frac{\tau^J}{\pi^{J}} \int_{P\gH} \mathds{1}_{\{\|Pu\|_{L^2}\leq \cK\}} e^{-\tau \|Pu\|_{L^2}^2} \sum_{n=0}^\infty \frac{(\tau\|Pu\|_{L^2}^2 )^n}{n!} f_\eta \left(\frac{n}{\tau}\right) e^{-\cE_{n,\tau} (Pu)} \dif u \nonumber
				\\ &  =  \frac{\tau^J}{\pi^{J}} \int_{P\gH} \mathds{1}_{\{\|Pu\|_{L^2}\leq \cK\}}  \mathbf{E} \Big( f_\eta  \Big( \frac{\mathbf{X}}{\tau}\Big)  e^{-\cE_{\mathbf{X},\tau } (Pu) } \Big)  \dif u,
		\end{align} 
		where, for each fixed $u$, $\mathbf{X}$ denotes a Poisson random variable with mean $\tau \|Pu\|_{L^2}^2$. Setting $\cE(Pu):= \langle u,PhPu \rangle- \cW^\eps_P$, we claim that,
	\begin{align}\label{assertion:diff:Hamiltonian}
	\mathds{1}_{\{n/\tau \leq \cK^2,\, \|Pu\|_{L^2}\leq \cK\} }  	\left| \cE_{n,\tau} (Pu) - \cE(Pu)  \right |
\leq C_{\cK,w} (\Lambda_e +\eps^{-2})  \left(\Big| \frac{n}{\tau } - \|Pu\|_{L^2}^2 \Big| +\tau^{-1}\right).
	\end{align}
	Indeed, since $PhP\leq \Lambda_e P$, the kinetic contribution is bounded, for all $n\geq0$, by
		\begin{align*}
			\left| \frac{n}{\tau\|Pu\|_{L^2}^2 }\langle u,PhPu \rangle - \langle u,PhPu \rangle  \right| \leq 	\left| 1- \frac{n}{\tau\|Pu\|_{L^2}^2 }  \right |  \langle u,PhPu \rangle \leq \Lambda_e  \left|  \frac{n}{\tau }  - \|Pu\|_{L^2}^2 \right |.
		\end{align*}
	For the interaction contribution, using $\cW^\eps_P \leq \eps^{-2} \|w\|_{L^\infty}^2 \|Pu\|_{L^2}^6 $, we obtain, for $n\geq 3$,
		\begin{align*}
			&	\mathds{1}_{\{n/\tau \leq \cK^2,\, \|Pu\|_{L^2}\leq \cK\} }  \Big| \frac{n(n-1)(n-2)}{\tau^3\|Pu\|_{L^2}^6 } \cW^\eps_P -  \cW^\eps_P \Big| 
			\\ &\qquad \leq  	\mathds{1}_{\{n/\tau \leq \cK^2,\, \|Pu\|_{L^2}\leq \cK\} }    \Big|   \frac{n(n-1)(n-2) }{\tau^3 } - \|Pu\|_{L^2}^6 \Big|  \eps^{-2} \|w\|_{L^\infty}^2  
			 \\ & \qquad \leq 	\mathds{1}_{\{n/\tau \leq \cK^2,\, \|Pu\|_{L^2}\leq \cK\} }     \left(  \Big| \frac{n^3 }{\tau^3 } -  \|Pu\|_{L^2}^6  \Big| + \frac{3 n^2 }{\tau^3 } + \frac{2n }{\tau^3 } \right) \eps^{-2} \|w\|_{L^\infty}^2
			 \\ &\qquad \leq C_{\cK,w} \left(\Big| \frac{n}{\tau } - \|Pu\|_{L^2}^2 \Big| +\tau^{-1}\right) \eps^{-2} .
		\end{align*}
	For $n=0,1,2$, the interaction term is absent from $\cE_{n,\tau}$. In this case the same bound follows from
		\begin{align*}
			& 	\mathds{1}_{\{n/\tau \leq \cK^2,\, \|Pu\|_{L^2}\leq \cK\} }   \cW^\eps_P \leq  \eps^{-2} \|w\|_{L^\infty}^2\cK^4 \|Pu\|_{L^2}^2
				\\ &\qquad \leq  \eps^{-2} \|w\|_{L^\infty}^2\cK^4  \Big|\frac{n}{\tau}-\|Pu\|_{L^2}^2\Big| + 2 \eps^{-2} \|w\|_{L^\infty}^2\cK^4 \frac{1}{\tau} \leq  C_{\cK,w}  \left(\Big| \frac{n}{\tau } - \|Pu\|_{L^2}^2 \Big| +\tau^{-1}\right) \eps^{-2} .
		\end{align*}
   This proves the claim \eqref{assertion:diff:Hamiltonian}. Turning back to \eqref{tr:lower_two_part}, we use $e^{-x} \geq 1-x$ together with \eqref{assertion:diff:Hamiltonian} to obtain
    \begin{align}\label{transform:numbercut:to:masscut}
    &	\mathds{1}_{\{ \|Pu\|_{L^2}\leq \cK\} }  	f_\eta  \Big( \frac{\mathbf{X}}{\tau}\Big)  e^{-\cE_{\mathbf{X},\tau } (Pu) } \geq \mathds{1}_{\{\|Pu\|_{L^2}\leq \cK\} }    f_\eta  \Big(\frac{\mathbf{X}}{\tau}\Big) e^{-\cE(Pu)} e^{-|\cE_{\mathbf{X},\tau}(Pu)-\cE(Pu) |} \nonumber
    \\ & \quad  \geq  \mathds{1}_{\{\|Pu\|_{L^2}\leq \cK\} }   f_\eta  \Big(\frac{\mathbf{X}}{\tau}\Big) e^{-\cE(Pu)} - \mathds{1}_{\{\|Pu\|_{L^2}\leq \cK\} }  f_\eta  \Big(\frac{\mathbf{X}}{\tau}\Big) e^{-\cE(Pu)} |\cE_{\mathbf{X},\tau}(Pu)-\cE(Pu) | \nonumber
    	\\ & \quad \geq \mathds{1}_{\{\|Pu\|_{L^2}\leq \cK\} }  \Big( 	f_\eta  \big(\|Pu\|_{L^2}^2\big) -\|f_\eta' \|_{L^\infty} \Big| 	 \frac{ \mathbf{X}}{\tau}- \|Pu\|_{L^2}^2   \Big|\Big) e^{-\cE(Pu)}  \nonumber
    	\\ &  \qquad\qquad \qquad-C_{\cK,w} (\Lambda_e +\eps^{-2})  \mathds{1}_{\{\|Pu\|_{L^2}\leq \cK\} }  \left(  \Big| \frac{ \mathbf{X}}{\tau } - \|Pu\|_{L^2}^2 \Big| +\tau^{-1}\right) e^{-\cE(Pu)}  . 
    \end{align}
    Since $\mathbf{X}$  is a Poisson random variable with mean $\tau\|Pu\|_{L^2}^2  $, we have
    \begin{align}\label{expectation:Poisson}
     \mathds{1}_{\{\|Pu\|_{L^2}\leq \cK\} } 	\mathbf{E} \Big| 	 \frac{ \mathbf{X}}{\tau}- \|Pu\|_{L^2}^2   \Big| \leq \tau^{-1}  \mathds{1}_{\{\|Pu\|_{L^2}\leq \cK\} }  \left(	\mathbf{E} \Big| 	  \mathbf{X}- \tau \|Pu\|_{L^2}^2   \Big|^2 \right)^{1/2} \leq \cK\tau^{-\frac12}.
    \end{align}
Taking expectation on both sides of \eqref{transform:numbercut:to:masscut} and inserting the resulting bound into \eqref{tr:lower_two_part}, we obtain
    	\begin{align}\label{lower:bound3}
    	\tr_{\gF(P\gH)}& \left( e^{-\bbH_{\tau,P}}f_\eta(\mathcal{N}_P/\tau) \right)  \geq  \frac{\tau^J}{\pi^{J}} \int_{P\gH} 	f_\eta  \big(\|Pu\|_{L^2}^2\big)  e^{-\cE(Pu)}  \dif u
    	\\ &- \frac{\tau^J}{\pi^{J}} \int_{P\gH} \eta^{-1}  \mathds{1}_{\{\|Pu\|_{L^2}\leq \cK\} }  \mathbf{E} \Big| 	 \frac{ \mathbf{X}}{\tau}- \|Pu\|_{L^2}^2   \Big| e^{-\cE(Pu)}  \dif u \nonumber
    	\\ &-  C_{\cK,w} (\Lambda_e +\eps^{-2})   \frac{\tau^J}{\pi^{J}} \int_{P\gH}  \mathds{1}_{\{\|Pu\|_{L^2}\leq \cK\} }\left(  \mathbf{E} \Big| \frac{ \mathbf{X}}{\tau } - \|Pu\|_{L^2}^2 \Big| +\tau^{-1}\right) e^{-\cE(Pu)} \dif u   \nonumber
    	\\ &\geq \frac{\tau^J}{\pi^{J}} \int_{P\gH} 	f_\eta  \big(\|Pu\|_{L^2}^2\big)  e^{-\cE(Pu)}  \dif u- C_{\cK,w} \frac{ \Lambda_e}{\sqrt{\tau}}   \frac{\tau^J}{\pi^{J}} \int_{P\gH}  \mathds{1}_{\{\|Pu\|_{L^2}\leq \cK\} } e^{-\cE(Pu)} \dif u,\nonumber
    \end{align} 
    where we used \eqref{expectation:Poisson} and $\eta^{-1}\eps^{-2}\leq \Lambda_e$ in the last inequality.
    
    	\textit{\underline{Step 3b:} 
    Lower bound on the principal term and removal of the cutoff errors.}
		From \cite[(8.18)]{LNR15}, the free partition function is given by $\tr_{\gF(P\gH)}(e^{-\bbH_{\tau,0,P}})=\prod_{j=1}^J \frac{1}{1-e^{-\lambda_j/\tau}}$. 
		As already observed in \cite[(9.15)]{LNR21}, Bernoulli's inequality, together with $\Lambda_e\leq \tau^{\frac14}$ yields
		\begin{align}\label{bernoulli:inequality}
			\prod_{j=1}^J \frac{\tau}{\lambda_j} (1-e^{-\lambda_j/\tau})  \geq \prod_{j=1}^J \left( 1- \frac{\lambda_j}{2\tau}\right) \geq 1-\sum_{j=1}^J \frac{\lambda_j}{2\tau} \geq 1-\frac{\Lambda_e^{3/2}}{\tau} \geq 1- \Lambda_e^{-1}.
		\end{align}
		Applying \eqref{bernoulli:inequality} to the first term on the right-hand side of \eqref{lower:bound3}, we obtain
		\begin{align}	\label{first_hard_term1}
			&\frac{1}{\tr_{\gF(P\gH)}(e^{-\bbH_{\tau,0,P}})}\frac{\tau^J}{\pi^{J}}  \int_{P\gH} f_\eta\left( \|Pu\|_{L^2}^2 \right) e^{-\cE(Pu)}
			\dif u 
			\\ & \qquad = \prod_{j=1}^J \frac{\tau}{\lambda_j} (1-e^{-\lambda_j/\tau})  \int_{P\gH} f_\eta\left( \|Pu\|_{L^2}^2 \right)  e^{ \cW_P^\eps }
			\prod_{j=1}^J \frac{\lambda_j}{\pi} e^{-\langle u,PhPu\rangle}	\dif u  \nonumber 
			\\&   \qquad \geq \left( 1-\Lambda_e^{-1} \right) \int_{P\gH} f_\eta\left( \|Pu\|_{L^2}^2 \right)  e^{\cW_P^\eps }  
		 \dif \mu_{0,P}(u)  
	\geq \int_{P\gH} f_\eta\left( \|Pu\|_{L^2}^2 \right)  e^{\cW_P^\eps }  
		 \dif \mu_{0,P}(u) - C\Lambda_e^{-\frac12} ,   \nonumber
		\end{align}  
		where the last inequality follows from the pointwise bound on the support of the cutoff,
		\begin{align*}
\Lambda_e^{-1} \int_{P\gH} f_\eta\left( \|Pu\|_{L^2}^2 \right)  e^{\cW_P^\eps } 
	\dif \mu_{0,P}(u) \leq \Lambda_e^{-1} e^{\eps^{-2}\|w\|_{L^\infty}^2 \cK^6} \leq \Lambda_e^{-\frac12}.
		\end{align*}
		For the second term on the right-hand side of \eqref{lower:bound3}, the same argument gives
			\begin{align}\label{second_easy_term}
			&C_{\cK,w} \frac{ \Lambda_e}{\sqrt{\tau}}  	\frac{1}{\tr_{\gF(P\gH)}(e^{-\bbH_{\tau,0,P}})}   \frac{\tau^J}{\pi^{J}} \int_{P\gH}  \mathds{1}_{\{\|Pu\|_{L^2}\leq \cK\} } e^{-\cE(Pu)} \dif u 
				\\ &  =C_{\cK,w} \frac{ \Lambda_e}{\sqrt{\tau}}   \prod_{j=1}^J \frac{\tau}{\lambda_j} (1-e^{-\lambda_j/\tau})  \int_{P\gH}  \mathds{1}_{\{\|Pu\|_{L^2}\leq \cK\} } e^{\cW_P^\eps } 
				\dif \mu_{0,P}(u) 
			 \leq C_{\cK,w} \frac{ e^{\eps^{-2}\|w\|_{L^\infty}^2 \cK^6} \Lambda_e  }{\sqrt{\tau}}  \leq  C \Lambda_e^{-\frac12},  \nonumber
			\end{align}
		where in the last inequality we used $\Lambda_e \leq   \tau^{\frac14}$.
	Combining \eqref{lower:bound3}, \eqref{first_hard_term1} and \eqref{second_easy_term}, we obtain \eqref{upper bound:end1}.
	\begin{remark}\label{rmk:upper:exponential2}
		 The estimates \eqref{first_hard_term1} and \eqref{second_easy_term} are another source of exponential dependence on $\eps^{-2}$ in the upper-bound argument. This comes from the pointwise control of the interaction weight on the support of the cutoff, namely $e^{\cW_P^\eps } \leq e^{\eps^{-2}\|w\|_{L^\infty}^2 \cK^6} $. In the subcritical regime, this pointwise estimate will be replaced by the uniform classical moment bound, thereby removing the corresponding exponential loss.
	\end{remark}

		\textbf{Step~4: Conclusion of \eqref{ineq:upper bound}.}
		Going back to  \eqref{upper bound to finite dimension}, we first use \eqref{upper bound:end1} to derive
		\begin{equation}\label{finite_analysis1}
			\begin{aligned}
				-\log \frac{\cZ^{f_\eta}_{\tau,P}}{\cZ^{f_\eta}_{\tau,0,P}} &= -\log  \frac{ \tr_{\gF(P\gH)}\left (e^{-\bbH_{\tau,P}}f_\eta(\mathcal{N}_P/\tau)\right) }{\tr_{\gF(P\gH)}( e^{-\bbH_{\tau,0,P}})}
				+ \log \tr_{\gF(P\gH)}(\Gamma_{\tau,0,P} f_\eta(\cN_P/\tau))
				\\ &\leq -\log \left( \int_{P \gH}  e^{\cW_P^\eps(u) }  f_\eta\left( \|Pu\|_{L^2}^2 \right) 
				\dif \mu_{0,P}(u) - \Lambda_e^{-\frac12} \right) +\log  \tr_{\gF(P\gH)}\big(\Gamma_{\tau,0,P} f_\eta(\cN_P/\tau)\big)
				\\ &= -\log \left( \int_{P \gH}  e^{\cW_P^\eps(u) }  f_\eta\left( \|Pu\|_{L^2}^2 \right) 
				\dif \mu_{0,P}(u) - \Lambda_e^{-\frac12} \right) 
				\\ &\qquad \qquad + \log \left(\int_{P \gH} f_\eta\left( \|Pu\|_{L^2}^2 \right) \dif \mu_{0,P}(u)   \right) + \log \left( \frac{\tr_{\gF(P\gH)}\big(\Gamma_{\tau,0,P} f_\eta(\cN_P/\tau)\big)}{\int_{\gH} f_\eta\left( \|Pu\|_{L^2}^2 \right) \dif \mu_{0}(u) }\right).
			\end{aligned}
		\end{equation}
		To control the right-hand side, we first provide a lower bound for $\int_{P \gH}  e^{\cW_P^\eps(u) }  f_\eta\left( \|Pu\|_{L^2}^2 \right) 
		\dif \mu_{0,P}(u) $.  Using \eqref{int:f(u)-f(Pu)}, we obtain that for all sufficiently large $\Lambda_e>0$ satisfying $\eta^{-1}\leq \Lambda_e^{\frac18}$,
		\begin{equation}\label{lower:bound:for:integral:f(Pu)}
			\begin{aligned}
				\int_{P \gH} &  e^{\cW_P^\eps(u) }  f_\eta\left( \|Pu\|_{L^2}^2 \right) 
			\dif \mu_{0,P}(u) 	\geq \int_{\gH}  f_\eta\left( \|Pu\|_{L^2}^2 \right) 
				\dif \mu_{0}(u) \geq \int_{\gH}  f_\eta\left( \|u\|_{L^2}^2 \right) 
				\dif \mu_{0}(u) -C\Lambda_e^{-\frac18} 
				\\ & \geq \int_{\gH}\mathds{1}_{\{ \|u\|_{L^2}^2 \leq \frac{\cK^2}{2}\}}	 \dif \mu_{0}(u) -C\Lambda_e^{-\frac18} \geq \frac12 \int_{\gH}\mathds{1}_{\{ \|u\|_{L^2}^2 \leq \frac{\cK^2}{2}\}}	 \dif \mu_{0}(u) >0.
			\end{aligned}
		\end{equation}
		By this bound and the fact $\log(1+t)=O(|t|)$ for small $|t|$, we deduce that for sufficiently large $\tau>0$ and $\Lambda_e>0$ satisfying $\eta^{-1}\leq \Lambda_e^{\frac18}$,
			\begin{align}\label{finite_analysis2}
				- &\log \left( \int_{P \gH}  e^{\cW_P^\eps(u) }  f_\eta\left( \|Pu\|_{L^2}^2 \right) 
				\dif \mu_{0,P}(u) - \Lambda_e^{-\frac12} \right)  + \log \left(\int_{P \gH} f_\eta\left( \|Pu\|_{L^2}^2 \right) \dif \mu_{0,P}(u)   \right)\nonumber
				\\ &=  -\log \left( \int_{P \gH}  e^{\cW_P^\eps(u) }  
				\dif \mu_{0,P}^{f_\eta}(u) \right) -\log \left(  1- \frac{\Lambda_e^{-\frac12}}{ \int_{P \gH}  e^{\cW_P^\eps(u) }  f_\eta\left( \|Pu\|_{L^2}^2 \right) 
					\dif \mu_{0,P}(u) } \right)
				\\& \leq -\log \left( \int_{P \gH}  e^{\cW_P^\eps(u) } 
				\dif \mu_{0,P}^{f_\eta}(u) \right) + C\Lambda_e^{-\frac12}.\nonumber
			\end{align}
		For the third term on the right-hand side of \eqref{finite_analysis1}, using \eqref{diff:trf(NP)-intf(u)}, \eqref{lower:bound:for:integral:f(Pu)}, \eqref{convergence on same index2:new} and $\log(1+t)=O(|t|)$ for $|t|$ small again, we derive for sufficiently large $\tau>0$ and $\Lambda_e>0$ satisfying $\eta^{-1}\leq \Lambda_e^{\frac18}\leq  \tau^{\frac{1}{32}}$,
		\begin{align}\label{finite_analysis3}
			\log &\left( \frac{\tr_{\gF(P\gH)}(\Gamma_{\tau,0,P} f_\eta(\cN_P/\tau))}{\int_{\gH} f_\eta\left( \|Pu\|_{L^2}^2 \right) \dif \mu_{0}(u) }\right) =  \log \left(1+ \frac{\tr_{\gF(P\gH)}(\Gamma_{\tau,0,P} f_\eta(\cN_P/\tau))-\int_{\gH} f_\eta\left( \|Pu\|_{L^2}^2 \right) \dif \mu_{0}(u) }{\int_{\gH} f_\eta\left( \|Pu\|_{L^2}^2 \right) \dif \mu_{0}(u) } \right) \nonumber
			\\ &\leq C\left|\tr_{\gF(P\gH)}\big(\Gamma_{\tau,0,P} f_\eta(\cN_P/\tau)\big)-\int_{\gH} f_\eta\left( \|Pu\|_{L^2}^2 \right) \dif \mu_{0}(u) \right|    
  \leq  C \tau^{-\frac14} \leq C\Lambda_e^{-\frac18}.
		\end{align}
		Finally, collecting all error terms and substituting \eqref{all:error_in_upper_bound}, \eqref{finite_analysis1}, \eqref{finite_analysis2} and \eqref{finite_analysis3} into \eqref{upper bound to finite dimension}, we obtain 
		\begin{align*}
			-\log \frac{\cZ^{f_\eta}_{\tau}}{\cZ^{f_\eta}_{\tau,0}}  \leq - \log \left( \int_{P\gH}  e^{\cW_P^\eps(u) } 
			\dif \mu^{f_\eta}_{0,P}(u) \right) + O\left(  \Lambda_e^{-\frac18}\right).
		\end{align*}
		This concludes the proof of Proposition~\ref{prop:upper bound}.
	\end{proof}
	
	\subsection{Free energy lower bound}\label{sec:lower bound}
	We now establish the lower bound matching Proposition~\ref{prop:upper bound}. Up to the first localization step, the proof follows the same strategy as for the upper bound: we project onto $P\gH$, apply the Berezin-Lieb inequality, and use lower symbols to reduce the quantum variational problem to a classical one. The control of the localization error proceeds exactly as in the proof of Lemma~\ref{errors:localization_interaction}. The main difficulty arises after this reduction, when one has to compare the lower symbol $\mu_{P,0}^{\tau^{-1}}$ of the truncated free Gibbs state with the projected Gaussian measure $\mu^{f_\eta}_{0,P}$.
Since the relevant observable is the potentially large weight $e^{\cW_P^\eps}$, one requires a quantitative $L^1$-comparison between these two measures. Another key observation is an interacting lower-symbol tail estimate (see Lemma~\ref{lemma:interacting:tail}), which gives the required control of the mass tail and allows us to complete this comparison. The following result is the lower-bound counterpart of Proposition~\ref{prop:upper bound}.
	\begin{proposition}\label{prop:lower bound}
		Let $w$ satisfy Assumption~\ref{assum:on:w}, and let $\cK>0$ be arbitrary. Let $\Lambda_e>0$ be sufficiently large, and let $\eta \in \big(0 ,\frac{1}{2}\cK^2  \big)$ and $\eps \in (0,1)$ satisfy $\eta^{-1} e^{M\eps^{-2}} \leq \Lambda_e^{\frac18}$, where $M=(2\|w\|^2_{L^\infty}+1)(\cK^2+4)^3$. Let the cutoff $f_\eta$ satisfy Assumption~\ref{assum:cutoff} with cutoff parameter $\cK$. 
		Let $\cZ^{f_\eta}_{\tau}$ and $\cZ^{f_\eta}_{\tau,0}$ be the partition functions defined in \eqref{def:interacting:gibbs_state} and \eqref{def:free:gibbs_state}, respectively, associated with the cutoff $f_\eta$. Then, for all sufficiently large $ \tau$ and $\Lambda_e$ satisfying $\Lambda_e \leq  \tau^{\frac14}$, we have
		\begin{align}\label{ineq:lower bound}
			-\log  \frac{\cZ^{f_\eta}_{\tau}}{\cZ^{f_\eta}_{\tau,0}} \geq -\log \left( \int_{P\gH} e^{ \cW_P^\eps(u) } \dif \mu^{f_\eta}_{0,P} (u)\right) - O\left( \Lambda_e^{-\frac{1}{8}}   \right) ,
		\end{align}
		where $P:=\mathds{1}(h\leq \Lambda_e)$ is the orthogonal projection on $ \gH$, $\cW^\eps_P$ is the interaction potential defined in \eqref{def:widetilde:WP} and $\mu^{f_\eta}_{0,P} $ is the projected Gaussian measure defined in \eqref{truncated Gaussian measure}.
	\end{proposition}
	\begin{proof}
		Let $J:=\tr[P]$ and define $W^\eps_P=P^{\otimes 3} W^\eps P^{\otimes 3} $. We denote by $(\Gamma^{f_\eta}_{\tau})_{P}^{(3)}= P^{\otimes 3} (\Gamma_{\tau}^{f_\eta})^{(3)} P^{\otimes 3}$ the $P$-localization of $(\Gamma_{\tau}^{f_\eta})^{(3)} $, and by $\mu_{P,\tau}^{\tau^{-1}}$ and $\mu_{P,0}^{\tau^{-1}}$  
		the corresponding lower symbols of $(\Gamma^{f_\eta}_\tau)_P$ and $(\Gamma^{f_\eta}_{\tau,0})_P$, respectively. 
		
		\textbf{Step~1: Localization of the interaction and control of the error term.}
		By applying the variational principle \eqref{variational principle quantum}, \eqref{variational principle classical} and the Berezin-Lieb type inequality \eqref{eq:Berezin-Lieb}, we obtain that
		\begin{equation}\label{lower:variational}
			\begin{aligned}
				-\log  \frac{\cZ^{f_\eta}_{\tau}}{\mathcal{Z}^{f_\eta}_{\tau,0}}  &=\cH(\Gamma^{f_\eta}_{\tau},\Gamma^{f_\eta}_{\tau,0}) - \frac{1}{ \tau^3} \Tr \big( W^\eps(\Gamma_{\tau}^{f_\eta})^{(3)} \big)
				\\ &\geq  \cH_{\mathrm{cl}} (\mu_{P,\tau}^{\tau^{-1}}, \mu_{P,0}^{\tau^{-1}})- \frac{1}{\tau^3} \tr  \left(W^\eps_{P} (\Gamma_{\tau}^{f_\eta})^{(3)} \right)
			- \frac{1}{\tau^3}\tr \left( \left(W^\eps-W^\eps_{P}\right) (\Gamma_{\tau}^{f_\eta})^{(3)} \right).
			\end{aligned}
		\end{equation}
	We begin with the localized interaction terms in \eqref{lower:variational}. By \eqref{eq:Chiribella}, we derive
		\begin{align*}
		-&	\frac{1}{\tau^3}  \tr  \left(W^\eps_{P} (\Gamma_{\tau}^{f_\eta})^{(3)} \right)= - \frac{1}{\tau^3} \tr \left(W^\eps_{P} P^{\otimes 3}(\Gamma_{\tau}^{f_\eta})^{(3)} P^{\otimes 3} \right) 
		\\ &= - \frac{1}{\tau^3} \tr \left(W^\eps_P \big((\Gamma_{\tau}^{f_\eta})_P\big)^{(3)} \right) \nonumber
			=- \frac{1}{6} \int_{P \gH} \langle  u^{\otimes 3}, W^\eps_{P} u^{\otimes 3}\rangle \dif \mu_{P,\tau}^{\tau^{-1}}(u) 
			\\&\qquad  \qquad  \qquad  \qquad  \qquad \qquad + \frac{1}{ \tau^3} \sum_{\ell=0}^{2}{3\choose \ell}^2  \frac{(3-\ell +J-1)!}{(J-1)!} \tr_{\gF(P\gH)} \Big( W^\eps_P \big((\Gamma^{f_\eta}_{\tau})_{P}^{(\ell)} \otimes_s \mathds{1}_{\otimes_s^{3-\ell} P\gH}\big) \Big) .
		\end{align*} 
		Since $W^\eps_P\geq 0$ and $(\Gamma^{f_\eta}_{\tau})_{P}^{(\ell)} \otimes_s \mathds{1}_{\otimes_s^{3-\ell} P\gH}\geq 0$, each trace term in the sum is nonnegative. Therefore,
		\begin{align}\label{lower bound interaction energy1}
		-	\frac{1}{\tau^3}  \tr  \left(W^\eps_{P} (\Gamma_{\tau}^{f_\eta})^{(3)} \right)\geq -\frac{1}{6} \int_{P\gH} \langle  u^{\otimes 3}, W^\eps_P u^{\otimes 3}\rangle \dif \mu_{P,\tau}^{\tau^{-1}}(u)= - \int_{P\gH}  \cW^\eps_P(u) \dif \mu_{P,\tau}^{\tau^{-1}}(u).
		\end{align}

		We next estimate the localization errors on the second line of \eqref{lower:variational}. As in Lemma~\ref{errors:localization_interaction}, these terms are reduced to the high-frequency one-body mass. More precisely, repeating the computations in \eqref{CS:ineq}--\eqref{Wp-W:main}, we obtain
		\begin{align*}
			\left| \frac{1}{\tau^3}\tr \left( \left(W^\eps-W^\eps_{P}\right) (\Gamma_{\tau}^{f_\eta})^{(3)}\right)\right|
		& \leq  C \eps^{-2}\|w\|^2_{L^\infty}  \cK^3  \left( \frac{1}{\tau}  \left| \tr \left( Q \big( (\Gamma_{\tau}^{f_\eta})^{(1)} - \Gamma_{\tau,0}^{(1)} \big) \right) \right|+ \frac{1}{\tau}   \tr \left( Q\Gamma_{\tau,0}^{(1)} \right) \right)^{\frac12}
			\\&  \quad + C\eps^{-2}\|w\|^2_{L^\infty}  \left( \frac{1}{\tau}  \left| \tr \left( Q \big((\Gamma_{\tau}^{f_\eta})^{(1)} - \Gamma_{\tau,0}^{(1)} \big) \right) \right|+ \frac{1}{\tau}   \tr \left( Q\Gamma_{\tau,0}^{(1)} \right) \right).
		\end{align*}
		To bound the right-hand side, it is enough to control the relative entropy $\cH(\Gamma^{f_\eta}_{\tau}, \Gamma_{\tau,0})$, thanks to \eqref{control_relative_entropy:1}. Repeating the argument in \eqref{control_relative_entropy:2}--\eqref{upper_bound_Ztau0}, we obtain for sufficiently large $\tau>0$
		\begin{align*}
			\cH\big( \Gamma^{f_\eta}_{\tau}, \Gamma_{\tau,0} \big)  
			& \leq  -\tr_{\gF(\gH)} \left( \bbW_{\tau} \Gamma^{f_\eta}_{\tau,0}\right)-\log\left( \frac{\cZ^{f_\eta}_{\tau,0}}{\cZ_{\tau,0}}\right) + \tr_{\gF(\gH)} \left( \bbW_{\tau}\Gamma^{f_\eta}_{\tau} \right)
			\\& \leq  -\log\left[ \tr_{\gF(\gH)} \left( \Gamma_{\tau,0} f_\eta \left(\frac{\cN}{\tau}\right) \right)  \right]+2\|w\|_{L^\infty}^2\eps^{-2}\cK^6
		 \\ & \leq \log \left( 1/ C(\cK)\right)+2\|w\|_{L^\infty}^2\eps^{-2}\cK^6.
		\end{align*}
	The remaining estimates are identical to those in \eqref{control_relative_entropy:11}. Hence, for all sufficiently large $\tau>0$, $\Lambda_e>0$, $ \eta\in \big(0,\frac{1}{2}\cK^2\big)$ and $ \eps \in(0,1)$ satisfying $ \eps^{-3}\leq \Lambda_e^{\frac{1}{8}}$, we conclude that
		\begin{align}\label{errors:localization_interaction2}
			\left| \frac{1}{\tau^3} \tr \left( \left(W^\eps-W^\eps_{P}\right) (\Gamma_{\tau}^{f_\eta})^{(3)}\right)  \right| \leq C\Lambda_e^{-\frac18}.
		\end{align}
	Combining \eqref{lower:variational}--\eqref{errors:localization_interaction2} and invoking the classical variational principle \eqref{variational principle classical} once more, we obtain that for all sufficiently large $\tau>0$, $\Lambda_e>0$, $ \eta\in \big(0,\frac{1}{2}\cK^2\big)$ and $ \eps \in(0,1)$ satisfying $ \eps^{-3}\leq \Lambda_e^{\frac{1}{8}}$, and any $R>0$,
		\begin{align}\label{variation}
			-\log  \frac{\cZ^{f_\eta}_{\tau}}{\mathcal{Z}^{f_\eta}_{\tau,0}}  
			&\geq \cH_{\mathrm{cl}} (\mu_{P,\tau}^{\tau^{-1}}, \mu_{P,0}^{\tau^{-1}})-  \int_{P\gH}  \cW^\eps_P   \dif \mu_{P,\tau}^{\tau^{-1}}(u)-C\Lambda_e^{-\frac18} \nonumber
			\\ & =  \cH_{\mathrm{cl}} (\mu_{P,\tau}^{\tau^{-1}}, \mu_{P,0}^{\tau^{-1}})-   \int_{P\gH}  \cW^\eps_P \mathds{1}_{\{\|u\|_{L^2}^2 \leq R \}}  \dif \mu_{P,\tau}^{\tau^{-1}}(u)- \int_{P\gH}  \cW^\eps_P \mathds{1}_{\{\|u\|_{L^2}^2 > R \}}  \dif \mu_{P,\tau}^{\tau^{-1}}(u)-C\Lambda_e^{-\frac18} \nonumber
			\\ & \geq -\log \left( \int_{P\gH} e^{ \cW_P^\eps \mathds{1}_{\{\|u\|_{L^2}^2 \leq R \}} } \dif   \mu_{P,0}^{\tau^{-1}}(u) \right)- \int_{P\gH}  \cW^\eps_P \mathds{1}_{\{\|u\|_{L^2}^2 > R \}}  \dif \mu_{P,\tau}^{\tau^{-1}}(u) -C\Lambda_e^{-\frac18} .
		\end{align}
			\textbf{Step 2: Interacting lower-symbol tail estimate.}
		 		To control the second term on the right-hand side of \eqref{variation}, we first prove a general tail estimate for interacting lower symbols.
		\begin{lemma}\label{lemma:interacting:tail}
			Let $K>0$ and $R>K^2$. Let $\Gamma$ be a state on	$\gF(\gH)$ supported in the low-particle-number sector $\{\cN\leq K^2\tau\}$. Let $P$ be a $J$-dimensional orthogonal projection on $\gH$, and let $\mu^{\tau^{-1}}_{P,\Gamma}$ denote the lower symbol of $\Gamma$ on $P\gH$ at scale $\tau^{-1}$, as defined in \eqref{eq:Husimi}. Assume that $J\leq \tau^{1/8}$. Then there exist constants $C(K,R), c(K,R)>0$, depending only on $K$ and $R$, such that
			\begin{align*}
				\int_{\{\|Pu\|_{L^2}^2>R\}}\|Pu\|_{L^2}^{6}\dif \mu^{\tau^{-1}}_{P,\Gamma} (u)\leq C(K,R)  e^{-c(K,R)\tau}.
			\end{align*}
		\end{lemma}
		\begin{proof}
			The idea is based on a generalization of the coherent-state resolution of the identity in \eqref{eq:resolution_coherent2}. For a radial observable depending only on $\|Pu\|_{L^2}^2$, its anti-Wick quantization acts as a scalar on each $n$-particle sector, by Schur’s lemma.  This scalar can be expressed as an expectation over a Gamma random variable and thus computed explicitly. Since the state is supported on the sectors $n\leq \cK^2 \tau$, whereas the tail estimate starts at $R\tau$ with $R>\cK^2$, this gap yields the desired exponential decay.
			
		Let $G$ be a nonnegative radial function on $	P\gH $, and define its anti-Wick quantization by
			\begin{align*}
				\operatorname{OP}_{\rm aw}(G):=\left(\frac{\tau}{\pi}\right)^J\int_{P\gH}G(\|u\|_{L^2}^2)\,|\xi(\sqrt\tau u)\rangle\langle\xi(\sqrt\tau u)|\,\dif u.
			\end{align*}
		 By the coherent-state expansion, the restriction of $\mathrm{OP}_{\rm aw}(G)$ to the $n$-particle sector is
			\begin{align*}
				\mathrm{OP}_{\rm aw} (G) \big|_{(P\gH)^{(n)}} =\left(\frac{\tau}{\pi}\right)^J\int_{P\gH} G(\|u\|_{L^2}^2) e^{-\tau \|u\|_{L^2}^2} \frac{\tau^n}{n!}\,|u^{\otimes n} \rangle\langle u^{\otimes n}|\,\dif u .
			\end{align*}
			Since $G$ is radial, the weight $G(\|u\|_{L^2}^2) e^{-\tau \|u\|_{L^2}^2}$ is invariant under the natural action of the unitary group $\cV(J)$ on $P \gH$. Together with the invariance of Lebesgue measure under unitary changes of variables, this implies that, for every
	 $\cV \in \cV(J)$,
			\begin{align*}
				\mathrm{OP}_{\rm aw} (G) \big|_{(P\gH)^{(n)}} \cV^{\otimes n}=\cV^{\otimes n}	\mathrm{OP}_{\rm aw} (G) \big|_{(P\gH)^{(n)}} .
			\end{align*}
			Moreover, $(P\gH)^{(n)} \simeq \mathrm{Sym}^n(\mathbb C^J)$ is an irreducible representation of $\cV(J)$. Hence, Schur's lemma (see, e.g., \cite[Lemma 1.7]{FH91}) implies that the above operator is a scalar multiple of the identity on $(P\gH)^{(n)}$. Thus, there exists $c_n\in\mathbb C$ such that
			\begin{align*}
				\mathrm{OP}_{\rm aw} (G) \big|_{(P\gH)^{(n)}} =c_n  \mathds{1}_{(P\gH)^{(n)}}.
			\end{align*}
		Taking the trace on both sides yields
			\begin{align*}
				\left(\frac{\tau}{\pi}\right)^J \int_{P\gH} G(\|u\|_{L^2}^2) e^{-\tau \|u\|_{L^2}^2} \frac{(\tau \|u\|_{L^2}^2)^n}{n!} \dif u =	c_n \mathrm{dim} \left( (P\gH)^{(n)} \right) = 	c_n { n+J-1 \choose n}.
			\end{align*}
		After the change of variables $y=\tau\|u\|_{L^2}^2$, the left-hand side becomes
			\begin{align*}
				\frac{1}{(J-1)!n!}\int_0^\infty G(y/\tau) e^{-y} y^{n+J-1} \dif y.
			\end{align*}
			This implies that
			\begin{align*}
				c_n=  \frac{1}{(n+J-1)!}\int_0^\infty G(y/\tau) e^{-y} y^{n+J-1} \dif y = \mathbf{E}[G(\mathbf{Y}/\tau)],
			\end{align*}
			where $\mathbf{Y} \sim \Gamma(n+J,1)$ has the Gamma distribution. We now choose $G(\|u\|_{L^2}^2)=\|u\|_{L^2}^6\mathds{1}_{\{\|u\|^2_{L^2} > R\}}$. Since the $P$-localized state $\Gamma_P$ is supported on $\{\cN_P/\tau \leq K^2\}$, this yields
			\begin{align*}
				\int_{\{\|Pu\|_{L^2}^2>R\}}\|Pu\|_{L^2}^{6}\dif \mu^{\tau^{-1}}_{P,\Gamma}(u)&=\tr_{\gF(P\gH)} \left( 	\mathrm{OP}_{\rm aw} (G) \Gamma_P \right) \leq \sup_{0\leq n\leq K^2\tau} c_n  \tr_{\gF(P\gH)} \left( 	\Gamma_P \right)
				\\ & = \sup_{0\leq n\leq K^2\tau} \mathbf{E} \left( \Big(\frac{\mathbf{Y}}{\tau}\Big)^3 \mathds{1}_{\{\mathbf{Y} > R\tau\}} \right).
			\end{align*}
		It remains to estimate this Gamma tail uniformly for $0\leq n \leq K^2\tau$. For $\theta \in (0,1)$, since $\mathds{1}_{(R\tau,\infty)}(x) \leq e^{\theta(x-R\tau)}$, we have
			\begin{align*}
				\mathbf{E} \left[ \Big(\frac{\mathbf{Y}}{\tau}\Big)^3 \mathds{1}_{\{\mathbf{Y} > R\tau\}} \right] & \leq   e^{-\theta R \tau} 	 \mathbf{E} \left[ \Big(\frac{\mathbf{Y}}{\tau}\Big)^3 e^{ \theta \mathbf{Y} } \right] 
				=  \frac{e^{-\theta R\tau}   }{\tau^3(n+J-1)!}\int_0^\infty e^{-(1-\theta )y} y^{n+J+2} \dif y 
				\\ &= \frac{(1-\theta)^{-n-J-3}e^{-\theta R\tau}   }{\tau^3(n+J-1)!}\int_0^\infty e^{-y} y^{n+J+2} \dif y 
			\\& 	 \leq  (1-\theta)^{-3}  \frac{(n+J+2)^3}{\tau^3} e^{-(n+J)\log (1-\theta)-\theta  R\tau}.
			\end{align*} 
			Since $K^2<R$, we may choose $\theta_0\in(0,1)$ sufficiently small such that $1\leq \frac{-\log(1-\theta_0)}{\theta_0} < \frac{R}{K^2}$. Using $J\leq  \tau^{\frac14}$ and $n\leq K^2 \tau$, we obtain, for all sufficiently large $\tau>0$, 
			\begin{align*}
				 e^{-(n+J)\log (1-\theta_0)-\theta_0  R\tau} \leq e^{-\tau^{1/4}\log (1-\theta_0)-(\theta_0  R+ \log(1-\theta_0)K^2)\tau}\leq e^{-\frac12(\theta_0  R+ \log(1-\theta_0)K^2)\tau},
			\end{align*}
			where in the last inequality we used $\frac{-2\log (1-\theta_0)}{\theta_0  R+ \log(1-\theta_0)K^2} \leq \sqrt{\tau}$ for sufficiently large $\tau>0$. Hence,
			\begin{align*}
				\sup_{0\leq n\leq K^2\tau} \mathbf{E} \left( \Big(\frac{\mathbf{Y}}{\tau}\Big)^3 \mathds{1}_{\{\mathbf{Y} > R\tau\}} \right)  & \leq (1-\theta_0)^{-3} \left(K^2 + \frac{\sqrt{\tau}+2}{\tau} \right)^3 e^{-\frac12(\theta_0  R+ \log(1-\theta_0)K^2)\tau}
			\leq C(K,R)  e^{-c(K,R)\tau}.
			\end{align*}
			This proves the desired estimate.
		\end{proof}
	Applying Lemma~\ref{lemma:interacting:tail} to the Gibbs state $\Gamma_\tau^{f_\eta}$, with $P=\mathds{1}(h\leq \Lambda_e)$ and $R=\cK^2+1$, we conclude that for all sufficiently large $\tau>0$ satisfying $\eps^{-3}\leq \Lambda_e^{\frac18} \leq \tau^{\frac{1}{32}}$,
		\begin{align*}
			\int_{P\gH}  \cW^\eps_P \mathds{1}_{\{\|u\|_{L^2}^2 > \cK^2+1 \}}  \dif \mu_{P,\tau}^{\tau^{-1}}(u)  & \leq \eps^{-2} \|w\|^2_{L^\infty} 	\int_{\{\|Pu\|_{L^2}^2>\cK^2+1\}}\|Pu\|_{L^2}^{6}\dif \mu_{P,\tau}^{\tau^{-1}}(u)   \\&\leq C(\cK,w) \eps^{-2} e^{-c(\cK)\tau} \leq C\Lambda_e^{-\frac18}.
		\end{align*}
			Inserting this bound into \eqref{variation} gives
			\begin{align}\label{lower:holder3}
					-\log  \frac{\cZ^{f_\eta}_{\tau}}{\mathcal{Z}^{f_\eta}_{\tau,0}}  
			 \geq -\log \left( \int_{P\gH} e^{ \cW_P^\eps \mathds{1}_{\{\|u\|_{L^2}^2 \leq \cK^2+1 \}} } \dif   \mu_{P,0}^{\tau^{-1}}(u) \right) -C\Lambda_e^{-\frac18} .
			\end{align}
			
		\textbf{Step 3: Quantitative $L^1$-comparison between $\mu_{P,0}^{\tau^{-1}}$
			and 	$ \mu^{f_\eta}_{0,P}$.}
		To prove \eqref{ineq:lower bound}, it remains to compare $ \int_{P\gH} e^{ \cW_P^\eps \mathds{1}_{\{\|u\|_{L^2}^2 \leq  \cK^2+1 \}} } \dif   \mu_{P,0}^{\tau^{-1}}(u) $ with $\int_{P\gH} e^{ \cW_P^\eps} \dif   \mu^{f_\eta}_{0,P}$.
	Since $\mu_{P,0}^{\tau^{-1}}$ involves a particle number cutoff, whereas $\mu^{f_\eta}_{0,P}$ incorporates a mass cutoff, a more refined comparison is needed. Our first goal is therefore to establish a quantitative $L^1$-estimate between them.
		\begin{lemma}\label{lemma:L1:norm}
			For all sufficiently large $\tau>0$, $\Lambda_e>0$ and $ \eta\in \big(0,\frac{1}{2}\cK^2\big)$ satisfying $\eta^{-1} \leq  \Lambda_e^{\frac18}\leq \tau^{\frac{1}{32}}$, we have
			\begin{align}\label{L1:norm}
				\| \mu_{P,0}^{\tau^{-1}}- \mu^{f_\eta}_{0,P}\|_{L^1(P\gH)}\leq C \Lambda_e^{-\frac14} ,
			\end{align}
			for some constant $C>0$ independent of $\tau$ and $\Lambda_e$.
		\end{lemma}
		\begin{proof}
			The argument is close in spirit to the semi-classical analysis used in Proposition~\ref{prop:upper bound}. We first decompose the lower symbol into a principal part, which gives rise to the projected Gaussian measure, and an error term arising from the replacement of the full cutoff by its projected analogue. We then derive a lower bound on the principal part via a coherent-state decomposition. Finally, a total variation estimate, together with the normalization by the free partition function, completes the proof. We begin by writing the lower symbol $\mu_{P,0}^{\tau^{-1}}$ as
			\begin{align}\label{first way}
				\dif \mu_{P,0}^{\tau^{-1}}(u)&= \left( \frac{\tau}{\pi}\right)^J \left\langle  \xi(\sqrt{\tau}u), (\Gamma^{f_\eta}_{\tau,0})_P \xi(\sqrt{\tau}u) \right\rangle \dif u  \nonumber
				\\&\geq  \left( \frac{\tau}{\pi}\right)^J \mathds{1}_{\{\|Pu\|_{L^2} \leq \cK \}} \left\langle  \xi(\sqrt{\tau}u), \frac{\Gamma_{\tau,0,P}f_\eta\big(\frac{\cN_P}{\tau}\big)}{\tr_{\gF(\gH)} \big(\Gamma_{\tau,0}f_\eta\big(\frac{\cN}{\tau}\big)\big)} \xi(\sqrt{\tau}u)  \right\rangle \dif u
				\\ &- \left( \frac{\tau}{\pi}\right)^J \mathds{1}_{\{\|Pu\|_{L^2} \leq \cK \}} \left\langle  \xi(\sqrt{\tau}u), \frac{\Gamma_{\tau,0,P}f_\eta\big(\frac{\cN_P}{\tau}\big)-\tr_{\gF(Q\gH)}\big( 	\cU \Gamma_{\tau,0}f_\eta\big(\frac{\cN}{\tau}\big)	\cU ^*\big)}{\tr_{\gF(\gH)}\big(\Gamma_{\tau,0}f_\eta \big(\frac{\cN}{\tau}\big)\big)} \xi(\sqrt{\tau}u) \right\rangle \dif u, \nonumber
			\end{align}
			where $\cU$ is the unitary equivalence introduced in \eqref{def:unitary}.
			Since $\cU \cN \cU^*=\cN_P  \otimes \mathds{1}_{\gF (Q\gH)} +  \mathds{1}_{\gF (P \gH)} \otimes \cN_Q $ and the Gaussian state is factorized $\cU \Gamma_{\tau,0} \cU^* =\Gamma_{\tau,0,P} \otimes \Gamma_{\tau,0,Q} $, we compare the second term on the right-hand side of \eqref{first way} as
			\begin{align}\label{compare:partial trace}
				&\Big|\Gamma_{\tau,0,P}f_\eta(\cN_P/\tau)-\tr_{\gF(Q\gH)}\big( 	\cU  \Gamma_{\tau,0}f_\eta(\cN/\tau) 	\cU ^*\big) \Big| \nonumber
				\\ & \quad = \Big|\tr_{\gF(Q\gH)} \big( (\Gamma_{\tau,0,P } \otimes  \Gamma_{\tau,0,Q })( f_\eta(\cN_P/\tau) \otimes \mathds{1}_{\gF (Q\gH)} ) \big)  - \tr_{\gF(Q\gH)}\big( 	\cU  \Gamma_{\tau,0} 	\cU^* 	\cU  f_\eta(\cN/\tau) 	\cU ^*\big)  \Big| \nonumber
				\\ & \quad = \Big|\tr_{\gF(Q\gH)} \Big( \Gamma_{\tau,0,P } \otimes  \Gamma_{\tau,0,Q } \Big( f_\eta\Big(\frac{\cN_P}{\tau}  \otimes \mathds{1}_{\gF (Q\gH)} \Big)  - f_\eta\Big(\frac{\cN_P}{\tau}  \otimes \mathds{1}_{\gF (Q\gH)} +  \mathds{1}_{\gF (P \gH)} \otimes \frac{\cN_Q}{\tau} \Big) \Big) \Big|
				\\& \quad \leq  \frac{1}{\tau}\tr_{\gF(Q\gH)}\left( ( \Gamma_{\tau,0,P } \otimes  \Gamma_{\tau,0,Q }) (  \mathds{1}_{\gF (P \gH)} \otimes \cN_Q ) \right)\|f_\eta'\|_{L^\infty}  \nonumber
				\\ & \quad=  \|f_\eta'\|_{L^\infty}  \tr_{\gF(Q\gH)}\left(   \Gamma_{\tau,0,Q }  \frac{\cN_Q}{\tau} \right) \Gamma_{\tau,0,P }  \leq  C\eta^{-1}\Lambda_e^{-\frac12} \Gamma_{\tau,0,P}, \nonumber
			\end{align}
		where we used \eqref{Wp-W:second} in the last inequality.
		Using the same computation as in \cite[Page 93]{NZZ25}, we obtain
			\begin{equation}\label{same NZZ}
				\begin{aligned}
					& \left(	 \frac{\tau}{\pi}\right)^J\left\langle  \xi(\sqrt{\tau}u), \Gamma_{\tau,0,P} \xi(\sqrt{\tau}u) \right\rangle \dif u 
					=\bigotimes_{j=1}^J \left[ \frac{\tau}{\pi} (1- e^{-\lambda_j/\tau}) \mathrm{exp} \left( -\tau|\alpha_j|^2( 1- e^{-\frac{\lambda_j}{\tau}}) \right)   \right] \dif \alpha_j
					\\ &\quad  \leq \bigotimes_{j=1}^J  \left[ \frac{\lambda_j}{\pi}e^{ -\lambda_j|\alpha_j|^2 +\frac{1}{\tau}\lambda_j^2 |\alpha_j|^2   }   \right] \dif \alpha_j 
					\leq e^{\frac{1}{\tau}\Lambda_e^2\|Pu\|_{L^2}^2} \dif  \mu_{0,P}(u), 
				\end{aligned}
			\end{equation}
			with the complex variable $\alpha_j=\langle u,u_j\rangle$. Here in the last line, we used $\lambda_j/\tau\geq 1-e^{-\lambda_j/\tau} \geq \lambda_j/\tau-\frac12 (\lambda_j/\tau)^2$.
			It follows from \eqref{compare:partial trace} and \eqref{same NZZ} that the second term on the right-hand side of \eqref{first way} is bounded by
			\begin{align}\label{first way4}
				&\frac{C\eta^{-1}\Lambda_e^{-\frac12}}{\tr_{\gF(\gH)} (\Gamma_{\tau,0}f_\eta(\cN/\tau))}   \mathds{1}_{\{\|Pu\|_{L^2} \leq \cK \}}\left( \frac{\tau}{\pi}\right)^J  \left\langle  \xi(\sqrt{\tau}u), \Gamma_{\tau,0,P} \xi (\sqrt{\tau}u) \right\rangle \dif u  \nonumber
				\\ &\leq \frac{C \eta^{-1} \Lambda_e^{-\frac12} \mathds{1}_{\{\|Pu\|_{L^2} \leq \cK \}}  e^{\frac{\Lambda_e^2}{\tau} \|Pu\|^2_{L^2}}   \dif \mu_{0,P}}{\tr_{\gF(\gH)}(\Gamma_{\tau,0}f_\eta(\cN/\tau)) }  
\leq  C\eta^{-1}  \Lambda_e^{-\frac12} e^{\frac{\Lambda_e^2}{\tau} \cK^2} \dif \mu_{0,P}  \leq C \Lambda_e^{-\frac14} \dif \mu_{0,P} ,
			\end{align}
			where we used \eqref{upper_bound_Ztau0} and $\Lambda_e\leq \tau^{\frac14} $ in the last inequality.
			
	We now turn to the first term on the right-hand side of \eqref{first way}. This term is treated exactly as in Step~3 of the proof of Proposition~\ref{prop:upper bound}: the quantum cutoff is replaced by its classical counterpart through a coherent-state decomposition, and the resulting error is negligible. Repeating the computation leading to \eqref{conherent-state:identity} and \eqref{tr:lower_two_part}, we have
			\begin{align}\label{first way1}
				& \frac{\mathds{1}_{\{\|Pu\|_{L^2} \leq \cK \}} }{\tr_{\gF(\gH)}(\Gamma_{\tau,0}f_\eta(\cN/\tau))} \left( \frac{\tau}{\pi}\right)^J \left\langle  \xi (\sqrt{\tau}u), \Gamma_{\tau,0,P}f_\eta(\cN_P/\tau) \xi(\sqrt{\tau}u) \right\rangle \dif u 
				\\ &\geq  \frac{  \prod_{j=1}^J \frac{\tau}{\pi}(1-e^{-\lambda_j/\tau}) }{\tr_{\gF(\gH)}(\Gamma_{\tau,0}f_\eta(\cN/\tau))}  \mathds{1}_{\{\|Pu\|_{L^2} \leq \cK \}}  \mathbf{E} \Big( f_\eta  \Big( \frac{\mathbf{X}}{\tau}\Big)  \mathrm{exp}  \Big( -\frac{\mathbf{X}}{\tau\|Pu\|_{L^2}^2 }\langle u,PhPu \rangle \Big)  \Big)    \dif u , \nonumber
			\end{align}
			where $\mathbf{X}$ has the Poisson distribution with mean $\tau \|Pu\|_{L^2}^2$. By the same calculations as in \eqref{assertion:diff:Hamiltonian} and \eqref{transform:numbercut:to:masscut}, it follows that
			\begin{align}\label{first way2}
			&	\mathds{1}_{\{\|Pu\|_{L^2} \leq \cK \}}  \mathbf{E} \Big( f_\eta  \Big( \frac{\mathbf{X}}{\tau}\Big)  \mathrm{exp}  \Big( -\frac{\mathbf{X}}{\tau\|Pu\|_{L^2}^2 }\langle u,PhPu \rangle \Big)  \Big) 
				\\ &\geq f_\eta\big( \|Pu\|_{L^2} ^2\big)  e^{-\langle u,PhPu \rangle}- (\eta^{-1} +\Lambda_e)\mathds{1}_{\{\|Pu\|_{L^2} \leq \cK \}}  e^{-\langle u,PhPu \rangle}  \mathbf{E} \Big|  \frac{\mathbf{X}}{\tau} - \|Pu\|_{L^2} ^2 \Big|
\nonumber				
\\ &\geq  f_\eta\big( \|Pu\|_{L^2} ^2\big)  e^{-\langle u,PhPu \rangle}-C_\cK\tau^{-\frac12}  \Lambda_e e^{-\langle u,PhPu \rangle} \nonumber .
			\end{align}
			Substituting \eqref{first way2} into \eqref{first way1} implies 
			\begin{align}\label{first:way3}
				& \frac{\mathds{1}_{\{\|Pu\|_{L^2} \leq \cK \}} }{\tr_{\gF(\gH)}(\Gamma_{\tau,0}f_\eta(\cN/\tau))} \left( \frac{\tau}{\pi}\right)^J \left\langle  \xi (\sqrt{\tau}u), \Gamma_{\tau,0,P}f_\eta(\cN_P/\tau) \xi(\sqrt{\tau}u) \right\rangle \dif u \nonumber
				\\ &\quad \geq  \frac{  \prod_{j=1}^J \frac{\tau}{\lambda_j}(1-e^{-\lambda_j/\tau}) }{\tr_{\gF(\gH)}(\Gamma_{\tau,0}f_\eta(\cN/\tau))}    f_\eta\big( \|Pu\|_{L^2} ^2\big)   \prod_{j=1}^J \frac{\lambda_j}{\pi}e^{-\langle u,PhPu \rangle} \dif u\nonumber
				\\ &\qquad -  \frac{ C_\cK\tau^{-\frac12}   \Lambda_e }{\tr_{\gF(\gH)}(\Gamma_{\tau,0}f_\eta(\cN/\tau))}    \prod_{j=1}^J \frac{\tau}{\lambda_j}(1-e^{-\lambda_j/\tau})  \prod_{j=1}^J \frac{\lambda_j}{\pi}e^{-\langle u,PhPu \rangle} \dif u \nonumber
				\\ & \quad \geq \left(1-\Lambda_e^{-1}\right) \frac{f_\eta(\|Pu\|_{L^2}^2 )}{\tr_{\gF(\gH)}(\Gamma_{\tau,0}f_\eta(\cN/\tau))}  \dif \mu_{0,P}(u)  - C_\cK\Lambda_e^{-1} \dif \mu_{0,P}(u), 
			\end{align}
			where we used \eqref{upper_bound_Ztau0}, \eqref{bernoulli:inequality} and $\Lambda_e\leq \tau^{\frac14}$ in the last line.
			Combining \eqref{first way4} and \eqref{first:way3}, we obtain
			\begin{align*}
				\mu_{P,0}^{\tau^{-1}}(u) \geq  \left(1-\Lambda_e^{-1}\right)  \frac{\int_{\gH} f_\eta(\|Pu\|_{L^2}^2) \dif \mu_{0}(u)}{\tr_{\gF(\gH)}(\Gamma_{\tau,0}f_\eta(\cN/\tau))}  \mu^{f_\eta}_{0,P}(u)- C\Lambda_e^{-\frac14}  \mu_{0,P}(u).
			\end{align*}
			Together with \eqref{upper_bound_Ztau0}, \eqref{convergence on same index2:new} and \eqref{diff:Gibbs:cutoff}, this further implies that for $\eta^{-1} \leq  \Lambda_e^{\frac18}\leq  \tau^{\frac{1}{32}}$,
				\begin{align}\label{first way5}
					( \mu_{P,0}^{\tau^{-1}}&-\mu^{f_\eta}_{0,P})_-(u)  \leq \frac{\big|\int_{\gH} f_\eta(\|Pu\|_{L^2}^2) \dif \mu_{0}(u)-\tr_{\gF(\gH)}(\Gamma_{\tau,0}f_\eta(\cN/\tau))\big|}{\tr_{\gF(\gH)}(\Gamma_{\tau,0}f_\eta(\cN/\tau))} \mu^{f_\eta}_{0,P}(u)  \nonumber
					\\ &\qquad \qquad \qquad+\Lambda_e^{-1} \frac{\int_{\gH} f_\eta(\|Pu\|_{L^2}^2) \dif \mu_{0}(u)}{\tr_{\gF(\gH)}(\Gamma_{\tau,0}f_\eta(\cN/\tau))} \mu^{f_\eta}_{0,P}(u)
					+C\Lambda_e^{-\frac14} \mu_{0,P}(u)  \nonumber
					\\ & \leq C \Big|\int_{\gH} f_\eta(\|Pu\|_{L^2}^2) \dif \mu_{0}(u)-\tr_{\gF(P\gH)}(\Gamma_{\tau,0,P}f_\eta(\cN_P/\tau)) \Big| \mu^{f_\eta}_{0,P}(u)  \nonumber
					\\&\quad +C \Big| \tr_{\gF(P\gH)}(\Gamma_{\tau,0,P}f_\eta(\cN_P/\tau))- \tr_{\gF(\gH)}(\Gamma_{\tau,0}f_\eta(\cN/\tau)) \Big| \mu^{f_\eta}_{0,P}(u)   \nonumber
					+ C\Lambda_e^{-1} \mu^{f_\eta}_{0,P}(u)+C\Lambda_e^{-\frac14}\mu_{0,P}(u)  \nonumber
					\\ & \leq C\tau^{-\frac14}  \mu^{f_\eta}_{0,P}(u)+ C\Lambda_e^{-\frac14} \mu^{f_\eta}_{0,P}(u) +C\Lambda_e^{-\frac14} \mu_{0,P}(u), 
				\end{align}
			where $f_{-}:=\mathrm{max}(-f,0)$ is the negative part. Since $\mu_{P,0}^{\tau^{-1}}$ and $\mu^{f_\eta}_{0,P}$ are probability measures on $P\gH$, we have
			\begin{align*}
				0=	\int_{P\gH}	 ( \mu_{P,0}^{\tau^{-1}}-\mu^{f_\eta}_{0,P}  ) =  \int_{P\gH}	 ( \mu_{P,0}^{\tau^{-1}}-\mu^{f_\eta}_{0,P}  )_{+} - \int_{P\gH}	 (\mu_{P,0}^{\tau^{-1}}-\mu^{f_\eta}_{0,P}  )_{-}  .
			\end{align*}
			We therefore deduce from \eqref{first way5} that
			\begin{equation*}
				\int_{P\gH}	| \mu_{P,0}^{\tau^{-1}}-\mu^{f_\eta}_{0,P}  | = 2 \int_{P\gH}	 ( \mu_{P,0}^{\tau^{-1}}-\mu^{f_\eta}_{0,P}  )_{-}   \leq  C\Lambda_e^{-\frac14} +C\tau^{-\frac14}\leq C\Lambda_e^{-\frac14},
			\end{equation*}
			with some constant $C>0$ independent of $\tau$ and $\Lambda_e$.
		\end{proof}
		
		\textbf{Step~4: Exponential moment estimate and conclusion of the lower bound.}
	With the $L^1$-comparison in hand, we control the right-hand side of \eqref{lower:holder3} by
		\begin{equation}\label{lower:holder1}
			\begin{aligned}
				\int_{P\gH} e^{ \cW_P^\eps \mathds{1}_{\{\|u\|_{L^2}^2 \leq \cK^2+ 1\}} } \dif   \mu_{P,0}^{\tau^{-1}}(u) & \leq  \int_{P\gH} e^{ \cW_P^\eps} \dif   \mu^{f_\eta}_{0,P}(u) + \int_{P\gH} e^{ \cW_P^\eps \mathds{1}_{\{\|u\|_{L^2}^2 \leq \cK^2+ 1 \}} } | \dif   \mu_{P,0}^{\tau^{-1}}-  \dif   \mu^{f_\eta}_{0,P}|
				\\ &\leq  \int_{P\gH} e^{ \cW_P^\eps} \dif   \mu^{f_\eta}_{0,P}(u)+ e^{\eps^{-2} \|w\|_{L^\infty}^2(\cK^2+ 1)^3} \|\mu_{P,0}^{\tau^{-1}}-\mu^{f_\eta}_{0,P}\|_{L^1(P\gH)}
				\\ &\leq  \int_{P\gH} e^{ \cW_P^\eps} \dif   \mu^{f_\eta}_{0,P}(u) + e^{M \eps^{-2}}\Lambda_e^{-\frac14}\leq  \int_{P\gH} e^{ \cW_P^\eps} \dif   \mu^{f_\eta}_{0,P}(u) + \Lambda_e^{-\frac18}.
			\end{aligned}
		\end{equation}
	 Using the fact that $\log(1+t)=O(|t|)$ for small $|t|$, together with \eqref{lower:holder3}, we obtain
	\begin{align}\label{conclusion:lower bound}
		-\log  \frac{\cZ^{f_\eta}_{\tau}}{\mathcal{Z}^{f_\eta}_{\tau,0}}   \nonumber
		&\geq -\log \left( \int_{P\gH} e^{\cW_P^\eps \mathds{1}_{\{\|u\|_{L^2}^2 \leq \cK^2+1 \}} } \dif   \mu_{P,0}^{\tau^{-1}}(u) \right)-C\Lambda_e^{-\frac18} 
		\\ &\geq -\log \left( \int_{P\gH} e^{\cW_P^\eps } \dif   \mu^{f_\eta}_{0,P}(u) \right)-\log \left(1+\frac{C \Lambda_e^{-\frac18} }{\int_{P\gH} e^{\cW_P^\eps} \dif   \mu^{f_\eta}_{0,P}(u)}  \right)-C\Lambda_e^{-\frac18} 
		\\ & \geq -\log \left( \int_{P\gH} e^{\cW_P^\eps} \dif   \mu^{f_\eta}_{0,P}(u)\right)- O\left(\Lambda_e^{-\frac18} \right),  \nonumber
	\end{align}
	where we used $\int e^{\cW_P^\eps} \dif   \mu^{f_\eta}_{0,P}(u) \geq 1$ in the last inequality.
	This concludes the proof of Proposition~\ref{prop:lower bound}.
	\end{proof}

		\begin{remark}\label{rmk:lower:exponential}
		Estimate \eqref{lower:holder1} is the only point in the lower-bound argument where an exponential dependence on $\eps^{-2}$
		appears. This loss comes from the tail estimate, which requires large mass parameter $R$ to be strictly bigger than the support parameter $\cK^2$. In the critical regime, there is no room to further adjust $R$ in order to exploit the tail estimate, and the exponential term has to be controlled by a pointwise bound. In the subcritical regime, however, one still has enough room to choose $R<\cK_c^2$, and the classical moment bound can then be used to control this term.
	\end{remark}
	
	\subsection{Proof of Proposition~\ref{Thm:Hartree}}\label{subsec:proof:hartree}
	With Propositions~\ref{prop:upper bound} and \ref{prop:lower bound} in hand, it remains to remove the finite-dimensional localization and recover the full Hartree expression. We begin with a Cauchy-type estimate for the projected classical partition functions, following the strategy of \cite[Theorem 10.2]{NZZ25}, by combining the upper and lower free-energy bounds from the previous two subsections for the same quantum quantity.
	\begin{proposition}\label{pro:cauchy:seq1}
		Let $w$ satisfy Assumption~\ref{assum:on:w}, and let $\cK>0$ be arbitrary. Let $\Lambda_e>0$ be sufficiently large, and let $ \eta\in \big(0,\frac{1}{2}\cK^2\big)$ and $\eps\in(0,1)$ satisfy $\eta^{-1}e^{M\eps^{-2}} \leq \Lambda_e^{\frac18}$, where $M=(2\|w\|^2_{L^\infty}+1)(\cK^2+4)^3$. Let the cutoff $f_\eta$ satisfy Assumption~\ref{assum:cutoff} with cutoff parameter $\cK$, and let $P:=\mathds{1}(h\leq \Lambda_e)$ be the orthogonal projection on $ \gH$. With $ \cW^\eps$ defined in \eqref{Hartree:interaction} and $\cW^\eps_P$ in \eqref{def:widetilde:WP}, we have
		\begin{align}\label{P:to:infty1}
		\left| \log \left( \int_{P\gH} e^{\cW^\eps_P(u)} \dif \mu^{f_\eta}_{0,P} (u) \right)- \log  \left( \int_{\gH} e^{ \cW^\eps(u)} \dif \mu^{f_\eta}_0 (u) \right)\right|\leq O\left( \Lambda_e^{-\frac18} \right).
		\end{align}
		Here $\mu^{f_\eta}_0$ and $\mu^{f_\eta}_{0,P} $ are the Gaussian measures defined in \eqref{def:truncated:Gaussian:measure} and \eqref{truncated Gaussian measure}, respectively, with cutoff $f_\eta$.
	\end{proposition}
	\begin{proof}
		Our starting point is the convergence of the quantum free energy established in the previous subsections. Combining Proposition~\ref{prop:upper bound} and Proposition~\ref{prop:lower bound}, we obtain that, for sufficiently large $\Lambda_e>0$, $ \eta\in \big(0,\frac{1}{2}\cK^2\big)$ and $\eps\in(0,1)$ satisfying $\Lambda_e^{\frac18} \geq \eta^{-1}e^{M\eps^{-2}}$,
		\begin{align}\label{P:to:infty2}
			\left| -\log  \frac{\cZ^{f_\eta}_{\tau}}{\mathcal{Z}^{f_\eta}_{\tau,0}} + \log \left( \int_{P \gH} e^{\cW_P^\eps(u) } \dif \mu^{f_\eta}_{0,P} (u)\right) \right| \leq O\left(  \Lambda_e^{-\frac18} \right),
		\end{align}
		where we have chosen $\tau \to \infty$ sufficiently fast such that $ \Lambda_e \leq \tau^{\frac14}$. Similarly, for any given
		$\overline{\Lambda}_e\geq \Lambda_e$, we also obtain the same bound with $P$ replaced by $\overline{P}=\mathds{1}(h\leq \overline{\Lambda}_e)$, provided that $\tau$ is sufficiently large such that $ \eta^{-1}e^{M\eps^{-2}}\leq \Lambda_e^{\frac18} \leq \overline{\Lambda}_e^{\frac18}  \leq \tau^{\frac{1}{32}}$. More precisely, we may choose the same sufficiently large $\tau$ such that
		\begin{align}\label{P:to:infty3}
			\left| -\log  \frac{\cZ^{f_\eta}_{\tau}}{\mathcal{Z}^{f_\eta}_{\tau,0}} + \log \left( \int_{\overline{P}\gH} e^{\cW_{\overline{P}}^\eps(u) } \dif \mu^{f_\eta}_{0,\overline{P}} (u)\right) \right| \leq O\left(  \overline{\Lambda}_e^{-\frac18} \right) \leq O\left(  \Lambda_e^{-\frac18} \right),
		\end{align}
		and \eqref{P:to:infty2} hold simultaneously.
		Thus by the triangle inequality, we find that for all $\overline{\Lambda}_e \geq \Lambda_e $
		\begin{align}\label{P:to:infty4}
		\left| \log \left( \int_{\overline{P}\gH} e^{\cW_{\overline{P}}^\eps(u) } \dif \mu^{f_\eta}_{0,\overline{P}} (u)\right) - \log \left( \int_{P \gH} e^{\cW_P^\eps(u) } \dif \mu^{f_\eta}_{0,P} (u)\right) \right| \leq O\left(  \Lambda_e^{-\frac18} \right).
		\end{align}
		Recalling the definitions of $\cW^\eps_P$ and $\cW^\eps$ in \eqref{def:widetilde:WP} and \eqref{Hartree:interaction}, it follows that, for any given $\eps>0$
		\begin{align*}
			\left|\cW_P^\eps(u) -\cW^\eps(u) \right| &\leq 	\frac{1}{3!} \left|  \big\langle   u^{\otimes 3}, W^\eps u^{\otimes 3}\big\rangle  -  \big\langle   P^{\otimes 3} u^{\otimes 3}, W^\eps P^{\otimes 3} u^{\otimes 3}\big\rangle  \right| 
			\\& \leq \frac{1}{3!} \|w^\eps\|_{L^\infty}^2\left( \|Pu\|_{L^2}^4+ \|Pu\|_{L^2}^2 \|u\|_{L^2}^2+ \|u\|_{L^2}^4 \right) \big( \|Pu\|_{L^2}+ \|u\|_{L^2}\big)  \|u-Pu\|_{L^2}
			\\ &\leq  \|w^\eps\|_{L^\infty}^2 \|u\|_{L^2}^5\|u-Pu\|_{L^2} .
		\end{align*}
		Since $\|w^\eps\|_{L^\infty}<\infty$ and $u\in L^2$ $\mu_0$-almost surely, it follows that $\cW_P^\eps(u)\to \cW^\eps(u)$ as $\Lambda_e \to \infty$, $\mu_0$-almost surely.
		By the continuity of the exponential function and $f_\eta$, we obtain
		\begin{align}\label{P:to:infty5}
			e^{\cW_P^\eps(u) } f_\eta(\|Pu\|_{L^2}^2) \xrightarrow{\Lambda_e \to \infty} e^{\cW^\eps(u) } f_\eta(\|u\|_{L^2}^2),
		\end{align}
		$\mu_0$-almost surely. Furthermore, for any given $\eps>0$, we have
		\begin{align}\label{P:to:infty6}
			\sup_{\Lambda_e \geq 1} e^{\cW_P^\eps(u) } f_\eta(\|Pu\|_{L^2}^2) \leq e^{\eps^{-2}\|w\|_{L^\infty}^2 \cK^6} \|f_\eta\|_{L^\infty} <\infty.
		\end{align}
		Combining \eqref{P:to:infty5} with \eqref{P:to:infty6}, then the dominated convergence theorem implies that
		\begin{align}\label{P:to:infty7}
			\lim_{\Lambda_e \to \infty}  \int_{P\gH} e^{\cW_P^\eps(u) } \dif \mu^{f_\eta}_{0,P} (u) =	\lim_{\Lambda_e \to \infty} \frac{ \int_{\gH} e^{\cW_P^\eps(u) } f_\eta(\|Pu\|_{L^2}^2) \dif \mu_0 (u)}{\int_{\gH} f_\eta(\|Pu\|_{L^2}^2) \dif \mu_0 (u)  } = \int_{\gH} e^{\cW^\eps(u) } \dif \mu^{f_\eta}_{0} (u).
		\end{align}
		Incorporating \eqref{P:to:infty7} with \eqref{P:to:infty4} and letting $\overline{\Lambda}_e \to \infty$, we obtain \eqref{P:to:infty1}.
	\end{proof}
	
	With the above preparations, we are now in a position to conclude \eqref{key:bound1} and \eqref{key:bound2}.
	\begin{proof}[Proof of \eqref{key:bound1}]
		The conclusion follows from a combination of Propositions~\ref{prop:upper bound}, \ref{prop:lower bound} and \ref{pro:cauchy:seq1}. We first verify that, under the assumptions of Proposition~\ref{Thm:Hartree}, the parameter constraints required in the above propositions can indeed be satisfied. The condition $1>\eps \geq M \left(\log  \tau \right)^{-\frac12}$ implies that $e^{64M\eps^{-2}}  \leq   \tau$.
		Consequently, for $\eta \in \big[\tau^{-\frac{1}{64}}, \frac12 \cK^2\big)$, we may choose $\Lambda_e>0$ sufficiently large such that 
		\begin{align} \label{relation:para}
		\eta^{-1}e^{M\eps^{-2}} \leq \Lambda_e^{\frac18}\leq  \tau^{\frac{1}{32}} .
		\end{align}
		Therefore, combining \eqref{P:to:infty1} and \eqref{P:to:infty2}, it follows that
		\begin{align*}
			&	\left| -\log  \frac{\cZ^{f_\eta}_{\tau}}{\mathcal{Z}^{f_\eta}_{\tau,0}} + \log \left( \int_{ \gH} e^{\cW^\eps(u) } \dif \mu^{f_\eta}_0 (u)\right) \right| 
			\leq 	\left| -\log  \frac{\cZ^{f_\eta}_{\tau}}{\mathcal{Z}^{f_\eta}_{\tau,0}} + \log \left( \int_{P \gH} e^{\cW_P^\eps(u) } \dif \mu^{f_\eta}_{0,P} (u)\right) \right|  
			\\ &\qquad \qquad  \qquad  \qquad + 	\left|  \log \left( \int_{P \gH} e^{\cW_P^\eps(u) } \dif \mu^{f_\eta}_{0,P} (u)\right) -\log \left( \int_{ \gH} e^{\cW^\eps(u) } \dif \mu^{f_\eta}_0 (u)\right)  \right| \leq C \Lambda_e^{-\frac18} .
		\end{align*} 
		By the elementary inequality $|e^x-e^y| \leq e^y(e^{|x-y|}-1)$, $x,y \in \R$, we obtain
		\begin{align*}
			\left|\frac{\cZ^{f_\eta}_{\tau}}{\mathcal{Z}^{f_\eta}_{\tau,0}} -\int_{ \gH} e^{\cW^\eps(u) } \dif \mu^{f_\eta}_0 (u)  \right| & \leq  \int_{ \gH} e^{\cW^\eps(u) } \dif \mu^{f_\eta}_0 (u)  \big(e^{C\Lambda_e^{-\frac18}}-1\big) \leq Ce^{\|w\|_{L^\infty}^2\cK^6\eps^{-2}} \Lambda_e^{-\frac18}
			\\ &\leq Ce^{\frac12 M\eps^{-2}}  \Lambda_e^{-\frac18} \leq C \Lambda_e^{-\frac{1}{16}},
		\end{align*}
		where we used \eqref{relation:para} in the last inequality. Letting $\tau \to \infty$ and $\Lambda_e \to \infty$ completes the proof.
	\end{proof}
	\begin{proof}[Proof of \eqref{key:bound2}]
	The proof consists of three steps. We first combine the quantitative de Finetti estimate with the $L^1$-control of the lower symbols to obtain trace-norm convergence after localization to $P\gH$. We then remove the projection on the quantum side using the one-body tail estimate from Lemma~\ref{errors:localization_interaction}. On the classical side, we compare the two partition functions directly via the pointwise convergence $\cW_P^\eps(u)\to \cW^\eps(u)$. Finally, we combine the localized convergence with these approximation estimates to obtain the full trace-norm statement. Throughout the proof, we assume that the parameters satisfy \eqref{relation:para},
	so that Propositions~\ref{prop:upper bound}, \ref{prop:lower bound} and \ref{pro:cauchy:seq1} apply directly.
		
	\textbf{Step 1: convergence after localization.} Let $P=\mathds{1}(h\leq \Lambda_e)$ and $J=\tr(P)$. Let $\cW_P^\eps$ and $\mu^{f_\eta}_{0,P}$ be defined as in \eqref{def:widetilde:WP} and \eqref{truncated Gaussian measure}, respectively. We introduce the probability measure $\mu^{\eps,f_\eta}_P$ on $P\gH$ by
	\begin{align*}
		\dif \mu^{\eps,f_\eta}_P(u):=\frac{e^{\cW_P^\eps(u)} \dif   \mu^{f_\eta}_{0,P}(u)}{\int e^{\cW_P^\eps(v)} \dif   \mu^{f_\eta}_{0,P}(v)}.
	\end{align*} 
	We also define $\big((\Gamma_{\tau}^{f_\eta})_P\big)^{(k)}=P^{\otimes k}(\Gamma^{f_\eta}_{\tau} \big)^{(k)} P^{\otimes k}$. In this step, we prove that, for all sufficiently large $\tau$ and $\Lambda_e$ satisfying \eqref{relation:para},
	\begin{align}\label{density matrix convergence}
		\left\| \frac{k!}{\tau^k}\big((\Gamma_{\tau}^{f_\eta})_P\big)^{(k)} - \int_{P\gH} |u^{\otimes k}\rangle  \langle u^{\otimes k} | \dif \mu^{\eps,f_\eta}_P(u) \right\| _{\gS^1(\gH^{(k)})}  \leq O\big(\Lambda_e^{-\frac{1}{32}}\big).
	\end{align}
To this end, let $\mu_{P,\tau}^{\tau^{-1}}$ denote the lower symbol associated with $\Gamma^{f_\eta}_\tau$ on $P\gH$. 
		By the quantitative de Finetti estimate \eqref{de finetti estimate}, we have for every $k\geq 1$,
		\begin{align}\label{de finetti error trace norm localized}
			 &	\left\| 	\frac{k!}{\tau^k}\big((\Gamma_{\tau}^{f_\eta})_P\big)^{(k)} -\int_{P\gH} |u^{\otimes k}\rangle  \langle u^{\otimes k} | \dif \mu^{\tau^{-1}}_{P,\tau}(u)  \right\|_{\gS^1(\gH^{(k)})}  \nonumber
			 	\\ &\quad  \leq \frac{1}{\tau^k}\sum_{\ell=0}^{k-1}{k\choose \ell}^2\frac{(k-\ell+J-1)!}{(J-1)!}\tr_{\gF(P\gH)}\big(\cN_P^\ell (\Gamma^{f_\eta}_{\tau})_P\big)
			 	\\ &\quad \lesssim_{k} \sum_{\ell=0}^{k-1} \left(\frac{J}{\tau}\right)^{k-\ell}   \tr_{\gF(\gH)}\left(\left(\frac{\cN}{\tau}\right)^\ell \Gamma^{f_\eta}_{\tau}\right) \lesssim_{k}   \sum_{\ell=0}^{k-1} \left(\frac{J}{\tau}\right)^{k-\ell}  \cK^{2\ell}  \leq C_{k,\cK}\Lambda_e^{-1}, \nonumber
		\end{align}
		where in the inequality we used $J=\tr(P)\leq \sqrt{\Lambda_e} \leq  \tau^{\frac18}$.
		
		It therefore remains to prove that, for all sufficiently large $\tau$ and $\Lambda_e$ satisfying \eqref{relation:para},
		\begin{align}\label{density matrix convergence3}
			\left\| \int_{P\gH} |u^{\otimes k}\rangle  \langle u^{\otimes k} | \dif \mu^{\tau^{-1}}_{P,\tau}(u)- \int_{P\gH} |u^{\otimes k}\rangle  \langle u^{\otimes k} | \dif \mu^{\eps,f_\eta}_P(u) \right\| _{\gS^1(\gH^{(k)})}  \leq O \left( \Lambda_e^{-\frac{1}{32}}\right).
		\end{align}
		Since $|u^{\otimes k}\rangle\langle u^{\otimes k}|$ is a positive rank-one operator with trace $\|u\|_{L^2}^{2k}$, we apply Cauchy--Schwarz inequality to obtain
			\begin{align}\label{weighted tv operator bound}
			&	\left\| \int_{P\gH} |u^{\otimes k}\rangle  \langle u^{\otimes k} | \dif \mu^{\tau^{-1}}_{P,\tau}(u)- \int_{P\gH} |u^{\otimes k}\rangle  \langle u^{\otimes k} | \dif \mu^{\eps,f_\eta}_P(u) \right\| _{\gS^1(\gH^{(k)})} \nonumber 
			\\ &\quad \leq \int_{P\gH}\|u\|_{L^2}^{2k}\,\big|\dif \mu^{\tau^{-1}}_{P,\tau}-\dif \mu^{\eps,f_\eta}_P\big|
			\\ &\quad \leq \left(\int_{P\gH}\|u\|_{L^2}^{4k}\,\dif \mu^{\tau^{-1}}_{P,\tau}(u)
			+\int_{P\gH}\|u\|_{L^2}^{4k}\,\dif \mu^{\eps,f_\eta}_P(u)\right)^{1/2}
			\|\mu^{\tau^{-1}}_{P,\tau}-\mu^{\eps,f_\eta}_P\|_{L^1(P\gH)}^{1/2}. \nonumber 
		\end{align}  
	We next bound the moments appearing above. By the definition of $\mu^{\eps,f_\eta}_P$,
		\begin{align}\label{density matrix convergence35}
			\int_{P\gH}\|u\|_{L^2}^{4k}\,\dif \mu^{\eps,f_\eta}_P(u)
			=\frac{\int_{\gH}\|Pu\|_{L^2}^{4k} e^{\cW_P^\eps(u)}f_\eta(\|Pu\|_{L^2}^2)\,\dif \mu_0(u)}{\int_{\gH} e^{\cW_P^\eps(u)}f_\eta(\|Pu\|_{L^2}^2)\,\dif \mu_0(u)}
			\leq \cK^{4k}.
		\end{align}
		For the moment associated with the lower symbol, applying \eqref{de finetti error trace norm localized} with $2k$ in place of $k$, we obtain for all sufficiently large $\tau$ and $\Lambda_e$ satisfying \eqref{relation:para},
			\begin{align}\label{density matrix convergence36}
				\int_{P\gH} &\|u\|_{L^2}^{4k} \dif \mu_{P,\tau}^{\tau^{-1}}(u)  =\tr\left(\int_{P\gH}|u^{\otimes 2k}\rangle\langle u^{\otimes 2k}|\,\dif \mu^{\tau^{-1}}_{P,\tau}(u)\right) \nonumber
				\\ &\leq \frac{(2k)!}{\tau^{2k}} \tr \left( \big((\Gamma_{\tau}^{f_\eta})_P\big)^{(2k)} \right) +	\left\| 	\frac{(2k)!}{\tau^{2k}} \big((\Gamma_{\tau}^{f_\eta})_P\big)^{(2k)} -\int_{P \gH} |u^{\otimes {2k}}\rangle  \langle u^{\otimes {2k}} | \dif \mu^{\tau^{-1}}_{P,\tau}(u)  \right\|_{\gS^1(\gH^{(2k)})}
				\\ &\leq  \tr_{\gF(\gH)} \left( \left( \frac{\cN}{\tau}\right)^{2k} \Gamma_{\tau}^{f_\eta} \right) + C_{k,\cK} \leq \cK^{4k} +C_{k,\cK} . \nonumber
			\end{align}
It thus suffices to derive an $L^1$-comparison between $\mu_{P,\tau}^{\tau^{-1}}$ and $\mu^{\eps,f_\eta}_P$. We assert that, for sufficiently large $\tau$ and $\Lambda_e$ satisfying \eqref{relation:para},
		\begin{align}\label{density matrix convergence30}
			\| \mu_{P,\tau}^{\tau^{-1}}-  \mu^{\eps,f_\eta}_P \|_{L^1(P\gH)} \leq C  \Lambda_e^{-\frac{1}{16}} .
		\end{align}
		Using the variational principle \eqref{variational principle quantum}, the Berezin-Lieb inequality \eqref{eq:Berezin-Lieb}, and invoking again \eqref{variation} with \eqref{conclusion:lower bound}, we obtain that for $R=\cK^2+1$,
			\begin{align}\label{density matrix convergence31}
				-\log  \frac{\cZ^{f_\eta}_{\tau}}{\mathcal{Z}^{f_\eta}_{\tau,0}}  &		\geq \cH_{\mathrm{cl}} (\mu_{P,\tau}^{\tau^{-1}}, \mu_{P,0}^{\tau^{-1}})-  \int_{P\gH}  \cW^\eps_P \mathds{1}_{\{\|u\|_{L^2}^2 \leq R \}}  \dif \mu_{P,\tau}^{\tau^{-1}}(u)-C\Lambda_e^{-\frac18}  \nonumber
				\\ & = \cH_{\mathrm{cl}} (\mu_{P,\tau}^{\tau^{-1}}, \mu')
				-\log \left( \int_{P\gH} e^{\cW_P^\eps (u) \mathds{1}_{\{\|u\|_{L^2}^2 \leq R \}}   } \dif   \mu_{P,0}^{\tau^{-1}}(u)\right)-C\Lambda_e^{-\frac18} 
				\\ & \geq  \cH_{\mathrm{cl}} (\mu_{P,\tau}^{\tau^{-1}}, \mu')-\log \left( \int_{P\gH} e^{\cW_P^\eps (u)} \dif  \mu^{f_\eta}_{0,P}(u)\right)- C \Lambda_e^{-\frac18} , \nonumber
			\end{align}
		where we used \eqref{conclusion:lower bound} in the last inequality and defined
		\begin{align*}
			\dif  \mu'(u):=\frac{e^{\cW_P^\eps(u)\mathds{1}_{\{\|u\|_{L^2}^2 \leq R \}}  } \dif   \mu_{P,0}^{\tau^{-1}}(u)}{\int_{P\gH} e^{\cW_P^\eps(v)\mathds{1}_{\{\|v\|_{L^2}^2 \leq R \}}  } \dif   \mu_{P,0}^{\tau^{-1}}(v)}.
		\end{align*}
		By combining \eqref{density matrix convergence31} with the free energy upper bound \eqref{ineq:upper bound}, and using Pinsker's inequality (see, e.g., \cite[Sec. 2]{CL14}), we obtain, for sufficiently large $\tau>0$ and $\Lambda_e>0$,
		\begin{align}\label{density matrix convergence32}
			\| \mu_{P,\tau}^{\tau^{-1}}-   \mu'  \|_{L^1(P\gH)} \leq \sqrt{2}  \cH_{\mathrm{cl}} (\mu_{P,\tau}^{\tau^{-1}}, \mu')^{\frac12} \leq C  \Lambda_e^{-\frac{1}{16}} .
		\end{align}
		Next, we compare $\mu'$ with $ \mu^{\eps,f_\eta}_P$:
			\begin{align}\label{density matrix convergence33}
					\| \mu' - \mu^{\eps,f_\eta}_P   \|_{L^1(P\gH)} & \leq \frac{1}{\int_{P\gH} e^{\cW_P^\eps \mathds{1}_{\{\|v\|_{L^2}^2 \leq R \}} } \dif   \mu_{P,0}^{\tau^{-1}}(v)} \left\|  e^{\cW_P^\eps \mathds{1}_{\{\|u\|_{L^2}^2 \leq R \}} }    \dif   \mu_{P,0}^{\tau^{-1}}-e^{\cW_P^\eps \mathds{1}_{\{\|u\|_{L^2}^2 \leq R \}} }   \dif \mu^{f_\eta}_{0,P} \right\|_{L^1(P\gH)}   \nonumber
				\\ & +\left| \frac{1}{\int_{P\gH} e^{\cW_P^\eps \mathds{1}_{\{\|v\|_{L^2}^2 \leq R \}} } \dif   \mu_{P,0}^{\tau^{-1}}(v)} - \frac{1}{\int_{P\gH} e^{\cW_P^\eps(v)} \dif   \mu^{f_\eta}_{0,P} (v)} \right| \int_{P\gH} e^{\cW_P^\eps(v)\mathds{1}_{\{\|v\|_{L^2}^2 \leq R \}} } \dif   \mu^{f_\eta}_{0,P} (v) \nonumber
				\\ & +\frac{1}{\int_{P\gH} e^{\cW_P^\eps(v)} \dif   \mu^{f_\eta}_{0,P} (v)} \int_{P\gH} \big| e^{\cW_P^\eps(v)\mathds{1}_{\{\|v\|_{L^2}^2 \leq R \}} } - e^{\cW_P^\eps(v) } \big| \dif   \mu^{f_\eta}_{0,P} (v)  .
					\end{align}
					The last term on the right-hand side vanishes because $  \mu^{f_\eta}_{0,P}$ is supported on $\{\|v\|_{L^2}^2 \leq \cK^2 \}$ and $R>\cK^2$. By $\int_{P\gH} e^{\cW_P^\eps \mathds{1}_{\{\|v\|_{L^2}^2 \leq R \}} } \dif   \mu_{P,0}^{\tau^{-1}}(v)\geq 1$ and the same calculations as in \eqref{lower:holder1}, the first term is controlled by
					\begin{align*}
						 \int_{P\gH}  e^{\cW_P^\eps \mathds{1}_{\{\|u\|_{L^2}^2 \leq R \}} }  \big|   \dif   \mu_{P,0}^{\tau^{-1}} -  \dif \mu^{f_\eta}_{0,P} \big| \leq e^{\eps^{-2} \|w\|_{L^\infty}^2R^3} \|\mu_{P,0}^{\tau^{-1}}-\mu^{f_\eta}_{0,P}\|_{L^1(P\gH)} \leq \Lambda_e^{-\frac18}.
					\end{align*}
					For the same reason, the second term on the right-hand side of \eqref{density matrix convergence33} is bounded by
	\begin{align*}
&	\left| 	\int_{P\gH} e^{\cW_P^\eps \mathds{1}_{\{\|v\|_{L^2}^2 \leq R \}} } \dif   \mu_{P,0}^{\tau^{-1}}(v)- 	\int_{P\gH} e^{\cW_P^\eps  } \dif   \mu^{f_\eta}_{0,P}(v) \right| 
	\\ &\leq  \int_{P\gH}  e^{\cW_P^\eps \mathds{1}_{\{\|u\|_{L^2}^2 \leq R \}} }  \big|   \dif   \mu_{P,0}^{\tau^{-1}} -  \dif \mu^{f_\eta}_{0,P} \big|+ \int_{P\gH} \big| e^{\cW_P^\eps(v)\mathds{1}_{\{\|v\|_{L^2}^2 \leq R \}} } - e^{\cW_P^\eps(v) } \big| \dif   \mu^{f_\eta}_{0,P} (v)  \leq \Lambda_e^{-\frac18}.
	\end{align*}					
		Collecting the above bounds together with \eqref{density matrix convergence32} and \eqref{density matrix convergence33}, the bound \eqref{density matrix convergence30} follows. Inserting \eqref{density matrix convergence35}--\eqref{density matrix convergence30} into \eqref{weighted tv operator bound}, we obtain \eqref{density matrix convergence3}.
		
\textbf{Step 2: removal of the projection on the quantum side.} We claim that, for all sufficiently large $\tau$ and $\Lambda_e$ satisfying \eqref{relation:para}, one has
\begin{align}\label{quantum tail claim trace norm}
		\left\| \frac{k!}{\tau^k} (\Gamma_{\tau}^{f_\eta})^{(k)} - \frac{k!}{\tau^k}\big((\Gamma_{\tau}^{f_\eta})_P\big)^{(k)}  \right\| _{\gS^1(\gH^{(k)})}  \leq O\big(\Lambda_e^{-\frac{1}{16}}\big).
\end{align}
Indeed, since $\big((\Gamma_{\tau}^{f_\eta})_P\big)^{(k)}=P^{\otimes k}( \Gamma^{f_\eta}_{\tau} \big)^{(k)} P^{\otimes k}$, the same argument as in \eqref{CS:ineq}--\eqref{CS:ineq4} yields
\begin{align*}
	&	\left\| \frac{k!}{\tau^k} (\Gamma_{\tau}^{f_\eta})^{(k)} - \frac{k!}{\tau^k}\big((\Gamma_{\tau}^{f_\eta})_P\big)^{(k)}  \right\| _{\gS^1(\gH^{(k)})} \leq \tr \left( (	\mathds{1}^{\otimes k}-P^{\otimes k}) \frac{k!}{\tau^k} (\Gamma_{\tau}^{f_\eta})^{(k)} \right)
		\\ &\qquad \qquad \qquad \qquad\qquad \qquad \qquad\qquad \qquad  +  2\sqrt{ \tr \left( (	\mathds{1}^{\otimes k}-P^{\otimes k}) \frac{k!}{\tau^k} (\Gamma_{\tau}^{f_\eta})^{(k)} \right) \tr \left( \frac{k!}{\tau^k} (\Gamma_{\tau}^{f_\eta})^{(k)} \right)}
		\\ &\qquad \qquad \leq  \tr \left( (	\mathds{1}^{\otimes k}-P^{\otimes k}) \frac{k!}{\tau^k} (\Gamma_{\tau}^{f_\eta})^{(k)} \right)+ 2\cK^k \sqrt{ \tr \left( (	\mathds{1}^{\otimes k}-P^{\otimes k}) \frac{k!}{\tau^k} (\Gamma_{\tau}^{f_\eta})^{(k)} \right) }.
\end{align*}
		It therefore remains to estimate the first trace term. To this end, we use the elementary inequality
		\begin{align*}
			\mathds{1}^{\otimes k}-P^{\otimes k}\leq \sum_{j=1}^k \mathds{1}^{\otimes(j-1)}\otimes Q\otimes \mathds{1}^{\otimes(k-j)},
		\end{align*} 
		and then argue exactly as in the proof of \eqref{WP-WGammatau}, with the $3$-body density matrix replaced by the $k$-body one. This yields
		\begin{align}\label{first quantum tail estimate}
			 \tr & \left( (	\mathds{1}^{\otimes k}-P^{\otimes k}) \frac{k!}{\tau^k} (\Gamma_{\tau}^{f_\eta})^{(k)} \right) \leq  \frac{k!}{\tau^k} \tr\left[\left(\sum_{j=1}^k \mathds{1}^{\otimes(j-1)}\otimes Q\otimes \mathds{1}^{\otimes(k-j)}\right)(\Gamma^{f_\eta}_{\tau})^{(k)}\right] \nonumber
				\\&	\lesssim_k \frac{1}{\tau} \tr_{\gF(\gH)} \left( \cU^* \big(\mathds{1}_{\gF(P\gH)} \otimes  \cN_Q \big) \cU \left( \frac{\cN}{\tau} \right)^{k-1} \Gamma^{f_\eta}_{\tau} \right)
		\\ &	\lesssim_k \frac{1}{\tau}  \cK^{2k-2}   \tr_{\gF(\gH)}  \left( \cU^* \big(\mathds{1}_{\gF(P\gH)} \otimes  \cN_Q \big) \cU \Gamma^{f_\eta}_{\tau} \right)  
		=  \frac{ C_k }{\tau}  \cK^{2k-2}  \tr \left( Q (\Gamma^{f_\eta}_{\tau})^{(1)}  \right)   \nonumber
			\\ &\lesssim_{k,\cK} \frac{1}{\tau}  \left| \tr \left( Q \big((\Gamma^{f_\eta}_{\tau})^{(1)}- \Gamma_{\tau,0}^{(1)} \big) \right) \right|+ \frac{1}{\tau}   \tr \left( Q\Gamma_{\tau,0}^{(1)} \right)  \leq C \Lambda_e^{-\frac18}, \nonumber
		\end{align}
where the last inequality follows from the same one-body tail estimates as those used in \eqref{Wp-W:second}--\eqref{control_relative_entropy:11}. This proves \eqref{quantum tail claim trace norm}.
		
	\textbf{Step 3: quantitative comparison between the finite-dimensional Hartree operator and the full Hartree operator.}
		Following the strategy used in the proof of Proposition~\ref{pro:cauchy:seq1}, we claim that, for all sufficiently large $\tau$ and $\Lambda_e$ satisfying \eqref{relation:para},
		\begin{align}\label{overline:mu:P:to:infty}
			\left\| \int_{P\gH} |u^{\otimes k}\rangle  \langle u^{\otimes k} | \dif \mu^{\eps,f_\eta}_P(u) -  \int_{\gH} |u^{\otimes k}\rangle  \langle u^{\otimes k} | \dif \mu^{\eps,f_\eta}(u)  \right\| _{\gS^1(\gH^{(k)})}  \leq O\big(\Lambda_e^{-\frac{1}{32}}\big).
		\end{align}
	Starting from \eqref{density matrix convergence} and \eqref{quantum tail claim trace norm}, it follows that
		\begin{align}\label{P:to:infty:new1}
		\left\| \frac{k!}{\tau^k} (\Gamma_{\tau}^{f_\eta})^{(k)} - \int_{P\gH} |u^{\otimes k}\rangle  \langle u^{\otimes k} | \dif \mu^{\eps,f_\eta}_P(u) \right\| _{\gS^1(\gH^{(k)})}  \leq O\big(\Lambda_e^{-\frac{1}{32}}\big) ,
		\end{align}
		provided that $\tau \to \infty$ sufficiently fast satisfying $   \Lambda_e\leq \tau^{\frac14}$. For any given $\overline{\Lambda}_e\geq \Lambda_e$ and $\overline{P}=\mathds{1}(h\leq \overline{\Lambda}_e)$, by the argument as in \eqref{P:to:infty2}--\eqref{P:to:infty4}, we may choose the same sufficiently large $\tau$ such that
		\begin{align*}
			\left\| \frac{k!}{\tau^k} (\Gamma_{\tau}^{f_\eta})^{(k)} - \int_{\overline{P}\gH} |u^{\otimes k}\rangle  \langle u^{\otimes k} | \dif \mu^{\eps,f_\eta}_{\overline{P}}(u) \right\| _{\gS^1(\gH^{(k)})} \leq O \left( \overline{\Lambda}_e^{-\frac{1}{32}}\right),
		\end{align*}
		and \eqref{P:to:infty:new1} hold simultaneously. Thus by the triangle inequality, we find that for all $\overline{\Lambda}_e\geq \Lambda_e\to \infty$,
		\begin{align}\label{P:to:infty:new}
			\left\| \int_{P\gH} |u^{\otimes k}\rangle  \langle u^{\otimes k} | \dif \mu^{\eps,f_\eta}_P(u)- \int_{\overline{P}\gH} |u^{\otimes k}\rangle  \langle u^{\otimes k} | \dif \mu^{\eps,f_\eta}_{\overline{P}}(u) \right\| _{\gS^1(\gH^{(k)})} \leq O \left( \Lambda_e^{-\frac{1}{32}}\right).
		\end{align}
		Thus $ \int_{P\gH} |u^{\otimes k}\rangle  \langle u^{\otimes k} | \dif \mu^{\eps,f_\eta}_P(u)$ is a Cauchy family in $\mathfrak{S}^1(\gH^{(k)})$. Then, by the definitions of $\mu^{\eps,f_\eta}_P$ and $\mu^{\eps,f_\eta}$, we obtain
			\begin{align}\label{P:to:infty:newS1}
			&\left\| \int_{P\gH} |u^{\otimes k}\rangle  \langle u^{\otimes k} | \dif \mu^{\eps,f_\eta}_P(u)- \int_{\gH} |u^{\otimes k}\rangle  \langle u^{\otimes k} | \dif \mu^{\eps,f_\eta}(u) \right\| _{\gS^1(\gH^{(k)})}
			\\ &\leq \left| \frac{1}{\int_{\gH} e^{\cW_P^\eps(u)}f_\eta(\|Pu\|_{L^2}^2) \dif \mu_0(u)} -  \frac{1}{\int_{\gH} e^{\cW^\eps(u)}f_\eta(\|u\|_{L^2}^2) \dif \mu_0(u)} \right| \int_{\gH} \|Pu\|_{L^2}^{2k} e^{\cW_P^\eps(u)}f_\eta(\|Pu\|_{L^2}^2) \dif \mu_0(u) \nonumber
			 \\ & \quad +\frac{ \left\| \int_{\gH } \big( |(Pu)^{\otimes k}\rangle  \langle (Pu)^{\otimes k} | e^{\cW_P^\eps(u)}f_\eta(\|Pu\|_{L^2}^2) -  |u^{\otimes k}\rangle  \langle u^{\otimes k} | e^{\cW^\eps(u)}f_\eta(\|u\|_{L^2}^2)\big)\dif \mu_0(u)   \right\| _{\gS^1(\gH^{(k)})} }{\int_{\gH} e^{\cW^\eps(u)}f_\eta(\|u\|_{L^2}^2) \dif \mu_0(u)}  \nonumber
			 \\ & \leq  \int_{\gH} \left| e^{\cW_P^\eps(u)}f_\eta(\|Pu\|_{L^2}^2)-  e^{\cW^\eps(u)}f_\eta(\|u\|_{L^2}^2) \right|  \dif \mu_0(u) \frac{\int_{\gH} \|Pu\|_{L^2}^{2k} \dif \mu^{\eps,f_\eta}_P(u)}{\int_{\gH} e^{\cW^\eps(u)}f_\eta(\|u\|_{L^2}^2) \dif \mu_0(u)}  \nonumber
			 \\ &  \quad +\frac{\int_{\gH } \big\| |(Pu)^{\otimes k}\rangle  \langle (Pu)^{\otimes k} | e^{\cW_P^\eps(u)}f_\eta(\|Pu\|_{L^2}^2) -  |u^{\otimes k}\rangle  \langle u^{\otimes k} | e^{\cW^\eps(u)}f_\eta(\|u\|_{L^2}^2)  \big\| _{\gS^1(\gH^{(k)})} \dif \mu_0(u)}{\int_{\gH} e^{\cW^\eps(u)}f_\eta(\|u\|_{L^2}^2) \dif \mu_0(u)} . \nonumber
		\end{align}
		To estimate these two terms, we first observe that for any $\eps \in (0,1)$, $\eta\in(0,\frac12 \cK^2)$ and $\Lambda_e>0$,
		\begin{align*}
			\int e^{\cW^\eps} f_\eta(\|u\|_{L^2}^2) \dif   \mu_{0}(u) &\geq \int   f_\eta(\|u\|_{L^2}^2)  \dif   \mu_{0}(u)  \geq  \int  \mathds{1}_{\{\|u\|_{L^2}^2 \leq \cK^2/2\}}  \dif   \mu_{0}(u)>0,
			\\ \int_{\gH} \|Pu\|_{L^2}^{2k} \dif \mu^{\eps,f_\eta}_P(u)& \leq  \cK^{2k}.
		\end{align*}
		Moreover, from \eqref{P:to:infty5} and  \eqref{P:to:infty6}, it holds that 	$\mu_0$-almost surely,
		\begin{align*}
			& \left\| |(Pu)^{\otimes k}\rangle  \langle (Pu)^{\otimes k} | e^{\cW_P^\eps(u)}f_\eta(\|Pu\|_{L^2}^2) -  |u^{\otimes k}\rangle  \langle u^{\otimes k} | e^{\cW^\eps(u)}f_\eta(\|u\|_{L^2}^2)  \right\| _{\gS^1(\gH^{(k)})} 
			 \\ &\quad  \leq 2 \|u\|_{L^2}^{k} \|(Pu)^{\otimes k} -u^{\otimes k}\|_{L^2(\gH^{(k)})} e^{\eps^{-2}\|w\|_{L^\infty}^2\cK^6}\|f_\eta\|_{L^\infty} 
			  \\ & \qquad  +  \|u\|_{L^2}^{2k} \big| e^{\cW_P^\eps(u)}f_\eta(\|Pu\|_{L^2}^2) -e^{\cW^\eps(u)}f_\eta(\|u\|_{L^2}^2)  \big|  \xrightarrow{\Lambda_e \to \infty} 0,
		\end{align*}
		and
		\begin{align*}
		&\left\| |(Pu)^{\otimes k}\rangle  \langle (Pu)^{\otimes k} | e^{\cW_P^\eps(u)}f_\eta(\|Pu\|_{L^2}^2) -  |u^{\otimes k}\rangle  \langle u^{\otimes k} | e^{\cW^\eps(u)}f_\eta(\|u\|_{L^2}^2)  \right\| _{\gS^1(\gH^{(k)})} 
		\\ &\quad 
			\leq 2\cK^{2k}e^{\eps^{-2}\|w\|_{L^\infty}^2\cK^6}\|f_\eta\|_{L^\infty}.
		\end{align*}
	Therefore, by the dominated convergence theorem, the second term on the right-hand side of \eqref{P:to:infty:newS1} converges to zero as $\Lambda_e \to \infty$. The first term can be treated in exactly the same way and also converges to zero. Hence, for every given $\eps\in (0,1)$,
	\begin{align*}
		\int_{P\gH} |u^{\otimes k}\rangle  \langle u^{\otimes k} | \dif \mu^{\eps,f_\eta}_P(u) \xrightarrow{\Lambda_e \to \infty} \int_{\gH} |u^{\otimes k}\rangle  \langle u^{\otimes k} | \dif \mu^{\eps,f_\eta}(u) ,\qquad \mathrm{in} \quad \gS^1(\gH^{(k)}).
	\end{align*}
Together with \eqref{P:to:infty:new}, this proves \eqref{overline:mu:P:to:infty} by letting $\overline{\Lambda}_e \to \infty$.
		Finally, combining \eqref{density matrix convergence}, \eqref{quantum tail claim trace norm}, and \eqref{overline:mu:P:to:infty}, and letting $\tau \to \infty$ and $\Lambda_e \to \infty$, we conclude the proof of \eqref{key:bound2}.
	\end{proof}
	The proof of Proposition~\ref{Thm:Hartree} applies verbatim to any fixed
	non-trivial $g\in C_c^\infty([0,\infty);\R_+)$. We can now conclude the proof of Theorem~\ref{Rmk:Hartree}.
\begin{proof}[Proof of Theorem~\ref{Rmk:Hartree}]
	We first observe that Proposition~\ref{Thm:Hartree} applies directly in the present setting. Indeed, in Theorem~\ref{Rmk:Hartree}, both $\eps>0$ and the cutoff $g\in C_c^\infty([0,\infty);\R_+)$ are fixed, so that the only remaining limit is $\tau \to \infty$. Therefore, the constraints in Proposition~\ref{Thm:Hartree} are automatically satisfied, and we obtain
	\begin{align*}
		\left| \frac{\mathcal{Z}^{g}_{\tau}}{\mathcal{Z}^{g}_{\tau,0}} - \frac{\int_{\gH} e^{\cW^\eps}  g( \|u\|^2_{L^2})\dif \mu_{0}(u)  }{\int_{\gH}  g( \|u\|^2_{L^2}) \dif \mu_{0}(u) } \right| \to 0.
	\end{align*}
	Together with \eqref{convergence:for:any:g}, this implies
	\begin{align*}
		 \frac{\mathcal{Z}^{g}_{\tau}}{\mathcal{Z}_{\tau,0}} = \frac{\mathcal{Z}^{g}_{\tau}}{\mathcal{Z}^{g}_{\tau,0}} \cdot   \tr_{\gF(\gH)} \left( \Gamma_{\tau,0} g\left( \cN/\tau \right) \right)  \xrightarrow{\tau \to \infty}  \int_{\gH} e^{\cW^\eps}  \dif \mu^g_{0}(u)   \cdot \int_{\gH}  g( \|u\|^2_{L^2}) \dif \mu_{0}(u)   = \int_{\gH} e^{\cW^\eps}  g( \|u\|^2_{L^2})\dif \mu_{0}(u)  .
	\end{align*}
	The convergence of the reduced density matrices is obtained exactly as in the proof of \eqref{key:bound2}.
\end{proof}
	
		\section{From quantum model to Hartree measure: mass-subcritical regime}\label{sec:mass:subcritical:regime}		
		This section aims to prove the quantum-to-Hartree convergence in the mass-subcritical regime. The proof follows the same variational strategy as in Section~\ref{sec3:lower:upper:free-energy}, but the subcritical assumption allows one to improve the estimates at the points where Section~\ref{sec3:lower:upper:free-energy} relied on pointwise exponential bounds, see also Remarks~\ref{rmk:upper:exponential}, ~\ref{rmk:upper:exponential2}, and ~\ref{rmk:lower:exponential}. For the lower bound on the relative free energy, the only change concerns the final estimate of the weighted integral against the lower symbol. In the critical regime, this term was controlled by the rough bound in \eqref{lower:holder1}, yielding an exponential loss of $\eps^{-2}$. In the present subcritical regime, this loss can be avoided by adjusting the tail region in Lemma~\ref{lemma:interacting:tail} and using the uniform projected Oh-Sosoe-Tolomeo bound in \cite[Sec. 4]{OST22} instead.
		A similar refinement applies to the upper bound: we improve the estimates in Lemma~\ref{lemma:relative entropy} and in the subsequent semiclassical analysis, in particular \eqref{first_hard_term1} and \eqref{second_easy_term} by replacing the pointwise exponential bounds with uniform moment bounds. The main result of this section is the following proposition.
			\begin{proposition}\label{Pro:Hartree}
			Let $w$ satisfy Assumption~\ref{assum:on:w}, and let $\cK\in(0,\cK_c)$ be arbitrary. For all sufficiently large $\tau>0$, let $\eta \in \big[ \tau^{-\frac{1}{64}} ,\frac{1}{2}\cK^2  \big)$, $\eps\in[   \tau^{-\frac{1}{96}},1)$, and let the cutoff $f_\eta$ satisfy Assumption~\ref{assum:cutoff} with cutoff parameter $\cK$.
			Denote by $\cZ^{f_\eta}_{\tau}$ and $\mathcal{Z}^{f_\eta}_{\tau,0}$ the associated partition functions. Then, as $\tau \to \infty$,
			\begin{align}\label{key:bound:sub1}
				\left| \frac{\cZ^{f_\eta}_{\tau}}{\cZ^{f_\eta}_{\tau,0}} - \frac{\int_{\gH} e^{ \cW^\eps(u) } f_\eta(\|u\|_{L^2}^2) \dif \mu_0(u)}{\int_{\gH} f_\eta(\|u\|_{L^2}^2) \dif \mu_0(u) } \right| \to 0,
			\end{align}
			where $	\cW^\eps$ is the Hartree interaction defined in \eqref{Hartree:interaction}. Moreover, for all $k\geq 1$,
			\begin{align}\label{key:bound:sub2}
				\left\| \frac{k!}{\tau^k}(\Gamma_{\tau}^{f_\eta})^{(k)} - \int_{\gH} |u^{\otimes k}\rangle  \langle u^{\otimes k} | \dif \mu^{\eps,f_\eta}(u) \right\| _{\gS^1(\gH^{(k)})}  \to 0,
			\end{align}
			where $	\mu^{\eps,{f_\eta}}$ is the Hartree measure defined in \eqref{def:Hartree:measure}.
		\end{proposition}
		\begin{proof}
			 The proof is divided into three steps, following the picture introduced in Section~\ref{sec3:lower:upper:free-energy}. Throughout the proof, we fix
	$\cK_s \in (\cK,\cK_c)$ and assume that $\tau$ and $\Lambda_e$ are sufficiently large. We also let $ \eta\in \big(0,\frac{1}{2}\cK^2\big)$ and $ \eps \in(0,1)$ such that
	\begin{align}\label{refined:relation}
		\eta^{-1} \leq \Lambda_e^{\frac{1}{8}}, \qquad  \eps^{-3}\leq \Lambda_e^{\frac{1}{8}}, \qquad  \Lambda_e\leq \tau^{\frac14}.
	\end{align}
	We first establish the lower bound on the free energy.
	
	\textbf{Step 1: Free energy lower bound in the mass-subcritical regime.}
	The key point in the subcritical regime is that the auxiliary mass parameter $\cK_s$ appearing in the tail estimate for the interacting lower symbol may be chosen strictly below the critical mass $\cK_c$. The large-mass contribution is still absorbed by Lemma~\ref{lemma:interacting:tail}, while the contribution from the complementary region can be controlled by the uniform Oh-Sosoe-Tolomeo bound in \cite[Sec. 4]{OST22}, rather than by the rough pointwise exponential estimate used in the critical case. Since the estimate \eqref{variation} from Subsection~\ref{sec:lower bound} is independent of pointwise exponential bounds, it gives, uniformly for all parameters satisfying \eqref{refined:relation},
			 	\begin{align*}
			 	-\log  \frac{\cZ^{f_\eta}_{\tau}}{\mathcal{Z}^{f_\eta}_{\tau,0}}  
			 	&\geq \cH_{\mathrm{cl}} (\mu_{P,\tau}^{\tau^{-1}}, \mu_{P,0}^{\tau^{-1}})-  \int_{P\gH}  \cW^\eps_P   \dif \mu_{P,\tau}^{\tau^{-1}}(u)-C\Lambda_e^{-\frac18} \nonumber
			 	\\ & \geq -\log \left( \int_{P\gH} e^{ \cW_P^\eps \mathds{1}_{\{\|u\|_{L^2}^2 \leq \cK_s ^2 \}} } \dif   \mu_{P,0}^{\tau^{-1}}(u) \right)- \int_{P\gH}  \cW^\eps_P \mathds{1}_{\{\|u\|_{L^2}^2 > \cK_s^2  \}}  \dif \mu_{P,\tau}^{\tau^{-1}}(u) -C\Lambda_e^{-\frac18} .
			 \end{align*}
			 For the mass-tail contribution, we apply Lemma~\ref{lemma:interacting:tail} to the Gibbs state $\Gamma_\tau^{f_\eta}$ with $R=\cK_s^2$. This yields, again uniformly for all parameters satisfying \eqref{refined:relation},
			 	\begin{align*}
			 		\int_{P\gH}  \cW^\eps_P \mathds{1}_{\{\|u\|_{L^2}^2 > \cK_s^2\}}  \dif \mu_{P,\tau}^{\tau^{-1}}(u)  & \leq \eps^{-2} \|w\|^2_{L^\infty} 	\int_{\{\|Pu\|_{L^2}^2>\cK_s^2\}}\|Pu\|_{L^2}^{6}\dif \mu_{P,\tau}^{\tau^{-1}}(u)   \\&\leq C(\cK,\cK_s,w)\eps^{-2} e^{-c(\cK,\cK_s)\tau} \leq C\Lambda_e^{-\frac18}.
			 	\end{align*}
		Combining the preceding two estimates, we obtain
			 \begin{align}\label{variation1}
			 		-\log  \frac{\cZ^{f_\eta}_{\tau}}{\mathcal{Z}^{f_\eta}_{\tau,0}}  \geq
			 		-\log \left( \int_{P\gH} e^{ \cW_P^\eps \mathds{1}_{\{\|u\|_{L^2}^2 \leq \cK_s^2 \}} } \dif   \mu_{P,0}^{\tau^{-1}}(u) \right)- C\Lambda_e^{-\frac18}. 
			 \end{align}
	It remains to compare the integral with respect to the free lower-symbol $ \mu_{P,0}^{\tau^{-1}}$ with the corresponding integral under the projected truncated Gaussian measure $ \mu^{f_\eta}_{0,P}$. We decompose
			 \begin{align}\label{lower:tail}
			 		\int_{P\gH} e^{ \cW_P^\eps \mathds{1}_{\{\|u\|_{L^2}^2 \leq \cK_s^2 \}} } \dif   \mu_{P,0}^{\tau^{-1}}(u) & \leq  \int_{P\gH} e^{ \cW_P^\eps} \dif   \mu^{f_\eta}_{0,P}(u) + \int_{P\gH} e^{ \cW_P^\eps \mathds{1}_{\{\|u\|_{L^2}^2 \leq \cK_s^2 \}} } | \dif   \mu_{P,0}^{\tau^{-1}}-  \dif   \mu^{f_\eta}_{0,P}| .
			 \end{align}
The following comparison is the only point where the subcritical argument differs from the lower-bound proof in Subsection~\ref{sec:lower bound}. Since $\cK_s<\cK_c$, the uniform projected Oh-Sosoe-Tolomeo bound gives a sufficiently small constant $\varsigma>0$ such that
\begin{align}\label{mian:appendix A}
	\sup_{\Lambda_e>0}   \int_{\gH} e^{\frac{1+\varsigma}{6}\|Pu\|_{L^6}^6} \mathds{1}_{\{  \|Pu\|_{L^2} \leq \cK_s \}} \dif \mu_{0}(u) <\infty.
\end{align}
For completeness, this estimate is proved in Appendix~\ref{sec:appendix:B}.
		The $L^1$-comparison estimate \eqref{L1:norm} is also independent of the pointwise exponential bound, and hence remains valid here: $\|\mu_{P,0}^{\tau^{-1}}-   \mu^{f_\eta}_{0,P} \|_{L^1(P\gH)}\leq C\Lambda_e^{-\frac14}$ whenever $\eta^{-1}\leq \Lambda_e^{\frac18}\leq \tau^{\frac{1}{32}}$. H\"older's inequality then gives
			\begin{align}\label{Holder:variation}
		&	\int_{P\gH} e^{ \cW_P^\eps \mathds{1}_{\{\|u\|_{L^2}^2 \leq \cK_s^2 \}} } | \dif   \mu_{P,0}^{\tau^{-1}}-  \dif   \mu^{f_\eta}_{0,P}|  \nonumber
			\\ &\leq  	\left( \int_{P\gH} e^{(1+\varsigma) \cW_P^\eps \mathds{1}_{\{\|u\|_{L^2}^2 \leq \cK_s^2 \}} } |\dif   \mu_{P,0}^{\tau^{-1}}-  \dif   \mu^{f_\eta}_{0,P}|\right)^{\frac{1}{1+\varsigma}}  \|\mu_{P,0}^{\tau^{-1}}-   \mu^{f_\eta}_{0,P} \|_{L^1(P\gH)}^{\frac{\varsigma}{1+\varsigma}} \nonumber
			\\ &\leq  	\left( \int_{P\gH}\big(  \mathds{1}_{\{\|u\|_{L^2}^2 \leq \cK_s^2 \}}  e^{(1+\varsigma) \cW_P^\eps } + \mathds{1}_{\{\|u\|_{L^2}^2 > \cK_s^2 \}}   \big)(\dif   \mu_{P,0}^{\tau^{-1}}+  \dif   \mu^{f_\eta}_{0,P})  \right)^{\frac{1}{1+\varsigma}} \Lambda_e^{-\frac{\varsigma}{4(1+\varsigma)}} \nonumber
			\\ &\leq \left( \int_{P\gH} \mathds{1}_{\{\|u\|_{L^2}^2 \leq \cK_s^2 \}}  e^{(1+\varsigma) \cW_P^\eps }  (\dif   \mu_{P,0}^{\tau^{-1}}+  \dif   \mu^{f_\eta}_{0,P}) +2  \right)^{\frac{1}{1+\varsigma}} \Lambda_e^{-\frac{\varsigma}{4(1+\varsigma)}}.
		\end{align}
	It remains to prove that the first factor on the right-hand side of \eqref{Holder:variation} is uniformly bounded. This is precisely where the mass-subcritical condition enters. By \cite[Lemma 2.3]{RS25} and $\|w^\eps\|_{L^1}=1$, for every $\eps >0$ we have
		\begin{align}\label{lemma2.3:RS25}
			\cW^\eps_P\leq \frac{1}{6} \|w^\eps\|_{L^1}^2 \|Pu\|_{L^6}^6=\frac{1}{6}  \|Pu\|_{L^6}^6.
		\end{align}
		Together with \eqref{mian:appendix A}, this implies
			\begin{align}\label{uniform:subcritical}
				\sup_{\eps>0,\, \Lambda_e>0} \int_{\{\|Pu\|_{L^2}^2 \leq \cK_s^2 \} } e^{(1+\varsigma) \cW_P^\eps }  \dif   \mu_{0,P} \leq  	\sup_{\Lambda_e>0} \int_{ \{\|Pu\|_{L^2}^2 \leq \cK_s^2 \}}  e^{\frac{1+\varsigma}{6} \|Pu\|_{L^6}^6 }\dif   \mu_{0,P} <\infty .
			\end{align}		
			We next establish the corresponding exponential moment bound for the lower symbol $ \mu_{P,0}^{\tau^{-1}}$. Combining \eqref{first way}--\eqref{same NZZ}, we obtain
			\begin{align*}
				\dif \mu_{P,0}^{\tau^{-1}}(u)
				&=  \left( \frac{\tau}{\pi}\right)^J \frac{1}{\tr_{\gF(\gH)} \big(\Gamma_{\tau,0}f_\eta\big(\frac{\cN}{\tau}\big)\big)} \left\langle  \xi(\sqrt{\tau}u),\tr_{\gF(Q\gH)}\big( 	\cU \Gamma_{\tau,0}f_\eta\big(\cN/ \tau\big)	\cU ^*\big) \xi(\sqrt{\tau}u)  \right\rangle \dif u
				\\ &\leq  C  \left( \frac{\tau}{\pi}\right)^J   \left\langle  \xi(\sqrt{\tau}u), \big(\Gamma_{\tau,0,P}f_\eta\big(\cN_P/ \tau\big)+ \eta^{-1}\Lambda_e^{-\frac12} \Gamma_{\tau,0,P} \big) \xi(\sqrt{\tau}u)  \right\rangle \dif u
				\\ &\leq C   \left( \frac{\tau}{\pi}\right)^J   \left\langle  \xi(\sqrt{\tau}u), \Gamma_{\tau,0,P} \xi(\sqrt{\tau}u)  \right\rangle \dif u
				\leq C	e^{\frac{1}{\tau}\Lambda_e^2\|Pu\|_{L^2}^2} \dif  \mu_{0,P}(u).
			\end{align*}
	Hence, using \eqref{uniform:subcritical} together with the constraint $\Lambda_e \leq \tau^{\frac14}$, we obtain
	\begin{align}\label{uniform:subcritical2}
		\sup_{\substack{ \eps>0, \\ 0< \Lambda_e \leq \tau^{1/4} }} & \int_{\{\|Pu\|_{L^2}^2 \leq \cK_s^2 \} } e^{(1+\varsigma) \cW_P^\eps }  \dif   \mu_{P,0}^{\tau^{-1}}  \leq  C	\sup_{0< \Lambda_e \leq \tau^{1/4} } \int_{ \{\|Pu\|_{L^2}^2 \leq \cK_s^2 \}}  e^{\frac{1+\varsigma}{6} \|Pu\|_{L^6}^6 } 	e^{\frac{1}{\tau}\Lambda_e^2\|Pu\|_{L^2}^2} \dif  \mu_{0,P} \nonumber
		\\ & \leq C	e^{\cK_s^2} \sup_{ \Lambda_e >0} \int_{ \{\|Pu\|_{L^2}^2 \leq \cK_s^2 \}}  e^{\frac{1+\varsigma}{6} \|Pu\|_{L^6}^6 } 	 \dif  \mu_{0,P} <\infty .
	\end{align}
 Combining \eqref{Holder:variation}, \eqref{uniform:subcritical}, and \eqref{uniform:subcritical2}, we obtain
 \begin{align}\label{observable:against:difference}
 	\int_{P\gH} e^{ \cW_P^\eps \mathds{1}_{\{\|u\|_{L^2}^2 \leq \cK_s^2 \}} } | \dif   \mu_{P,0}^{\tau^{-1}}-  \dif   \mu^{f_\eta}_{0,P}| \leq C \Lambda_e^{-\frac{\varsigma}{4(1+\varsigma)}}.
 \end{align}
 Inserting \eqref{observable:against:difference} into \eqref{variation1} and \eqref{lower:tail}, we conclude that, for all sufficiently large $\tau$ and $\Lambda_e$, and all $ \eta\in \big(0,\frac{1}{2}\cK^2\big)$ and $ \eps \in(0,1)$ satisfying \eqref{refined:relation},
 \begin{align}\label{lower:bound:sub}
 	-\log  \frac{\cZ^{f_\eta}_{\tau}}{\mathcal{Z}^{f_\eta}_{\tau,0}}  \geq
 	-\log \left( \int_{P\gH} e^{ \cW_P^\eps  } \dif     \mu^{f_\eta}_{0,P}(u) \right)- C\Lambda_e^{-\frac{\varsigma}{4(1+\varsigma)}}.
 \end{align}
 This proves the desired projected lower bound in the mass-subcritical regime.
	
	\textbf{Step 2: Free energy upper bound in the mass-subcritical regime.}
	We now prove the matching upper bound. Compared with Section~\ref{sec3:lower:upper:free-energy}, two modifications are needed in the mass-subcritical regime. The first concerns the localization error in the relative entropy estimate of Lemma~\ref{lemma:relative entropy}. The second occurs at the final stage of the semiclassical analysis, namely in \eqref{first_hard_term1} and \eqref{second_easy_term}, where one has to control the error produced by the cutoff. In both places, the pointwise exponential bounds from the critical argument are replaced by the uniform moment bounds available in the subcritical regime. We also note that Lemma~\ref{errors:localization_interaction} remains valid here, since its proof does not use pointwise exponential bounds. We begin with the corresponding refinement of Lemma~\ref{lemma:relative entropy}.
	\begin{lemma}\label{lemma:relative entropy:sub}
	For all sufficiently large $\tau$ and $\Lambda_e$, and all $ \eta\in \big(0,\frac{1}{2}\cK^2\big)$, $\eps \in(0,1)$ satisfying \eqref{refined:relation}, we have
	\begin{align*}
		\left| \cH(\Gamma_{\mathrm{ts}}, \Gamma^{f_\eta}_{\tau,0}) -\cH( \Gamma^{f_\eta}_{\tau,P}, \Gamma^{f_\eta}_{\tau,0,P})\right| \leq  C \Lambda_e^{-\frac{1}{4}},
	\end{align*}
	where $C$ is a constant depending only on $w$ and $\cK$.
\end{lemma}
\begin{proof}
We use the notation introduced in Lemma~\ref{lemma:relative entropy}. It is enough to compare the three quantities 
\begin{align}\label{aim:three:terms}
	\left| \frac{ \textbf{A}_{\tau,P}}{ \textbf{B}_{\tau,P}}-\frac{ \textbf{A}_\tau}{ \textbf{B}_\tau}\right|  ,\qquad \log\!\left(\frac{ \textbf{B}_\tau}{ \textbf{B}_{\tau,P}}\right), \qquad \log\!\left(\frac{ \textbf{C}_\tau}{ \textbf{C}_{\tau,P}}\right).
\end{align}
We start with the first term. By \eqref{first part in Gamma_tau,P-Gamma_1} and \eqref{first part in Gamma_tau,P-Gamma_1.4}, we have
\begin{align*}
 \left| \textbf{A}_{\tau}- \textbf{A}_{\tau,P} \right| &\leq  C(  \cK,w)  \eta^{-1}\eps^{-2} \cZ_{\tau,0,Q}^{-1}  \tr_{\gF(\gH)}  \left(  e^{-\bbH_{\tau,0}+\bbW_{\tau,P,\cU} } \frac{\cN_{Q,\cU}}{\tau} \mathds{1}_{\{ \cN_{P,\cU}/\tau  \leq \cK^2\}}  \right)
	\\&=C(  \cK,w)  \eta^{-1}\eps^{-2}\tr_{\gF(P\gH)\otimes \gF(Q\gH)}\Big( e^{-\bbH_{\tau,P}}  \mathds{1}_{\{ \cN_P/\tau  \leq \cK^2\}} \otimes \Gamma_{\tau,0,Q} \cN_Q /\tau\Big) 
	\\ &= C(  \cK,w) \eta^{-1}\eps^{-2}  \tr_{\gF(P\gH)}\Big( e^{-\bbH_{\tau,P}}  \mathds{1}_{\{ \cN_P/\tau  \leq \cK^2\}} \Big)   \tr_{\gF(Q\gH)}\Big(  \Gamma_{\tau,0,Q} \cN_Q /\tau\Big) 
	\\ &\leq  C   \eta^{-1}\eps^{-2}\Lambda_e^{-\frac12} \tr_{\gF(P\gH)}\Big( e^{-\bbH_{\tau,P}}  \mathds{1}_{\{ \cN_P/\tau  \leq \cK^2\}} \Big),
\end{align*}
where in the last inequality we used  \eqref{Wp-W:second}.
The same argument gives the corresponding denominator estimate:
\begin{align*}
	\left| \textbf{B}_{\tau}- \textbf{B}_{\tau,P} \right| \leq  C   \eta^{-1}\Lambda_e^{-\frac12} \tr_{\gF(P\gH)}\Big( e^{-\bbH_{\tau,P}}  \mathds{1}_{\{ \cN_P/\tau  \leq \cK^2\}} \Big).
\end{align*}
Inserting these two bounds into \eqref{Gamma_tau,P-Gamma_12}, and using \eqref{lower:bound:for:denominator} together with $\frac{\textbf{A}_\tau}{\textbf{B}_\tau} \leq \eps^{-2}\|w\|_{L^\infty}^2   \cK^6$, we obtain
	\begin{align}\label{comparison:truncated:partition}
		& \left| \frac{ \textbf{A}_{\tau,P}}{ \textbf{B}_{\tau,P}}-\frac{ \textbf{A}_\tau}{ \textbf{B}_\tau}\right| 
 \leq  \frac{| \textbf{A}_{\tau,P}- \textbf{A}_{\tau}| }{  \textbf{B}_{\tau,P} } + \frac{ \eps^{-2}\|w\|_{L^\infty}^2   \cK^6 | \textbf{B}_{\tau,P}- \textbf{B}_{\tau}|  }{\textbf{B}_{\tau,P}   } \nonumber
	\\  & \leq C   \eta^{-1}\eps^{-2}\Lambda_e^{-\frac12} \frac{\tr_{\gF(P\gH)}\Big( e^{-\bbH_{\tau,P}}  \mathds{1}_{\{ \cN_P/\tau  \leq \cK^2\}} \Big)}{\tr_{\gF(P\gH)}\Big( e^{-\bbH_{\tau,0,P}}  f_\eta( \cN_P/\tau  ) \Big)} 
	\\ & =C   \eta^{-1}\eps^{-2}\Lambda_e^{-\frac12} \frac{\tr_{\gF(P\gH)}\Big( e^{-\bbH_{\tau,P}}  \mathds{1}_{\{ \cN_P/\tau  \leq \cK^2\}} \Big)}{\tr_{\gF(P\gH)}\Big( e^{-\bbH_{\tau,0,P}}  \mathds{1}_{\{ \cN_P/\tau  \leq \cK^2\}} \Big)} \cdot \frac{\tr_{\gF(P\gH)}\Big( e^{-\bbH_{\tau,0,P}}  \mathds{1}_{\{ \cN_P/\tau  \leq \cK^2\}}  \Big)}{\tr_{\gF(P\gH)}\Big( e^{-\bbH_{\tau,0,P}}f_\eta( \cN_P/\tau  )\Big)}. \nonumber
	\end{align}
	 We first control the second factor on the right-hand side of \eqref{comparison:truncated:partition}. To do so, replace the sharp cutoff by a smooth one. Choose $\tilde{f} \in C_c^\infty([0,\infty);\R)$ such that $\tilde{f}\equiv 1$ on $[0,\cK^2]$. Applying Lemma~\ref{project:convergence:rate} to $\tilde{f}$ and $f_\eta$, we obtain, for all $\eta \in(0,\frac{\cK^2}{2})$ and all sufficiently large $\tau$ and $\Lambda_e$ satisfying \eqref{refined:relation},
	\begin{align}\label{bound:for:second:ratio}
	&	\frac{\tr_{\gF(P\gH)}\Big( e^{-\bbH_{\tau,0,P}}  \mathds{1}_{\{ \cN_P/\tau  \leq \cK^2\}}  \Big)}{\tr_{\gF(P\gH)}\Big( e^{-\bbH_{\tau,0,P}}  f_\eta( \cN_P/\tau  ) \Big)} \leq \frac{\tr_{\gF(P\gH)}\Big( \Gamma_{\tau,0,P} \tilde{f}(\cN_P/\tau )\Big)}{\tr_{\gF(P\gH)}\Big(  \Gamma_{\tau,0,P}  f_\eta(\cN_P/\tau )\Big)} \nonumber
  \leq \frac{ \int_{\mathfrak{H}} \tilde{f} (\|Pu\|_{L^2}^2) \, \dif \mu_0 (u) +C\tau^{-\frac14}}{\int_{\mathfrak{H}} f_\eta (\|Pu\|_{L^2}^2) \, \dif \mu_0 (u) -C\tau^{-\frac14}} 
  	\\ &\quad \leq \frac{\|\tilde{f}\|_{L^\infty}+1}{\int_{\mathfrak{H}} f_\eta (\|u\|_{L^2}^2) \, \dif \mu_0 (u) -C\Lambda_e^{-\frac18}-C\tau^{-\frac14} } \leq   \frac{\|\tilde{f}\|_{L^\infty}+1}{ \frac12 \int_{\gH} \mathds{1}_{\{ \|u\|_{L^2}^2 \leq \frac{\cK^2}{4}\}} \dif \mu_0(u) } \leq C.
	\end{align}
	Here, in the last line we used \eqref{int:f(u)-f(Pu)} and $\int_{\mathfrak{H}} f_\eta (\|u\|_{L^2}^2)  \dif \mu_0 (u) -C\Lambda_e^{-\frac18}-C\tau^{-\frac14} \geq \int_{\gH} \mathds{1}_{\{ \|u\|_{L^2}^2 \leq \frac{\cK^2}{4}\}} \dif \mu_0(u)-C\Lambda_e^{-\frac18} \geq  \frac12\int_{\gH} \mathds{1}_{\{ \|u\|_{L^2}^2 \leq \frac{\cK^2}{4}\}} \dif \mu_0(u)$, valid for all sufficiently large $\tau$ and $\Lambda_e$.
	
It remains to obtain a uniform bound on the first ratio of truncated partition functions on the right-hand side of \eqref{comparison:truncated:partition}. We use the estimates from Subsection~\ref{sec:lower bound} to derive a lower bound on the corresponding relative free energy. Define
	\begin{align*}
		 \Gamma^{\mathds{1}_\cK}_{\tau,P}:= \frac{e^{-\bbH_{\tau,P}}  \mathds{1}_{\{ \cN_P/\tau  \leq \cK^2\}}}{\tr_{\gF(P\gH)}\Big( e^{-\bbH_{\tau,P}}  \mathds{1}_{\{ \cN_P/\tau  \leq \cK^2\}} \Big)}  ,\qquad  	 \Gamma^{\mathds{1}_\cK}_{\tau,0,P}: = \frac{e^{-\bbH_{\tau,0,P}}  \mathds{1}_{\{ \cN_P/\tau  \leq \cK^2\}}}{\tr_{\gF(P\gH)}\Big( e^{-\bbH_{\tau,0,P}}  \mathds{1}_{\{ \cN_P/\tau  \leq \cK^2\}} \Big)}.
	\end{align*}
	Let $\nu^{\tau^{-1}}_{P,\tau}$ and $\nu^{\tau^{-1}}_{P,0}$ be the interacting and free lower symbols associated with these two states on $P\gH$ at scale $\tau^{-1}$, respectively. Combining \eqref{lower:variational}, \eqref{lower bound interaction energy1}, and \eqref{variation}, and using the variational principle together with the Berezin-Lieb type inequality \eqref{eq:Berezin-Lieb}, we obtain
		\begin{align}\label{lower:variational:sub}
			-\log  &  \frac{\tr_{\gF(P\gH)}\Big( e^{-\bbH_{\tau,P}}  \mathds{1}_{\{ \cN_P/\tau  \leq \cK^2\}} \Big)}{\tr_{\gF(P\gH)}\Big( e^{-\bbH_{\tau,0,P}}  \mathds{1}_{\{ \cN_P/\tau  \leq \cK^2\}} \Big)} 
		=\cH( \Gamma^{\mathds{1}_\cK}_{\tau,P}, \Gamma^{\mathds{1}_\cK}_{\tau,0,P}) -  \frac{1}{\tau^3} \tr  \left(W^\eps_{P} (	 \Gamma^{\mathds{1}_\cK}_{\tau,P})^{(3)} \right)\nonumber
		\\&	\qquad \geq  \cH_{\mathrm{cl}} (\nu_{P,\tau}^{\tau^{-1}}, \nu_{P,0}^{\tau^{-1}})-       \int_{P\gH}  \cW^\eps_P(u) \dif \nu_{P,\tau}^{\tau^{-1}}(u)    \nonumber
			\\ &\qquad \geq - \log \left( \int_{P\gH} e^{ \cW_P^\eps \mathds{1}_{\{\|u\|_{L^2}^2 \leq \cK_s^2 \}} } \dif   \nu_{P,0}^{\tau^{-1}}(u) \right)-\int_{P\gH}  \cW^\eps_P(u) \mathds{1}_{\{\|u\|_{L^2}^2 > \cK_s^2   \}} \dif \nu_{P,\tau}^{\tau^{-1}}(u).\nonumber
			\\ &\qquad  \geq - \log \left( \int_{P\gH} \mathds{1}_{\{\|u\|_{L^2}^2 \leq \cK_s^2 \}}  e^{ \cW_P^\eps } \dif   \nu_{P,0}^{\tau^{-1}}(u) +1 \right)-\int_{P\gH}  \cW^\eps_P(u) \mathds{1}_{\{\|u\|_{L^2}^2 > \cK_s^2   \}} \dif \nu_{P,\tau}^{\tau^{-1}}(u).
		\end{align}
		The two terms on the right-hand side are controlled as in Step~1. First, by \eqref{same NZZ}, we have
	     \begin{align*}
	     		\dif \nu_{P,0}^{\tau^{-1}}(u)
	     	&=  \left( \frac{\tau}{\pi}\right)^J \frac{1}{\tr_{\gF(P\gH)} \big(\Gamma_{\tau,0,P}\mathds{1}_{\{ \cN_P/\tau  \leq \cK^2\}} \big)} \left\langle  \xi(\sqrt{\tau}u),	\Gamma_{\tau,0,P}\mathds{1}_{\{ \cN_P/\tau  \leq \cK^2\}}  \xi(\sqrt{\tau}u)  \right\rangle \dif u
	     	\\ &\leq  C  \left( \frac{\tau}{\pi}\right)^J   \left\langle  \xi(\sqrt{\tau}u),   \Gamma_{\tau,0,P}  \xi(\sqrt{\tau}u)  \right\rangle \dif u
	     	\leq C  e^{\frac{1}{\tau}\Lambda_e^2\|Pu\|_{L^2}^2} \dif  \mu_{0,P}(u).
	     \end{align*}
	     Here we used $\tr_{\gF(P\gH)} \big(\Gamma_{\tau,0,P}\mathds{1}_{\{ \cN_P/\tau  \leq \cK^2\}} \big)\geq C$, which follows from the same argument as \eqref{bound:for:second:ratio}. Therefore, by \eqref{uniform:subcritical2} and the constraint $\Lambda_e \leq \tau^{\frac14}$, we get
	     \begin{align*}
	     	\sup_{\substack{ \eps>0, \\ 0< \Lambda_e \leq \tau^{1/4} }}  \int_{P\gH}  \mathds{1}_{\{\|u\|_{L^2}^2 \leq \cK_s^2 \}}  e^{ \cW_P^\eps } \dif   \nu_{P,0}^{\tau^{-1}}(u) \leq  C	\sup_{0< \Lambda_e \leq \tau^{1/4} } \int_{ \{\|Pu\|_{L^2}^2 \leq \cK_s^2 \}}  e^{\frac{1+\varsigma}{6} \|Pu\|_{L^6}^6 } 	e^{\frac{1}{\tau}\Lambda_e^2\|Pu\|_{L^2}^2} \dif  \mu_{0,P}<\infty.
	     \end{align*}
	     For the second term on the right-hand side of \eqref{lower:variational:sub}, we apply Lemma~\ref{lemma:interacting:tail} to the Gibbs state $ \Gamma^{\mathds{1}_\cK}_{\tau,P}$ with $R=\cK_s^2$. This gives, for all admissible parameters satisfying \eqref{refined:relation},
	     \begin{align*}
	     	\int_{P\gH}  \cW^\eps_P(u) \mathds{1}_{\{\|u\|_{L^2}^2 > \cK_s^2   \}} \dif \nu_{P,\tau}^{\tau^{-1}}(u)\leq \eps^{-2} \|w\|^2_{L^\infty} 	\int_{\{\|Pu\|_{L^2}^2>\cK_s^2\}}\|Pu\|_{L^2}^{6}\dif \nu_{P,\tau}^{\tau^{-1}}(u)   \leq C \eps^{-2} e^{-c\tau} <\infty.
	     \end{align*}
Combining the preceding two estimates, we conclude that, uniformly for all parameters satisfying \eqref{refined:relation},
	\begin{align}\label{bound:for:first:ratio}
	\frac{\tr_{\gF(P\gH)}\Big( e^{-\bbH_{\tau,P}}  \mathds{1}_{\{ \cN_P/\tau  \leq \cK^2\}} \Big)}{\tr_{\gF(P\gH)}\Big( e^{-\bbH_{\tau,0,P}}  \mathds{1}_{\{ \cN_P/\tau  \leq \cK^2\}} \Big)}  \leq C.
	\end{align}
	Hence, whenever $\eta^{-1}\leq \Lambda_e^{\frac18}$, $\eps^{-3} \leq \Lambda_e^{\frac18}$, and $\Lambda_e \leq \tau^{\frac14}$, we deduce from \eqref{comparison:truncated:partition}, \eqref{bound:for:second:ratio} and \eqref{bound:for:first:ratio} that
	\begin{align*}
		 \left| \frac{ \textbf{A}_{\tau,P}}{ \textbf{B}_{\tau,P}}-\frac{ \textbf{A}_\tau}{ \textbf{B}_\tau}\right| \leq C \eta^{-1}\eps^{-2}\Lambda_e^{-\frac12}  \leq C\Lambda_e^{-\frac14}.
	\end{align*}
	For the second term in \eqref{aim:three:terms}, we use $ \log(1+t) = O(t) $ for $|t|$ sufficiently small, together with \eqref{comparison:truncated:partition}, \eqref{bound:for:second:ratio} and \eqref{bound:for:first:ratio}. Thus, whenever $\eta^{-1}\leq \Lambda_e^{\frac18}$, $\eps^{-3} \leq \Lambda_e^{\frac18}$, and $\Lambda_e \leq \tau^{\frac14}$,
		\begin{align*}
		\left| \log \left(\frac{ \textbf{B}_\tau}{ \textbf{B}_{\tau,P}}\right) \right| =  \left| \log\!\left(1+ \frac{   \textbf{B}_\tau -\textbf{B}_{\tau,P} }{   \textbf{B}_{\tau,P}}\right) \right| 
		\leq \frac{ |\textbf{B}_\tau -\textbf{B}_{\tau,P} | }{  \textbf{B}_{\tau,P}}
		\leq C\Lambda_e^{-\frac{1}{4}}.
	\end{align*}
Finally, the term $ \log\!\left(\frac{ \textbf{C}_\tau}{ \textbf{C}_{\tau,P}}\right)$ involves only the free Gibbs states, so the estimate \eqref{third:free:state:control} applies unchanged. This completes the proof of the lemma.
		\end{proof}
We now return to the projected semiclassical estimate. Since the estimates \eqref{lower:bound3} and \eqref{bernoulli:inequality} do not rely on the pointwise exponential bound, they remain valid here. Therefore, for all parameters satisfying \eqref{refined:relation}, we have
\begin{align}\label{uniform:frequencey:cutoff}
	\frac{\tr _{\gF(P\gH)}\left(e^{-\bbH_{\tau,P}}f_\eta(\mathcal{N}_P/\tau) \right) }{\tr_{\gF(P\gH)}(e^{-\bbH_{\tau,0,P}})} & \geq  \left( 1-\Lambda_e^{-1} \right) \int_{P\gH} f_\eta\left( \|Pu\|_{L^2}^2 \right)  e^{\cW_P^\eps }  
	\dif \mu_{0,P}(u)  \nonumber
	\\ &\qquad - C_{\cK_c,w} \frac{ \Lambda_e}{\sqrt{\tau}}    \int_{P\gH}  \mathds{1}_{\{\|Pu\|_{L^2}\leq \cK\} } e^{\cW_P^\eps } 
	\dif \mu_{0,P}(u) .
\end{align}
Using the uniform estimate \eqref{uniform:subcritical} once more, we obtain
		\begin{align*}
			\sup_{\eps, \, \Lambda_e>0}  \int_{P\gH} f_\eta\left( \|Pu\|_{L^2}^2 \right)  e^{\cW_P^\eps } 
			\dif \mu_{0,P}(u)  \leq  	 \sup_{\Lambda_e>0}  \int_{P\gH} \mathds{1}_{\{\|Pu\|_{L^2} \leq  \cK \}} e^{\frac{1}{6}  \|Pu\|_{L^6}^6 } 
			\dif \mu_{0,P}(u)  <\infty.
		\end{align*}
	Substituting this bound into \eqref{uniform:frequencey:cutoff} yields
			\begin{align}\label{semi:subcritical:lower}
			\frac{\tr _{\gF(P\gH)}\left(e^{-\bbH_{\tau,P}}f_\eta(\mathcal{N}_P/\tau) \right) }{\tr_{\gF(P\gH)}(e^{-\bbH_{\tau,0,P}})}  \geq   \int_{P\gH} f_\eta\left( \|Pu\|_{L^2}^2 \right)  e^{\cW_P^\eps }  
			\dif \mu_{0,P}(u)  - C\Lambda_e^{-1}.
		\end{align}
	We then combine \eqref{semi:subcritical:lower} with Lemmas~\ref{errors:localization_interaction} and \ref{lemma:relative entropy:sub}, and repeat the argument leading to \eqref{finite_analysis1}--\eqref{finite_analysis3}. The pointwise exponential bounds are not used in these estimates. Hence, for all sufficiently large $\tau$ and $\Lambda_e$, and all $ \eta\in \big(0,\frac{1}{2}\cK^2\big)$, $\eps \in(0,1)$ satisfying \eqref{refined:relation}, we have
	\begin{align}\label{upper:bound:sub}
		-\log \frac{\cZ^{f_\eta}_{\tau}}{\cZ^{f_\eta}_{\tau,0}}  \leq - \log \left( \int_{P\gH}  e^{\cW_P^\eps(u) } 
		\dif \mu^{f_\eta}_{0,P}(u) \right) + O\left(  \Lambda_e^{-\frac18}\right).
	\end{align}
	This proves the desired projected upper bound in the mass-subcritical regime.
			
				\textbf{Step 3: Conclusion of Proposition~\ref{Pro:Hartree}.} The preceding two steps identify the projected free energy. It remains to remove the projection. This is done by the Cauchy argument from Proposition~\ref{pro:cauchy:seq1}: the projected Hartree partition functions form a Cauchy family as $\Lambda_e\to\infty$, and their limit is the corresponding full partition function. This argument is independent of the pointwise exponential bounds. Hence, combining \eqref{lower:bound:sub} and \eqref{upper:bound:sub}, we obtain that, for all sufficiently large $\Lambda_e>0$ and all $\eta\in \big(0,\frac{1}{2}\cK^2\big)$, $\eps\in(0,1)$ satisfying $\eta^{-1}\leq \Lambda_e^{\frac18}$ and $\eps^{-3}\leq \Lambda_e^{\frac18}$,
				\begin{align*}
				\left|	\log \left( \int_{P\gH}  e^{\cW_P^\eps(u) } 
				\dif \mu^{f_\eta}_{0,P}(u)\right) -\log \left( \int_{\gH}  e^{\cW^\eps(u) } 
				\dif \mu^{f_\eta}_{0}(u)\right)\right| \leq O\left(  \Lambda_e^{-\sigma} \right),
				\end{align*}
			for some sufficiently small $\sigma>0$, depending only on the exponent $\varsigma$ in the projected Oh-Sosoe-Tolomeo estimate \eqref{mian:appendix A}. Under the assumptions $\eta\geq \tau^{-\frac{1}{64}}$ and $\eps\geq \tau^{-\frac{1}{96}}$, we may choose $\Lambda_e$ so that \eqref{refined:relation} holds and $\Lambda_e\to\infty$. Therefore,
					\begin{align*}
					&	\left| -\log  \frac{\cZ^{f_\eta}_{\tau}}{\mathcal{Z}^{f_\eta}_{\tau,0}} + \log \left( \int_{ \gH} e^{\cW^\eps(u) } \dif \mu^{f_\eta}_0 (u)\right) \right| 
					\leq 	\left| -\log  \frac{\cZ^{f_\eta}_{\tau}}{\mathcal{Z}^{f_\eta}_{\tau,0}} + \log \left( \int_{P \gH} e^{\cW_P^\eps(u) } \dif \mu^{f_\eta}_{0,P} (u)\right) \right|  
					\\ &\qquad \qquad  \qquad  \qquad + 	\left|  \log \left( \int_{P \gH} e^{\cW_P^\eps(u) } \dif \mu^{f_\eta}_{0,P} (u)\right) -\log \left( \int_{ \gH} e^{\cW^\eps(u) } \dif \mu^{f_\eta}_0 (u)\right)  \right| \leq O\left(  \Lambda_e^{-\sigma} \right).
				\end{align*} 
					Using \eqref{lemma2.3:RS25} and the elementary inequality $|e^x-e^y| \leq e^y(e^{|x-y|}-1)$, $x,y\in\R$, we then obtain
			\begin{align*}
				\left|\frac{\cZ^{f_\eta}_{\tau}}{\mathcal{Z}^{f_\eta}_{\tau,0}} -\int_{ \gH} e^{\cW^\eps(u) } \dif \mu^{f_\eta}_0 (u)  \right| & \leq  \int_{ \gH} e^{\cW^\eps(u) } \dif \mu^{f_\eta}_0 (u)  \big(e^{C\Lambda_e^{-\sigma}}-1\big) 
				\\ &\leq C\Lambda_e^{-\sigma} \int_{ \{\|u\|_{L^2} \leq \cK\}} e^{\frac16 \|u\|_{L^6}^6}  \dif \mu_0 (u)  \leq C \Lambda_e^{-\sigma},
			\end{align*}
			where the last step follows from the subcritical integrability result of \cite[Theorem 1.1]{OST22}. Letting $\tau\to\infty$ and $\Lambda_e\to\infty$ proves \eqref{key:bound:sub1}.
				
			We next prove the convergence of the reduced density matrices. The argument follows the proof of \eqref{key:bound2}; we only record the changes needed in the subcritical regime. First, we claim convergence to the projected Hartree density matrices. Namely, there exists $\beta>0$, independent of $\eps,\eta,\tau$ and $\Lambda_e$, such that
				\begin{align}\label{sub:dm:projected-convergence}
					\left\|\frac{k!}{\tau^k}((\Gamma^{f_\eta}_\tau)_P)^{(k)}
					-\int_{P\gH}|u^{\otimes k}\rangle\langle u^{\otimes k}|\dif\mu_P^{\eps,f_\eta}(u)\right\|_{\gS^1(\gH^{(k)})}
					\leq C_{k,\cK}\Lambda_e^{-\beta},
				\end{align}
			where the projected Hartree measure is defined by 
				$	\dif \mu^{\eps,f_\eta}_P(u)
				:= \frac{e^{\cW_P^\eps(u)}\dif\mu^{f_\eta}_{0,P}(u)}{\int_{P\gH} e^{\cW_P^\eps(v)}\dif\mu^{f_\eta}_{0,P}(v)}$ and $((\Gamma^{f_\eta}_\tau)_P)^{(k)}=P^{\otimes k}(\Gamma^{f_\eta}_\tau)^{(k)}P^{\otimes k}$. Since the estimates \eqref{de finetti error trace norm localized}--\eqref{density matrix convergence36} do not use pointwise exponential bounds, they remain valid under the present constraint \eqref{refined:relation}. Thus,
				\begin{align*}
					&	\left\|\frac{k!}{\tau^k}((\Gamma^{f_\eta}_\tau)_P)^{(k)}
					-\int_{P\gH}|u^{\otimes k}\rangle\langle u^{\otimes k}|\dif\mu_P^{\eps,f_\eta}(u)\right\|_{\gS^1(\gH^{(k)})}
						\\ &\quad \leq \left\|\frac{k!}{\tau^k}((\Gamma^{f_\eta}_\tau)_P)^{(k)}-
						\int_{P\gH}|u^{\otimes k}\rangle\langle u^{\otimes k}|\dif\mu^{\tau^{-1}}_{P,\tau}(u)\right\|_{\gS^1(\gH^{(k)})}  
				\\ &\quad +	\left\| \int_{P\gH} |u^{\otimes k}\rangle  \langle u^{\otimes k} | \dif \mu^{\tau^{-1}}_{P,\tau}(u)- \int_{P\gH} |u^{\otimes k}\rangle  \langle u^{\otimes k} | \dif \mu^{\eps,f_\eta}_P(u) \right\| _{\gS^1(\gH^{(k)})} 
				\leq C_{k,\cK}\big( \Lambda_e^{-1}+	\|\mu^{\tau^{-1}}_{P,\tau}-\mu^{\eps,f_\eta}_P\|_{L^1}^{1/2}\big).
				\end{align*}
				It remains to control $\|\mu^{\tau^{-1}}_{P,\tau}-\mu^{\eps,f_\eta}_P\|_{L^1}$ without invoking the pointwise exponential bounds. Combining the lower bound argument leading to \eqref{variation1} and \eqref{lower:bound:sub} with the upper bound \eqref{upper:bound:sub}, we obtain
					\begin{align}\label{density matrix convergence42}
						- \log & \left( \int_{P\gH}  e^{\cW_P^\eps(u) } 
						\dif \mu^{f_\eta}_{0,P}(u) \right) + C \Lambda_e^{-\frac18}  \geq 
					-\log  \frac{\cZ^{f_\eta}_{\tau}}{\mathcal{Z}^{f_\eta}_{\tau,0}}  \nonumber 
					 \\ &		\geq \cH_{\mathrm{cl}} (\mu_{P,\tau}^{\tau^{-1}}, \mu_{P,0}^{\tau^{-1}})-  \int_{P\gH}  \cW^\eps_P \mathds{1}_{\{\|u\|_{L^2}^2 \leq \cK_s^2 \}}  \dif \mu_{P,\tau}^{\tau^{-1}}(u)-C\Lambda_e^{-\frac18}  \nonumber
					\\ & = \cH_{\mathrm{cl}} (\mu_{P,\tau}^{\tau^{-1}}, \mu')
					-\log \left( \int_{P\gH} e^{\cW_P^\eps (u) \mathds{1}_{\{\|u\|_{L^2}^2 \leq  \cK_s^2  \}}   } \dif   \mu_{P,0}^{\tau^{-1}}(u)\right)-C\Lambda_e^{-\frac18}  \nonumber
					\\ & \geq  \cH_{\mathrm{cl}} (\mu_{P,\tau}^{\tau^{-1}}, \mu')-\log \left( \int_{P\gH} e^{\cW_P^\eps (u)} \dif  \mu^{f_\eta}_{0,P}(u)\right)-  C\Lambda_e^{-\frac{\varsigma}{4(1+\varsigma)}}, 
				\end{align}
				where $	\dif  \mu'(u):=\frac{ \mathrm{\exp}\big(\cW_P^\eps(u)\mathds{1}_{\{\|u\|_{L^2}^2 \leq \cK_s^2 \}} \big)    \dif   \mu_{P,0}^{\tau^{-1}}(u)}{\int_{P\gH} \mathrm{\exp}\big(\cW_P^\eps(v)\mathds{1}_{\{\|v\|_{L^2}^2 \leq \cK_s^2 \}} \big) \dif   \mu_{P,0}^{\tau^{-1}}(v)}$. Pinsker's inequality gives
				\begin{align*}
					\|\mu^{\tau^{-1}}_{P,\tau}-\mu'\|_{L^1} \leq \sqrt{2}\cH_{\mathrm{cl}}\big(\mu^{\tau^{-1}}_{P,\tau},\mu'\big)^{\frac12}
					\leq C\Lambda_e^{-\frac{\varsigma}{8(1+\varsigma)}}. 
				\end{align*}
				We now compare $\mu'$ with $\mu^{\eps,f_\eta}_P$. By the same calculation as in \eqref{density matrix convergence33},
					\begin{align*}
						\| \mu' - \mu^{\eps,f_\eta}_P &   \|_{L^1(P\gH)} \leq \frac{1}{\int_{P\gH} e^{\cW_P^\eps \mathds{1}_{\{\|v\|_{L^2}^2 \leq \cK_s^2 \}} } \dif   \mu_{P,0}^{\tau^{-1}}(v)} \int_{P\gH}  e^{\cW_P^\eps \mathds{1}_{\{\|u\|_{L^2}^2 \leq \cK_s^2  \}} }  \big|  \dif   \mu_{P,0}^{\tau^{-1}}-  \dif \mu^{f_\eta}_{0,P} \big|  \nonumber
						\\ & +\frac{1}{\int_{P\gH} e^{\cW_P^\eps \mathds{1}_{\{\|v\|_{L^2}^2 \leq \cK_s^2 \}} } \dif   \mu_{P,0}^{\tau^{-1}}(v)} \left| \int_{P\gH} e^{\cW_P^\eps \mathds{1}_{\{\|v\|_{L^2}^2 \leq \cK_s^2 \}} } \dif   \mu_{P,0}^{\tau^{-1}}(v) - \int_{P\gH} e^{\cW_P^\eps(v)} \dif   \mu^{f_\eta}_{0,P} (v) \right| 
					\\ & \leq C \int_{P\gH}  e^{\cW_P^\eps \mathds{1}_{\{\|u\|_{L^2}^2 \leq \cK_s^2  \}} }  \big|  \dif   \mu_{P,0}^{\tau^{-1}}-  \dif \mu^{f_\eta}_{0,P} \big| + C \int_{P\gH} \big| e^{\cW_P^\eps(v)\mathds{1}_{\{\|v\|_{L^2}^2 \leq \cK_s^2 \}} } - e^{\cW_P^\eps(v) } \big| \dif   \mu^{f_\eta}_{0,P} (v)  ,
					\end{align*}
						where we used $\int_{P\gH} e^{\cW_P^\eps \mathds{1}_{\{\|v\|_{L^2}^2 \leq \cK_s^2 \}} } \dif   \mu_{P,0}^{\tau^{-1}}(v)\geq 1$ in the last line. The second term on the right-hand side vanishes because $\mu^{f_\eta}_{0,P}$ is supported on $\{\|v\|_{L^2}^2\leq \cK^2\}$ and $\cK_s^2>\cK^2$. The first term is exactly \eqref{observable:against:difference}, and is bounded by $C\Lambda_e^{-\frac{\varsigma}{4(1+\varsigma)}}$. Consequently, $\|\mu^{\tau^{-1}}_{P,\tau}-\mu^{\eps,f_\eta}_P\|_{L^1}^{1/2}\leq C\Lambda_e^{-\beta}$ for some sufficiently small $\beta>0$. This proves the localized convergence \eqref{sub:dm:projected-convergence}.
					
				The removal of the projection on the quantum side is unchanged, since it relies only on Lemma~\ref{errors:localization_interaction}; the assumptions $\eps^{-3}\leq\Lambda_e^{1/8}$ and $\Lambda_e\leq\tau^{1/4}$ are precisely those needed for that lemma. Therefore, 
				\begin{align}\label{sub:dm:quantum-projection-removal}
					\left\|\frac{k!}{\tau^k}(\Gamma^{f_\eta}_\tau)^{(k)}-\frac{k!}{\tau^k}((\Gamma^{f_\eta}_\tau)_P)^{(k)}\right\|_{\gS^1(\gH^{(k)})}
					\leq C_{k,\cK}\Lambda_e^{-\frac18}.
				\end{align}
				The removal of the projection on the classical side uses only \eqref{sub:dm:projected-convergence}, \eqref{sub:dm:quantum-projection-removal}, and the Cauchy argument of Proposition~\ref{pro:cauchy:seq1}. These ingredients are independent of the pointwise exponential bounds. Thus, under the constraint \eqref{refined:relation},
				\begin{align}\label{sub:dm:classical-projection-removal}
					\left\|\int_{P\gH}|u^{\otimes k}\rangle\langle u^{\otimes k}|\dif\mu_P^{\eps,f_\eta}(u)
					-\int_{\gH}|u^{\otimes k}\rangle\langle u^{\otimes k}|\dif\mu^{\eps,f_\eta}(u)\right\|_{\gS^1(\gH^{(k)})}
					\leq C_{k,\cK}\Lambda_e^{-\beta}.
				\end{align}
				Combining \eqref{sub:dm:projected-convergence}, \eqref{sub:dm:quantum-projection-removal}, and \eqref{sub:dm:classical-projection-removal}, we obtain
				\begin{align*}
					\left\|\frac{k!}{\tau^k}(\Gamma^{f_\eta}_\tau)^{(k)}-
					\int_{\gH}|u^{\otimes k}\rangle\langle u^{\otimes k}|\dif\mu^{\eps,f_\eta}(u)\right\|_{\gS^1(\gH^{(k)})}
					\leq C_{k,\cK}\Lambda_e^{-\beta}.
				\end{align*}
				Finally, under the hypotheses $\eta\geq\tau^{-1/64}$ and $\eps\geq\tau^{-1/96}$, we may choose $\Lambda_e$ so that \eqref{refined:relation} holds and $\Lambda_e\to\infty$. This proves \eqref{key:bound:sub2} and completes the proof of Proposition~\ref{Pro:Hartree}.
		\end{proof}

	\section{From Hartree measure to the focusing $\Phi^6_1$ measure}\label{sec:hartree to phi61}
 The purpose of this section is to recover the focusing $\Phi^6_1$ measure. The argument has two parts. First, the Hartree interaction $\cW^\eps$ converges to the local interaction $\cW$ by mollifier approximation and dominated convergence. Second, the cutoff $f_\eta(\|u\|_{L^2}^2)$ converges to $\mathds{1}_{\{\|u\|_{L^2}\leq \cK\}}$, reducing the problem to continuity of the law of $\|u\|_{L^2}^2$ under the Gaussian measure $\mu_0$.
	The key input is the sharp normalizability at the threshold $\cK_c$ from \cite[Theorem 1.4]{OST22}, which provides the integrable bound needed for dominated convergence near the critical mass. 
	\begin{proposition}\label{prop:hartree to phi}
		Let $w$ satisfy Assumption~\ref{assum:on:w}. For any $\cK \in (0,\cK_c]$, let the cutoff $f_\eta$ satisfy Assumption~\ref{assum:cutoff} with  cutoff parameter $\cK$. For the interaction potential energy $\cW^\eps$ defined in \eqref{Hartree:interaction}, the Hartree partition function converges to the focusing $\Phi^6_1$ partition function as $\eps, \eta \to 0$:
		\begin{align*}
			\left| \int_{\gH} e^{\cW^\eps(u)} f_\eta(\|u\|_{L^2}^2) \dif \mu_0 (u)- \int_{\gH} e^{\cW(u)}  \mathds{1}_{\{ \|u\|_{L^2}\leq \cK \}}\dif \mu_{0}(u) \right| \to 0,
		\end{align*}
		where $\cW=\frac16 \|u\|_{L^6}^6$ is the classical potential energy and $\mu_0$ is the free Gibbs measure defined in \eqref{def:free Gibbs measure}.
	\end{proposition}
	\begin{proof}
		First, for any $\eps,\eta>0$, we have
		\begin{equation}\label{hartree:to:phi61:1}
			\begin{aligned}
			&	\left| \int_{\gH} e^{\cW^\eps(u)} f_\eta(\|u\|_{L^2}^2) \dif \mu_0 (u)- \int_{\gH} e^{\cW(u)}  \mathds{1}_{\{ \|u\|_{L^2}\leq \cK \}}\dif \mu_{0}(u) \right| 
			\\ &\leq \int_{\gH} \left|e^{\cW^\eps(u)} -e^{\cW(u)}  \right| f_\eta(\|u\|_{L^2}^2) \dif \mu_0 (u)+\int_{\gH} e^{\cW(u)}  \left| \mathds{1}_{\{ \|u\|_{L^2}\leq \cK\}} - f_\eta(\|u\|_{L^2}^2) \right| \dif \mu_{0}(u) 
			\\ &\leq \int_{\gH} \left|e^{\cW^\eps(u)} -e^{\cW(u)}  \right| \mathds{1}_{\{ \|u\|_{L^2}\leq \cK_c \}}   \dif \mu_0 (u)+\int_{\gH} e^{\cW(u)} \mathds{1}_{\{ \cK^2 -\eta\leq  \|u\|_{L^2}^2 \leq \cK^2 \}} \dif \mu_{0}(u) .
			\end{aligned}
		\end{equation}
		We aim to show that both terms on the right-hand side of \eqref{hartree:to:phi61:1} vanish as $\eps, \eta \to 0$. Concerning the first term, we claim that
		\begin{align}\label{hartree:to:phi61:1:a.s.cv}
			\lim_{\eps \to 0}  e^{\cW^\eps(u)}  = e^{\cW(u)}  , \qquad \mu_0\mathrm{-almost \,\, surely}.
		\end{align}
		By a fundamental factorization and H\"older's inequality, we deduce that for any $\eps>0$
			\begin{align}\label{dominated control 1}
				\left| \cW^\eps(u)-\cW(u) \right| &=\frac16 \left|\int_{\T} \left(w^\eps*|u|^2 -|u|^2 \right)
				\left(w^\eps*|u|^2 +|u|^2 \right)|u|^2  \dif x \right| \nonumber
				\\&\leq  \left\| w^\eps*|u|^2-|u|^2 \right\|_{L^3} \left(\left\| w^\eps*|u|^2 \right\|_{L^3} +  \|u\|_{L^6}^2\right)  \|u\|_{L^6}^2  \nonumber
				\\ &\leq  \left\| w^\eps*|u|^2-|u|^2 \right\|_{L^3}  \left(\left\| w^\eps \right\|_{L^1} \|u\|_{L^6}^2+  \|u\|_{L^6}^2\right)  \|u\|_{L^6}^2  \nonumber
				\\ &=2  \left\| w^\eps*|u|^2-|u|^2 \right\|_{L^3} \|u\|_{L^6}^4,
			\end{align}
		where we used Young's inequality in the third line and the fact $\| w^\eps \|_{L^1} =1$ in the last line. By the Sobolev embedding $H^{\frac13}\subset L^6$, we have that $|u|^2 \in L^3$, $\mu_0$-almost surely. Then, it follows that 
		\begin{align}\label{dominated control 2}
			\| w^\eps*|u|^2-|u|^2 \|_{L^3} \to0 , \qquad \mu_0-\mathrm{almost \, \, surely}.
		\end{align}
		By combining \eqref{dominated control 1}, \eqref{dominated control 2} together with the fact that $ \|u \|_{L^6}< \infty$ $\mu_0$-almost surely, we obtain
		\begin{align}\label{a.s.cv}
			\lim_{\eps \to 0}  \cW^\eps(u) = \cW(u) , \qquad \mu_0\mathrm{-almost \,\, surely}.
		\end{align}
		Then, the claim \eqref{hartree:to:phi61:1:a.s.cv} follows from the continuity of the exponential function.	
		Combining \eqref{lemma2.3:RS25} with \cite[Theorem 1.4]{OST22} yields that for any $\eps>0$
		\begin{align}\label{uniform integrability1}
			\left| e^{\cW^\eps(u)} - e^{\cW(u)}  \right| \mathds{1}_{\{\|u\|_{L^2}\leq \cK_c \}}   \leq  2 e^{\frac{1}{6}\|u\|_{L^6}^6} \mathds{1}_{\{\|u\|_{L^2}\leq \cK_c \}} \in L^1(\dif \mu_0).
		\end{align}
		In view of \eqref{hartree:to:phi61:1:a.s.cv}, \eqref{uniform integrability1}, and the dominated convergence theorem, it follows that
		\begin{align}\label{hartree:to:phi61:1:cv}
			\lim_{\eps \to 0} \int_{\gH} \left|e^{\cW^\eps(u)} -e^{\cW(u)}  \right| \mathds{1}_{\{ \|u\|_{L^2}\leq \cK_c \}}   \dif \mu_0 (u)=0.
		\end{align}
		
		Next, we turn to the second term on the right-hand side of \eqref{hartree:to:phi61:1} and show that
		\begin{align}\label{f:to:indicator}
			\lim_{\eta \to 0}\int_{\gH} e^{\frac16 \|u\|_{L^6}^6}  \mathds{1}_{\{ \cK^2-\eta \leq \|u\|_{L^2}^2 \leq \cK^2\}}  \dif \mu_{0} (u) =0,
		\end{align}
		for any $\cK\in(0,\cK_c]$. Consider the measure $\dif \mu_{\Phi}(u):= \mathds{1}_{\{  \|u\|_{L^2}\leq  \cK\}}  \dif \mu_{0} (u)$. It follows from \cite[Theorem 1.4]{OST22} that \begin{align*}
			e^{\frac16 \|u\|_{L^6}^6} \in L^1(\dif \mu_{\Phi}).
		\end{align*}
		To derive \eqref{f:to:indicator}, it suffices to show that
		\begin{align*}
			\lim_{\eta \to 0}	 \mu_{\Phi} \left( \left\{ \cK^2-\eta \leq \|u\|_{L^2}^2 \leq \cK^2\right\} \right)=0.
		\end{align*}
		Since $ \mu_{\Phi} ( \{ \cK^2-\eta \leq \|u\|_{L^2}^2 \leq \cK^2 \} ) = \mu_0 ( \{ \cK^2-\eta \leq \|u\|_{L^2}^2 \leq \cK^2 \} \cap  \{  \|u\|_{L^2}^2 \leq \cK^2 \}  ) = \mu_0 ( \{ \cK^2-\eta \leq \|u\|_{L^2}^2 \leq \cK^2 \}   )$, this reduces to proving that
		\begin{align*}
			\lim_{\eta \to 0}	 \mu_0 \left( \left\{ \cK^2-\eta \leq \|u\|_{L^2}^2 \leq \cK^2\right\} \right)=0.
		\end{align*}
		Under the Gaussian free field measure $\mu_0$, we can write $\|u\|_{L^2}^2=\sum_{j=1}^\infty |\alpha_j|^2$, where $\{\alpha_j\}_{j\geq 1}$ denotes a sequence of independent complex-valued Gaussian random variables. Specifically, $\mathrm{Re} \, \alpha_j$ and $\mathrm{Im} \,\alpha_j$ are independent real-valued mean-zero Gaussian random variables with variance $(2\lambda_j)^{-1}$. Consequently, $\{|\alpha_j|^2\}_{j\geq 1}$ are independent exponential random variables with density $f_j(x)= \lambda_j e^{-\lambda_j x}\mathds{1}_{\{x\geq0\}}$. 
		Define the random variable $S_N=\sum_{j=1}^N |\alpha_j|^2$, it follows that $S_N\xrightarrow{N \to \infty} S_\infty:=\|u\|_{L^2}^2$ $\mu_0$-almost surely.
		Our objective is to establish that $S_\infty$ admits a continuous density with respect to the Lebesgue measure. To this end, we first observe that the characteristic function of the random variable $|\alpha_j|^2$ is given by
		\begin{align*}
			\phi_j(t)=	\int e^{it |\alpha_j|^2}  \dif \mu_0 =\frac{\lambda_j}{\lambda_j-it}, \qquad t \in \R.
		\end{align*}
		Since $S_N$ converges to $S_\infty$ $\mu_0$-almost surely, the dominated convergence theorem and the independence of $|\alpha_j|^2$ imply that the characteristic function of $S_\infty$ is given by
		\begin{align*}
			\phi_{S_\infty}(t)=\lim_{N\to \infty}   \phi_N(t)=\lim_{N\to \infty}  \prod_{j=1}^N \int e^{it |\alpha_j|^2}  \dif \mu_0= \prod_{j=1}^\infty \frac{\lambda_j}{\lambda_j-it}, \qquad  t \in \R.
		\end{align*}
		We then verify 
		\begin{align*}
			\int_{\R} |	\phi_{S_\infty}(t) |\dif t &\leq 	\int_{\R}  \prod_{j=1}^\infty \left| \frac{\lambda_j}{\lambda_j-it}\right| \dif t \leq 	\int_{\R}  \prod_{j=1}^\infty \frac{\lambda_j}{\sqrt{\lambda_j^2+t^2} }\dif t
			\\ &\leq \int_{|t|\leq 1}1 \dif t+ \int_{|t|> 1} \frac{\lambda_1\lambda_2}{\sqrt{(\lambda_1^2+t^2)(\lambda_2^2+t^2)}} \dif t <\infty.
		\end{align*}
		By the inversion formula of the characteristic function (see, e.g., \cite[Theorem 3.3.5]{Dur19}), we obtain that $S_\infty$ has bounded continuous density $f_{S_\infty}$. Then, we obtain
		\begin{align*}
			\lim_{\eta \to 0}\mu_0 \left( \left\{ \cK^2-\eta \leq \|u\|_{L^2}^2 \leq \cK^2\right\} \right) =\lim_{\eta \to 0} \int_{\cK^2 -\eta}^{\cK^2} f_{S_\infty}(x) \dif x=0.
		\end{align*}
		Combining \eqref{hartree:to:phi61:1}, \eqref{hartree:to:phi61:1:cv} and \eqref{f:to:indicator}, we complete the proof of Proposition~\ref{prop:hartree to phi}.
	\end{proof}
	Following the proof of \eqref{f:to:indicator}, we also derive for any $\cK\in (0,\cK_c]$,
	\begin{align}\label{f:to:indicator1}
		\lim_{\eta \to 0} \int_{\gH} f_\eta(\|u\|_{L^2}^2) \dif \mu_{0} (u) = \int_{\gH}  \mathds{1}_{\{  \|u\|_{L^2} \leq \cK \}}  \dif \mu_{0} (u) .
	\end{align}

	\section{Conclusion of the main results}\label{Proof of main results}
We are now ready to complete the proofs of the main results. In Subsection~\ref{subsec:5.1}, we prove the convergence of the relative partition function in both the subcritical and critical regimes. In Subsection~\ref{conclusion:density:matrix}, we establish the convergence of the reduced density matrices in these two regimes. Finally, using the non-normalizability result of \cite[Theorem 2.2(b)]{LRS88}, we prove the divergence of the relative partition function stated in Theorem~\ref{blow:up:thm}.
	
	\subsection{Convergence of the relative partition function in the subcritical/critical regime}\label{subsec:5.1}
In this subsection, we prove the convergence of the relative partition function in both the mass-subcritical and critical regimes. The strategy is the same in the two cases. We first use the quantum-to-Hartree convergence at fixed mass cutoff obtained in the preceding sections, and then pass from the Hartree-level approximation to the focusing $\Phi^6_1$ measure.

	\begin{proof}[Proof of \eqref{main:result:partition} and \eqref{main:result:partition:sub}]
		Let $\cK\in (0,\cK_c]$ be arbitrary. Throughout this proof, and in the subsequent arguments, let the cutoff $f_\eta$ satisfy Assumption~\ref{assum:cutoff} with cutoff parameter $\cK$.
		Combining Proposition~\ref{Thm:Hartree} in the critical case, Proposition~\ref{Pro:Hartree} in the subcritical case, and Proposition~\ref{prop:hartree to phi}, and then letting $\tau \to \infty$ and $\eps,\eta \to 0$, we obtain
		\begin{align}\label{proof of main1}
			\left| \frac{\mathcal{Z}^{f_\eta}_{\tau}}{\mathcal{Z}^{f_\eta}_{\tau,0}} - \frac{\int_{\gH} e^{\frac16 \|u\|_{L^6}^6}  \mathds{1}_{\{ \|u\|_{L^2}\leq \cK \}} \dif \mu_{0}(u)  }{\int_{\gH}  \mathds{1}_{\{ \|u\|_{L^2}\leq \cK \}} \dif \mu_{0}(u)  }\right| \to 0.
		\end{align}
	The only point requiring comment is the admissible choice of the interaction scale. If $\cK<\cK_c$, then we may take $\eps \geq \tau^{-\frac{1}{96}} \to 0$ whereas at the threshold $\cK=\cK_c$, the argument requires $\eps  \gtrsim ( \log \tau)^{-\frac12} \to 0$. It remains to identify the denominator with the corresponding quantum free partition function. Combining \eqref{f:to:indicator1} with \eqref{convergence:for:feta}, we have, as $\tau \to \infty$ and $\eta\to 0$,
		\begin{align*}
			&\left|	 \frac{\mathcal{Z}^{f_\eta}_{\tau,0}}{\cZ_{\tau,0}} -
			\int_{\mathfrak{H}} \mathds{1}_{\{\|u\|_{L^2}\leq \cK\}} \, \dif \mu_0 (u) \right|
			\\ &\leq 	\left| \tr_{\gF(\gH)} \left( \Gamma_{\tau,0} f_\eta\left( \cN/\tau\right) \right)-\int_{\gH} f_\eta(\|u\|^2_{L^2}) \dif \mu_0(u) \right| +\left| \int_{\gH}  \mathds{1}_{\{  \|u\|_{L^2} \leq \cK \}} -f_\eta(\|u\|_{L^2}^2)  \dif \mu_{0} (u)  \right|  \to 0.
		\end{align*}
		Together with \eqref{proof of main1}, this proves \eqref{main:result:partition} and \eqref{main:result:partition:sub}.
	\end{proof}

		\subsection{Convergence of the density matrices in the subcritical/critical regime}\label{conclusion:density:matrix}
	We next prove the trace-class convergence of the reduced density matrices. The argument is again the same in the mass-subcritical and critical regimes.
	\begin{proof}[Proof of \eqref{main:result:density matrices} and \eqref{main:result:density matrices:sub}]
	Let $\cK\in(0,\cK_c]$ be arbitrary. By the triangle inequality,
		\begin{align}\label{triangle decomposition final trace norm}
		&	 \left\| \frac{k!}{\tau^k}(\Gamma^{f_\eta}_{\tau})^{(k)} -\int_{\gH}|u^{\otimes k}\rangle\langle u^{\otimes k}|\,\dif \mu^{\cK}(u) \right\|_{\gS^1(\gH^{(k)})}	\nonumber
		 \\ &\qquad  \leq  \left\| \frac{k!}{\tau^k}(\Gamma^{f_\eta}_{\tau})^{(k)}  
	- \int_{\gH}|u^{\otimes k}\rangle\langle u^{\otimes k}|\,\dif \mu^{\eps,f_\eta}(u) \right\|_{\gS^1(\gH^{(k)})}  \nonumber
		\\ &\qquad \qquad +\left\| \int_{\gH}|u^{\otimes k}\rangle\langle u^{\otimes k}|\,\dif \mu^{\eps,f_\eta}(u)  - \int_{\gH}|u^{\otimes k}\rangle\langle u^{\otimes k}|\,\dif \mu^{\cK}(u) \right\|_{\gS^1(\gH^{(k)})} .
		\end{align}
		The first term on the right-hand side of \eqref{triangle decomposition final trace norm} tends to $0$ by \eqref{key:bound2} in the critical regime and by \eqref{key:bound:sub2} in the subcritical regime. We now estimate the second term. By the same calculations as in \eqref{P:to:infty:newS1}, we deduce for all $k\geq 1$,
		\begin{align}\label{density matrix convergence4:new}
	&	\left\| \int_{\gH}|u^{\otimes k}\rangle\langle u^{\otimes k}|\,\dif \mu^{\eps,f_\eta}(u)  - \int_{\gH}|u^{\otimes k}\rangle\langle u^{\otimes k}|\,\dif \mu^{\cK}(u) \right\|_{\gS^1(\gH^{(k)})} 
		\\&\quad \leq  \frac{1}{ \int_{\gH} e^{\cW(u)}  \mathds{1}_{\{ \|u\|_{L^2}\leq \cK \}} \dif \mu_{0}(u)} \int_{\gH} \|u\|_{L^2}^{2k}	\left|  e^{\cW^\eps(u)} f_\eta(\|u\|_{L^2}^2) - e^{\cW(u)}  \mathds{1}_{\{ \|u\|_{L^2}\leq \cK \}}  \right|  \dif \mu_{0}(u)\nonumber 
		\\ &\quad + \frac{\int_{\gH} \|u\|_{L^2}^{2k}  \dif \mu^{\eps,f_\eta}(u)}{ \int_{\gH} e^{\cW(u)}  \mathds{1}_{\{ \|u\|_{L^2}\leq \cK \}} \dif \mu_{0}(u)} \int_{\gH} \left|  e^{\cW^\eps(u)} f_\eta(\|u\|_{L^2}^2) - e^{\cW(u)}  \mathds{1}_{\{ \|u\|_{L^2}\leq \cK \}}  \right|  \dif \mu_{0}(u). \nonumber 
		\end{align}
		By $\int_{\gH} \|u\|_{L^2}^{2k}  \dif \mu^{\eps,f_\eta}(u)\leq \cK^{2k}$ and Proposition~\ref{prop:hartree to phi}, it only remains to show that as $\eps,\eta \to 0$,
		\begin{equation}\label{density matrix convergence41}
			\begin{aligned}
				&	\int_{\gH} \|u\|_{L^2}^{2k}	\left|  e^{\cW^\eps(u)} f_\eta(\|u\|_{L^2}^2) - e^{\cW(u)}  \mathds{1}_{\{ \|u\|_{L^2}\leq \cK \}}  \right|  \dif \mu_{0}(u)
				\\	& \leq \int_{\gH}   \|u\|_{L^2}^{2k}	 \left|e^{\cW^\eps(u)} -e^{\cW(u)}  \right| \mathds{1}_{\{ \|u\|_{L^2}\leq \cK_c \}}   \dif \mu_0 (u)
		+\int_{\gH}  \|u\|_{L^2}^{2k}	 e^{\cW(u)} \mathds{1}_{\{ \cK^2 -\eta\leq  \|u\|_{L^2}^2 \leq \cK^2 \}} \dif \mu_{0}(u) \to0.
			\end{aligned}
		\end{equation}
		Repeating the arguments as in the proof of Proposition~\ref{prop:hartree to phi}, one can check that it holds $\mu_0$-almost surely,
		\begin{align*}
			\lim_{ \eps \to 0}  e^{\cW^\eps(u)}  \|u\|_{L^2}^{2k}	 = e^{\cW(u)}   \|u\|_{L^2}^{2k}	,
		\end{align*}
		for all $k\geq 1$.
		Also, it follows from \eqref{uniform integrability1} that
		\begin{equation*}
			\begin{aligned}
			 \|u\|_{L^2}^{2k}		\left| e^{\cW^\eps(u)} - e^{\cW(u)}  \right| \mathds{1}_{\{\|u\|_{L^2}\leq \cK_c \}}   \leq  C_{k,\cK_c} e^{\frac{1}{6}\|u\|_{L^6}^6} \mathds{1}_{\{\|u\|_{L^2}\leq \cK_c \}} \in L^1(\dif \mu_0).
			\end{aligned}
		\end{equation*}
		Hence, by using the dominated convergence theorem, we obtain that as $\eps \to 0$
		\begin{align}\label{hartree:to:phi61:1:cv:new}
			\int_{\gH}  \|u\|_{L^2}^{2k}	 \left|e^{\cW^\eps(u)} -e^{\cW(u)}  \right| \mathds{1}_{\{ \|u\|_{L^2}\leq \cK_c \}}   \dif \mu_0  (u)\to 0.
		\end{align}
		Invoking again \eqref{f:to:indicator} implies that as $\eta\to 0$
		\begin{align}
			\int_{\gH} \|u\|_{L^2}^{2k}	e^{\cW(u)} \mathds{1}_{\{ \cK^2 -\eta\leq  \|u\|_{L^2}^2 \leq \cK^2 \}} \dif \mu_{0}(u)\lesssim_{\cK,k}	\int_{\gH}  e^{\cW(u)} \mathds{1}_{\{ \cK^2 -\eta\leq  \|u\|_{L^2}^2 \leq \cK^2 \}} \dif \mu_{0}(u) \to0.
		\end{align}
		This, together with \eqref{hartree:to:phi61:1:cv:new}, implies \eqref{density matrix convergence41}. Hence, we complete the proof of \eqref{main:result:density matrices} and \eqref{main:result:density matrices:sub} .
	\end{proof}

	\subsection{Blow-up of the relative partition function above the threshold}
	In this subsection, we prove the divergence of the relative partition function by combining the Hartree-type convergence result in Proposition~\ref{Thm:Hartree} with the non-normalizability above the threshold established in \cite[Theorem 2.2(b)]{LRS88}. The mechanism is straightforward. Proposition~\ref{Thm:Hartree} yields a semiclassical lower bound in terms of the Hartree partition function, which diverges as $\eps \to 0$ above the threshold. Hence, the same divergence must occur for the quantum relative partition function.
	\begin{proof}[Proof of Theorem~\ref{blow:up:thm}]
		To apply Proposition~\ref{Thm:Hartree}, we first choose the cutoff $f_\eta$ as in Assumption~\ref{assum:cutoff}, with cutoff parameter $\cK> \cK_c$. For $\eta>0$ sufficiently small, we have $\cK^2- \eta > \frac12(\cK^2_c+\cK^2)$, and hence $\mathds{1}_{[0,\frac12(\cK^2_c+\cK^2) ] } (x)\leq f_\eta (x)\leq g(x)$. It follows from Proposition~\ref{Thm:Hartree} and \eqref{convergence:for:feta} that, for $1>\eps \geq M \left(\log  \tau\right)^{-\frac{1}{2}}$ and $\eta \in [\tau^{-\frac{1}{64}}, \frac12 \cK^2)$, as $\tau \to \infty$,
		\begin{align}\label{blow:up11}
			\frac{\tr_{\gF(\gH)}\left(e^{-\bbH_\tau}  g(\cN/\tau)  \right)}{\tr_{\gF(\gH)} \left(e^{-\mathbb{H}_{\tau,0}}   \right)}&  \geq  	\frac{\tr_{\gF(\gH)}\left(e^{-\bbH_\tau}  f_\eta(\cN/\tau) \right)}{\tr_{\gF(\gH)} \left(e^{-\mathbb{H}_{\tau,0}}   \right)} 
			\to \int_{\gH} e^{\cW^\eps(u)} f_\eta(\|u\|_{L^2}^2) \dif \mu_0(u) 
			\\ & \qquad \qquad \geq \int_{\gH} e^{\cW^\eps(u)} \mathds{1}_{\{\|u\|^2_{L^2}\leq \frac12\cK^2_c+\frac12 \cK^2 \}} \dif \mu_0(u)  =:    \cZ^\eps. \nonumber
		\end{align}
	To conclude the proof, it remains to verify that $	 \lim_{\eps \to 0} \cZ^\eps=\infty$.
		Given $\eps\geq 0$ and $R>0$, we define an auxiliary function
		\begin{align*}
			F_{\eps,R}(u):=
			\begin{cases}
				e^{\min  \left\{ \cW^\eps(u), R \right\}}, &\eps >0,\\
				e^{\min  \left\{ \cW(u), R \right\}}, & \eps =0.
			\end{cases}
		\end{align*}
		Then, for any $\eps>0$ and $R>0$,  it holds that
		\begin{align}\label{Zeps:below}
			\cZ^\eps \geq \int_{ \gH}  	F_{\eps,R}(u) \mathds{1}_{\{\|u\|^2_{L^2}\leq \frac12\cK^2_c+\frac12 \cK^2 \}}   \dif \mu_0(u).
		\end{align}
		Observe that $F_{0,R}$ increases to $e^{\cW(u)}$ almost surely. By the monotone convergence theorem, we obtain
		\begin{align*}
			\int_{ \gH}  	F_{0,R}(u) \mathds{1}_{\{\|u\|^2_{L^2}\leq \frac12\cK^2_c+\frac12 \cK^2 \}} \dif \mu_0(u)  \xrightarrow{R\to \infty} \int_{ \gH}  e^{\cW(u)} \mathds{1}_{\{\|u\|^2_{L^2}\leq \frac12\cK^2_c+\frac12 \cK^2 \}} \dif \mu_0(u) =\infty ,
		\end{align*}
		where the last equality follows from the non-normalizability of the focusing $\Phi^6_1$ measure above the optimal threshold, see \cite[Theorem 2.2 (b)]{LRS88}. Hence, for any $A>0$, there exists $R(A)$ such that 
		\begin{align}\label{blow:up2}
			\int_{ \gH}  	F_{0,R(A)}(u) \mathds{1}_{\{\|u\|^2_{L^2}\leq \frac12\cK^2_c+\frac12 \cK^2 \}}  \dif \mu_0(u) \geq 2A.
		\end{align}
		From \eqref{a.s.cv}, we deduce that $	F_{\eps,R(A)}(u)\xrightarrow{\eps \to 0} e^{\min  \left\{ \cW(u), R(A) \right\}}$, $\mu_0$-almost surely. Since $	F_{\eps,R(A)}\leq e^{R(A)}$, the dominated convergence theorem yields
		\begin{align*}
			\int_{ \gH}  	F_{\eps,R(A)}(u) \mathds{1}_{\{\|u\|^2_{L^2}\leq \frac12\cK^2_c+\frac12 \cK^2 \}} \dif \mu_0(u) \xrightarrow{\eps \to 0} \int_{ \gH}  	F_{0,R(A)}(u) \mathds{1}_{\{\|u\|^2_{L^2}\leq \frac12\cK^2_c+\frac12 \cK^2 \}}  \dif \mu_0(u).
		\end{align*}
		Together with \eqref{Zeps:below} and \eqref{blow:up2}, this implies that for any $A>0$, there exists $\eps(A)>0$ such that for any $\eps \in(0,\eps(A))$
		\begin{align*}
			\cZ^\eps \geq 	\int_{ \gH}  	F_{\eps,R(A)}(u) \mathds{1}_{\{\|u\|^2_{L^2}\leq \frac12\cK^2_c+\frac12 \cK^2 \}}  \dif \mu_0(u)  \geq \frac12 \int_{ \gH}  	F_{0,R(A)}(u) \mathds{1}_{\{\|u\|^2_{L^2}\leq \frac12\cK^2_c+\frac12 \cK^2 \}}  \dif \mu_0(u) \geq A.
		\end{align*}
		Hence, we obtain $ \lim_{\eps \to 0} \cZ^\eps=\infty$. This completes the proof of Theorem~\ref{blow:up:thm}.
	\end{proof}

		\appendix
	\renewcommand{\appendixname}{Appendix~\Alph{section}}
	\renewcommand{\theequation}{A.\arabic{equation}}
	\section{Quantitative convergence rate for the truncated free Gibbs state}\label{sec:appendix A}
	In this appendix, we prove the auxiliary convergence statement for the free Gibbs state that is used repeatedly in the variational argument. The point of the result is to justify, with a quantitative error, the replacement of the quantum mass cutoff $f_\eta(\cN/\tau)$ by its classical counterpart $f_\eta(\|u\|_{L^2}^2)$. Since the cutoff
	depends on the small parameter $\eta$, we keep track of the $\eta$-dependence in the estimates.
	\begin{proposition}\label{lemma:non-interacting}
	Let $\cK>0$ be arbitrary, and let $\tau >0$ be sufficiently large. For any $\eta \in \big[ \tau^{-\frac{1}{64}}, \frac12 \cK^2\big)$, let the cutoff $f_\eta$ satisfy Assumption~\ref{assum:cutoff} with cutoff parameter $\cK$. Let $\Gamma_{\tau,0}$  be the free Gibbs state defined in \eqref{def:free:Gibbs:state}. Then, as $\tau \to \infty$, $\eta \to 0$,
	\begin{align}\label{convergence:for:feta}
	\left|	\tr_{\gF(\gH)} \left( \Gamma_{\tau,0} f_\eta (\cN/\tau)\right) - \int_{\mathfrak{H}}  f_\eta(\|u\|_{L^2}^2) \,\dif \mu_0(u) \right| \to 0.
	\end{align}
	 More generally, for every $g\in C^\infty_c([0,\infty); \R_+)$,
	\begin{align}\label{convergence:for:any:g}
		\lim_{\tau \to \infty}	 \tr_{\gF(\gH)} \left( \Gamma_{\tau,0} g\left( \cN/\tau \right) \right) =  \int_{\mathfrak{H}}  g(\|u\|_{L^2}^2) \,\dif \mu_0(u) .
	\end{align}
\end{proposition}
The proof of Proposition~\ref{lemma:non-interacting} proceeds in two steps. First, we restrict the problem to a finite-dimensional low-
frequency space $P\gH$ and prove a projected quantum-to-classical estimate with a quantitative rate. This is the content of Lemma~\ref{project:convergence:rate} below. In this finite-dimensional setting, the coherent-state decomposition turns the quantum trace into an integral of Poisson expectations. The desired convergence then follows by comparing the Poisson random variable with its mean and by comparing the finite-dimensional quantum Gaussian density with the corresponding classical Gaussian density. Second, we remove the projection, both on the quantum side and on the classical side, using the high-frequency tail of the free Gibbs state and of the Gaussian measure.
	\begin{lemma}\label{project:convergence:rate}
	We use the same conventions as in Proposition~\ref{lemma:non-interacting}. Assume that $\eta^{-1} \leq \Lambda_e^{\frac18} \leq \tau^{\frac{1}{32}}$, and let $P:=\mathds{1}(h\leq \Lambda_e)$ be the orthogonal projection on $ \gH$. Let $\Gamma_{\tau,0,P} $ be the free Gibbs state on $\gF(P\gH)$. Then, there exists a constant $C>0$ independent of $\tau$, $\Lambda_e$ and $\eta$, such that
		\begin{align}\label{convergence on same index2:new}
			\Big| \tr_{\gF(P\gH)} \left( \Gamma_{\tau,0,P} f_\eta (\cN_P/\tau)\right) - \int_{\mathfrak{H}} f_\eta (\|Pu\|_{L^2}^2) \, \dif \mu_0 (u) \Big| \leq C \tau^{-\frac14}.
		\end{align}
	\end{lemma}

	\begin{proof}
Since $P\gH=\mathrm{span}\{u_1,\ldots,u_J\}$, we identify $\gF(P\gH)\simeq\bigotimes_{j=1}^J \gF(\mathbb{C} u_j)$. For the corresponding occupation-number basis, we write
$$|n_1,\ldots,n_J \rangle =|n_1\rangle_1 \otimes \cdots \otimes |n_J\rangle_J ,\qquad |n_j\rangle_j= \frac{1}{\sqrt{n_j!}}(a^\dagger_j)^{n_j} |0\rangle . $$ Then, $a_j^\dagger a_j  |n_j\rangle_j =n_j |n_j\rangle_j$, and it follows that
	\begin{equation}\label{number:hit:on:basis}
	\begin{aligned}
		 f_\eta\left(\frac{\cN_P}{\tau}\right) \big|n_1,\ldots,n_J \big\rangle &= f_\eta \left( \frac{\sum_{j=1}^J n_j}{\tau} \right) \big|n_1,\ldots,n_J \big\rangle ,
		 \\ e^{-\sum_{j=1}^J\lambda_ja^\dagger_ja_j/\tau}\big|n_1,\ldots,n_J \big\rangle &= e^{-\sum_{j=1}^J \lambda_jn_j/\tau} \big|n_1,\ldots,n_J \big\rangle.
	\end{aligned}
		\end{equation}
	For $u=\sum_{j=1}^J \alpha_j u_j\in P\gH$, the coherent state factorizes into one-mode states:
	\begin{align}\label{eq:extension}
		\xi(u)=\bigotimes_{j=1}^J \xi(\alpha_j u_j)=e^{-\sum_{j=1}^J |\alpha_j|^2/2} \sum_{n_1,\ldots,n_J\geq 0} \prod_{j=1}^J \frac{\alpha_j^{n_j}}{\sqrt{n_j!}} |n_1,\ldots,n_J \rangle .
	\end{align}
Using the resolution of the identity in coherent states \eqref{eq:resolution_coherent2}, together with \eqref{number:hit:on:basis} and \eqref{eq:extension}, we obtain
	\begin{align*}
		&\tr_{\gF(P\gH)} \left( \Gamma_{\tau,0,P} f_\eta (\cN_P/\tau)\right) 
		=\prod_{j=1}^J \frac{\tau}{\pi} (1-e^{-\lambda_j/\tau})
		\int_{P\gH} \left\langle \xi \big( \sqrt{\tau}u \big), e^{-\bbH_{\tau,0,P}}f_\eta (\cN_P/\tau) \xi \big( \sqrt{\tau}u \big) \right\rangle \,\dif u \\
		&= \int_{P\gH}  \prod_{j=1}^J \frac{\tau}{\pi} (1-e^{-\lambda_j/\tau}) e^{-\tau  |\alpha_j|^2}
		\sum_{n_1,\ldots,n_J\geq 0} f_\eta \left( \frac{\sum_{j=1}^J n_j}{\tau} \right)
		\prod_{j=1}^J \frac{(\tau|\alpha_j|^2e^{-\lambda_j/\tau})^{n_j}}{n_j!} \,\dif u \\
		&=  \int_{P\gH} \prod_{j=1}^J  \frac{q_j}{\pi}  e^{-q_j |\alpha_j|^2}
		\sum_{n_1,\ldots,n_J\geq 0} f_\eta \left( \frac{\sum_{j=1}^J n_j}{\tau} \right)
		\prod_{j=1}^J e^{-\tau|\alpha_j|^2e^{-\lambda_j/\tau}}
		\frac{(\tau|\alpha_j|^2e^{-\lambda_j/\tau})^{n_j}}{n_j!} \,\dif u,
	\end{align*}
	where we defined $q_j:= \tau(1-e^{-\lambda_j/\tau})$. Let $\{\mathbf{X}_j\}_{j=1}^J$ be independent Poisson random variables with means $\tau |\alpha_j|^2e^{-\lambda_j/\tau} $, respectively. Then the preceding identity becomes
\begin{align*}
	 \tr_{\gF(P\gH)} \left( \Gamma_{\tau,0,P} f_\eta (\cN_P/\tau)\right) = \int_{P\gH} \mathbf{E} \left[ f_\eta  \left(\frac{\sum_{j=1}^J \mathbf{X}_j}{\tau} \right)\right]   \prod_{j=1}^J  \frac{q_j}{\pi}  e^{-q_j |\alpha_j|^2}   \dif u.
\end{align*}
We compare this expression with the classical Gaussian integral in three elementary steps: the Poisson variables are first replaced by their means; the means are then replaced by $\sum_j |\alpha_j|^2$; finally, the Gaussian density with variances $q_j^{-1}$ is compared with the density with variances $\lambda_j^{-1}$. This gives
\begin{align*}
  \Big| &\tr_{\gF(P\gH)}  \left( \Gamma_{\tau,0,P} f_\eta (\cN_P/\tau)\right) - \int_{\gH}  f_{\eta} (\|Pu\|_{L^2}^2)\dif \mu_{0}(u) \Big|
	 \\ &\leq  \int_{P\gH} \mathbf{E} \Big| f_\eta  \Big( \frac{1}{\tau} \sum_{j=1}^J \mathbf{X}_j \Big) -  f_\eta  \Big(\sum_{j=1}^J |\alpha_j|^2 e^{-\lambda_j/\tau} \Big)  \Big|   \prod_{j=1}^J  \frac{q_j}{\pi}  e^{-q_j |\alpha_j|^2}   \dif u \\
	 &\quad +   \int_{P\gH}  \Big| f_\eta  \Big(\sum_{j=1}^J |\alpha_j|^2 e^{-\lambda_j/\tau} \Big) - f_\eta  \Big(\sum_{j=1}^J |\alpha_j|^2 \Big) \Big|   \prod_{j=1}^J  \frac{q_j}{\pi}  e^{-q_j |\alpha_j|^2}   \dif u \\
	 &\quad +  \int_{P\gH}   f_\eta  \Big(\sum_{j=1}^J |\alpha_j|^2 \Big) \Big|   \prod_{j=1}^J  \frac{q_j}{\pi}  e^{-q_j |\alpha_j|^2} -  \prod_{j=1}^J  \frac{\lambda_j}{\pi}  e^{-\lambda_j |\alpha_j|^2}  \Big| \dif u =: \cE_{\mathrm{Pois}}+\cE_{\mathrm{cut}}+ \cE_{\mathrm{dens}}.
\end{align*}
We now estimate these three errors separately. Since $\sum_{j=1}^J \mathbf{X}_j$ is Poisson with mean $\tau \sum_{j=1}^J |\alpha_j|^2e^{-\lambda_j/\tau} $, Assumption~\ref{assum:cutoff} and H\"older's inequality imply
\begin{align*}
	\cE_{\mathrm{Pois}} & \leq    \frac{\|f_\eta'\|_{L^\infty} }{\tau} \int_{P\gH} \mathbf{E} \Big|     \sum_{j=1}^J \mathbf{X}_j  -  \tau \sum_{j=1}^J |\alpha_j|^2 e^{-\lambda_j/\tau}  \Big|   \prod_{j=1}^J  \frac{q_j}{\pi}  e^{-q_j |\alpha_j|^2}   \dif u \\
	&\leq   \frac{\eta^{-1} }{\tau}   \int_{P\gH} \left( \mathbf{E} \Big|     \sum_{j=1}^J \mathbf{X}_j  -  \tau \sum_{j=1}^J |\alpha_j|^2 e^{-\lambda_j/\tau}  \Big|^2 \right)^{1/2}   \prod_{j=1}^J  \frac{q_j}{\pi}  e^{-q_j |\alpha_j|^2}   \dif u \\
	&\leq   \frac{\eta^{-1} }{\sqrt{\tau}}   \int_{P\gH} \Big( \sum_{j=1}^J |\alpha_j|^2 e^{-\lambda_j/\tau} \Big)^{1/2}   \prod_{j=1}^J  \frac{q_j}{\pi}  e^{-q_j |\alpha_j|^2}   \dif u 
	\leq   \frac{\eta^{-1} }{\sqrt{\tau}} \Big( \sum_{j=1}^J e^{-\lambda_j/\tau} \frac{q_j}{\pi}   \int_{\mathbb{C}} |\alpha_j|^2   e^{-q_j |\alpha_j|^2}   \dif \alpha_j \Big)^{1/2}  \\
	& =  \frac{\eta^{-1} }{\sqrt{\tau}} \Big( \sum_{j=1}^J q_j^{-1}e^{-\lambda_j/\tau}  \Big)^{1/2}
	=   \frac{\eta^{-1} }{\sqrt{\tau}}  \Big(\sum_{j=1}^J \frac{1}{\tau(e^{\lambda_j/\tau}-1)}  \Big)^{1/2} 
	\leq  \frac{\eta^{-1} }{\sqrt{\tau}}  \tr(Ph^{-1}) ^{1/2} \leq C \tau^{-\frac14}.
\end{align*}
Here the last inequality follows from $\eta^{-1}\leq \Lambda_e^{1/8} \leq  \tau^{1/32}$ and $\tr(Ph^{-1}) \leq \tr(h^{-1})\lesssim 1$. The second error $\cE_{\mathrm{cut}}$ measures the difference between $e^{-\lambda_j/\tau}$ and $1$ inside the cutoff. Using the Lipschitz bound on $f_\eta$, we obtain
\begin{align*}
	\cE_{\mathrm{cut}} & \leq \|f_\eta' \|_{L^\infty} \sum_{j=1}^J (1-e^{-\lambda_j/\tau}  )    \frac{q_j}{\pi}  \int_{\mathbb{C}}  |\alpha_j|^2 e^{-q_j |\alpha_j|^2}   \dif \alpha_j
	\\ & \leq \eta^{-1}  \sum_{j=1}^J (1-e^{-\lambda_j/\tau}  )    q_j^{-1} \leq \eta^{-1} \sqrt{\Lambda_e} \tau^{-1} \leq C \tau^{-\frac14}.
\end{align*}
For the last error $\cE_{\mathrm{dens}}$, we first reduce the $J$-dimensional total variation to one-dimensional comparisons by telescoping the product densities:
\begin{align}\label{J-dimensional:variation}
\int_{P\gH}& \Big| \prod_{j=1}^J  \frac{q_j}{\pi}  e^{-q_j |\alpha_j|^2}    -  \prod_{j=1}^J  \frac{\lambda_j}{\pi}  e^{-\lambda_j |\alpha_j|^2}   \Big| \dif u \nonumber \\
&\leq  \sum_{k=1}^J  \int_{\mathbb{C}^J} \Big( \prod_{j<k}  \frac{q_j}{\pi}  e^{-q_j |\alpha_j|^2} \Big) \Big|  \frac{q_k}{\pi}  e^{-q_k |\alpha_k|^2}  - \frac{\lambda_k}{\pi}  e^{-\lambda_k |\alpha_k|^2}   \Big| \Big( \prod_{j>k} \frac{\lambda_j}{\pi}  e^{-\lambda_j |\alpha_j|^2} \Big) \dif \alpha_1 \cdots \dif \alpha_J \\
&= \sum_{j=1}^J  \int_{\mathbb{C}} \Big| \frac{q_j}{\pi}  e^{-q_j |\alpha_j|^2}    -    \frac{\lambda_j}{\pi}  e^{-\lambda_j |\alpha_j|^2}   \Big| \dif \alpha_j
= \sum_{j=1}^J  \int_{0}^\infty \Big| \frac{q_j}{\lambda_j} e^{-q_j r /\lambda_j} -e^{-r}  \Big| \dif r. \nonumber
\end{align}
In the last step we used polar coordinates and the change of variables $r=\lambda_j |\alpha_j|^2$.
	By the elementary inequalities $\frac{x}{1+x}\leq 1- e^{-x}\leq x$, $x \geq 0$, and by $\lambda_j \leq \Lambda_e \leq \tau$, we have
	\begin{align*}
	\frac12 \leq \frac{1}{1+\lambda_j/\tau}	\leq  \frac{q_j}{\lambda_j}  =\frac{1- e^{-\lambda_j/\tau}}{\lambda_j/\tau} \leq 1.
	\end{align*}
	Combining this bound with the mean value theorem gives
	\begin{align*}
	&  \int_{0}^\infty \Big| \frac{q_j}{\lambda_j} e^{-q_j r /\lambda_j} -e^{-r}  \Big| \dif r
	\leq \Big| 1-\frac{q_j}{\lambda_j}  \Big| \int_{0}^\infty \sup_{s\in[\frac12,1]} |1-sr| e^{-sr } \dif r
	\\ &\quad \leq  \Big| 1-\frac{q_j}{\lambda_j}  \Big| \int_{0}^\infty (1+r) e^{-r/2 } \dif r \leq C  \Big| 1-\frac{q_j}{\lambda_j}  \Big| =C\Big| \frac{\lambda_j/\tau-1+e^{-\lambda_j/\tau}}{\lambda_j/\tau}\Big| \leq \frac{C\lambda_j}{2\tau},
	\end{align*}
	where we used $0\leq x-1+e^{-x}\leq x^2/2$, for $x\geq 0$ in the last inequality. Inserting this estimate into \eqref{J-dimensional:variation}, we obtain
	\begin{align*}
		\cE_{\mathrm{dens}}  \leq C\tau^{-1} \sum_{j=1}^J \lambda_j \leq C\Lambda_e^{\frac32} \tau^{-1} \leq C\tau^{-\frac14}.
	\end{align*}
Combining the estimates on $	\cE_{\mathrm{Pois}} $, $	\cE_{\mathrm{cut}} $, $	\cE_{\mathrm{dens}} $  proves \eqref{convergence on same index2:new}.
	\end{proof}
	
We now remove the projection. The quantum error is controlled by the expected number of particles in the high-frequency sector, while the classical error is controlled by the Gaussian high-frequency tail. 
	We first have that, for any $\tau>0$, $\Lambda_e>0$ and $\eta\in \big(0,\frac12 \cK^2\big)$,
	\begin{align*}
		& \tr_{\gF(P\gH)} \Big( f_\eta(\cN_P/\tau)    \Gamma_{\tau,0,P} \Big)   =   \tr_{\gF(P\gH) \otimes \gF(Q\gH)}  \Big( (f_\eta(\cN_P/\tau)   \Gamma_{\tau,0,P}  )\otimes \Gamma_{\tau,0,Q}  \Big)
		\\ &\quad =  \tr_{\gF(\gH)}  \Big( f_\eta(\cN_{P,\cU} / \tau )\cU^*(\Gamma_{\tau,0,P} \otimes \Gamma_{\tau,0,Q}) \cU\Big)
		=   \tr_{\gF(\gH)}  \Big( f_\eta(\cN_{P,\cU} /\tau)  \Gamma_{\tau,0} \Big),
	\end{align*}
	where we defined $\cN_{P,\cU}=\cU^* (\cN_P \otimes \mathds{1}_{\gF(Q\gH)})\cU$. Since
	$\cN=\cN_{P,\cU}+\cU^*(\mathds{1}_{\gF(P\gH)}\otimes \cN_Q)\cU$, the mean-value formula implies that for $\eta^{-1}\leq \Lambda_e^{\frac18}$,
	\begin{align}\label{diff:Gibbs:cutoff}
		&  \Big|\tr_{\gF(P\gH)} \Big( f_\eta(\cN_P/\tau)    \Gamma_{\tau,0,P} \Big)  -\tr_{\gF(\gH)}  \left(  f_\eta \left(\cN/\tau\right) \Gamma_{\tau,0} \right)\Big| \nonumber \\
		&=   \Big| \tr_{\gF(\gH)}  \Big( f_\eta(\cN_{P,\cU} /\tau)  \Gamma_{\tau,0} \Big)  -\tr_{\gF(\gH)}  \left(  f_\eta \left(\cN/\tau\right) \Gamma_{\tau,0} \right)\Big| \\
		& \leq \|f_\eta'\|_{L^\infty} \tr_{\gF(\gH)}  \Big( \Gamma_{\tau,0}\cU^* (\mathds{1}_{\gF(P\gH)}\otimes \cN_{Q})\cU/\tau \Big)
		\leq \eta^{-1} \tr_{\gF(Q\gH)} \big( \Gamma_{\tau,0,Q} \cN_Q/\tau \big) \leq C\eta^{-1}\Lambda_e^{-\frac12} \leq C\Lambda_e^{-\frac14},  \nonumber
	\end{align}
	where we used \eqref{Wp-W:second} in the last line.
On the classical side, the same Lipschitz estimate gives
		\begin{align}\label{int:f(u)-f(Pu)}
		\Big|  \int_{\gH} &  f_\eta\left( \|u\|_{L^2}^2 \right) 
		\dif \mu_{0}(u)  - \int_{\gH}  f_\eta\left( \|Pu\|_{L^2}^2 \right) 
		\dif \mu_{0}(u) \Big| 
		\leq \|f_\eta'\|_{L^\infty} \int_{\gH} ( \|u\|_{L^2}+  \|Pu\|_{L^2})  \|u-Pu\|_{L^2} \dif \mu_0(u)   \nonumber \\
		&\lesssim \eta^{-1} \Big(\int_{\gH} \|u\|_{L^2}^2\dif \mu_0(u)\Big)^{\frac12}  \Big(\int_{\gH} \|u-Pu\|_{L^2}^2\dif \mu_0(u)\Big)^{\frac12}   \nonumber \\
		&\lesssim \eta^{-1}\Big(\sum_{j\geq 1}\lambda_j^{-1}\Big)^{\frac12} \Big(\sum_{\lambda_j> \Lambda_e}\lambda_j^{-1}\Big)^{\frac12}
		\lesssim \eta^{-1}\Lambda_e^{-\frac14} \leq  C\Lambda_e^{-\frac18},
		\end{align}
		where in the last line we used $\tr(h^{-1})=\sum_{j\geq 1}\lambda_j^{-1}<\infty$ and $\eta^{-1}\leq \Lambda_e^{\frac18}$.
\begin{proof}[Proof of Proposition~\ref{lemma:non-interacting}]
	 For $\tau$ sufficiently large,  since $\eta \geq  \tau^{-\frac{1}{64}}$, we may choose $\Lambda_e \to \infty $ such that $\eta^{-1} \leq \Lambda_e^{\frac18} \leq \tau^{\frac{1}{32}}$. Combining \eqref{convergence on same index2:new}, \eqref{diff:Gibbs:cutoff} and \eqref{int:f(u)-f(Pu)}, we obtain
	 \begin{align*}
	 		&\Big|\tr_{\gF(\gH)} \left( \Gamma_{\tau,0} f_\eta (\cN/\tau)\right) - \int_{\mathfrak{H}}  f_\eta(\|u\|_{L^2}^2) \,\dif \mu_0(u) \Big|  \\
	 		&\leq 
	 		\Big| \tr_{\gF(\gH)}  \left(  f_\eta \left(\cN/\tau\right) \Gamma_{\tau,0} \right) -\tr_{\gF(P\gH)} \Big( f_\eta(\cN_P/\tau)    \Gamma_{\tau,0,P} \Big)  \Big|  \\
	 		&\quad +  \Big| \tr_{\gF(P\gH)} \left( \Gamma_{\tau,0,P} f_\eta (\cN_P/\tau)\right) - \int_{\mathfrak{H}} f_\eta (\|Pu\|_{L^2}^2) \, \dif \mu_0 (u) \Big|  \\
	 		&\quad +\Big|   \int_{\gH}  f_\eta\left( \|Pu\|_{L^2}^2 \right) 
	 		\dif \mu_{0}(u)-  \int_{\gH}  f_\eta\left( \|u\|_{L^2}^2 \right) 
	 		\dif \mu_{0}(u)   \Big| 
	 		\leq C\Lambda_e^{-\frac18} + C\tau^{-\frac14}.
	 \end{align*}
	 Letting $\tau \to \infty$ and $\Lambda_e \to \infty$ proves \eqref{convergence:for:feta}. 
	The convergence in \eqref{convergence:for:any:g} follows by the same argument with the fixed test function $g\in C^\infty_c([0,\infty);\R_+)$. Since $g$ does not depend on $\eta$, the preceding estimates are simpler and no $\eta$-dependent loss has to be tracked. This completes the proof of the proposition.
\end{proof}

		\renewcommand{\appendixname}{Appendix~\Alph{section}}
	\renewcommand{\theequation}{B.\arabic{equation}}
	\section{Uniform projected Oh-Sosoe-Tolomeo bound in the mass subcritical regime}\label{sec:appendix:B}
	In this appendix, we prove the uniform estimate \eqref{mian:appendix A}. The argument is an adaptation of \cite[Sec. 4]{OST22}. There are two differences in the present setting: the nonlinear weight is raised to the power $1+\varsigma$, and the field is replaced by its low-frequency projection $Pu$. Since we work strictly below the	critical mass, the extra factor $1+\varsigma$ can be absorbed by choosing $\varsigma>0$ sufficiently small. The projection introduces no further difficulty, as the estimates in the original argument remain valid after projection. We therefore record only the points at which these two modifications enter the proof.
	
	We first recall the following Gagliardo-Nirenberg-Sobolev inequality on $\mathbb{T}$, see \cite[Lemma 3.3]{OST22}.
		For every $m > 0$, there exists  $C(m) > 0$ such that
		\begin{align}\label{GNS:ineq}
			\|v\|_{L^6(\T)}^6 \leq  \left( C_{\mathrm{GNS}}(1,6)+m \right)  \|v'\|^2_{L^2(\T)} 	 	\|v\|^4_{L^2(\T)} 	 	+C(m) 	\|v\|^6_{L^2(\T)},
		\end{align}
	for all $v \in H^1(\T)$.
	Here, $C_{\mathrm{GNS}}(1,6)=3\|Q\|_{L^2(\R)}^{-4}$ and $\|Q\|_{L^2(\R)}$ is the optimal threshold given by \eqref{optimal:threshold}. We use the following standard Littlewood--Paley projections on $\mathbb{T}$:
	\begin{align*}
		u_{\leq k} :=
			 \sum_{|n| \leq 2^k} \widehat{u}(n) e^{2\pi i n x} , \quad  k\geq 0, \qquad \qquad 
		u_{\geq k}:= 
		\begin{cases}
			 	 \sum_{ |n| > 2^{k-1}} \widehat{u}(n) e^{2\pi i n x},   & k\geq 1,
			 	 \\ u , &k=0.
		\end{cases}
	\end{align*}
	Then $u=	u_{\leq k-1}+ 	u_{\geq k}$ for every $k\geq 1$.	
	\begin{proof}[Proof of \eqref{mian:appendix A}]
	We first split the integral into the regions $\{\|Pu\|_{L^6}\leq 1\}$ and $\{\|Pu\|_{L^6}> 1\}$:
	\begin{align*}
		\int_{\gH} e^{\frac{1+\varsigma}{6}\|Pu\|_{L^6}^6} \mathds{1}_{\{  \|Pu\|_{L^2} \leq \cK_s \}}  \dif \mu_{0}(u) &=  \int_{\gH} e^{\frac{1+\varsigma}{6}\|Pu\|_{L^6}^6} \mathds{1}_{\{  \|Pu\|_{L^2} \leq \cK_s , \, \|Pu\|_{L^6}\leq 1 \}}  \dif \mu_{0}(u) 
		\\ &\quad +\int_{\gH} e^{\frac{1+\varsigma}{6}\|Pu\|_{L^6}^6} \mathds{1}_{\{  \|Pu\|_{L^2} \leq \cK_s , \, \|Pu\|_{L^6}> 1 \}}  \dif \mu_{0}(u) . 
	\end{align*}
	The first contribution is trivially bounded by $e^{\frac{1+\varsigma}{6}}$, so it remains to control the latter region.
	For $k\geq 1$, define $E_k$ by
	\begin{equation}\label{def:Ek}
		\begin{aligned}
			E_k&:=\{\|(Pu)_{\geq 0}\|_{L^6 } >1 , \ldots,\|(Pu)_{\geq k-1}\|_{L^6 } >1 , \|(Pu)_{\geq k}\|_{L^6 } \leq 1 \} 
			\\ &\quad \subset \{ \|(Pu)_{\geq k-1}\|_{L^6 } >1 , \|(Pu)_{\geq k}\|_{L^6 } \leq 1 \}.
		\end{aligned}
	\end{equation}
The sets $E_k$ are disjoint. Since $(Pu)_{\geq k}=0$ for $k> J:=\tr[P]$, then $ E_k=\emptyset$ for any $k> J$. Thus, the $E_k$ form a disjoint decomposition of $\{ \|(Pu)\|_{L^6 } >1 \}$, and it is harmless to sum over
all $k\geq 1$, since only finitely many terms are non-zero: 
	\begin{equation*}\label{series uniform bound}
		\int_{\gH} e^{\frac{1+\varsigma}{6}\|Pu\|_{L^6}^6} \mathds{1}_{\{  \|Pu\|_{L^2} \leq \cK_s , \, \|Pu\|_{L^6}> 1 \}}  \dif \mu_{0}(u)=
			 \sum_{k\geq 1} \int_{\gH} e^{\frac{1+\varsigma}{6}\|Pu\|_{L^6}^6} \mathds{1}_{\{  \|Pu\|_{L^2} \leq \cK_s , \, E_k \}}  \dif \mu_{0}(u) . 
	\end{equation*}
As in \cite[(4.4)]{OST22}, H\"older's inequality and Young’s inequality imply that, for $u$ satisfying $\|(Pu)_{\geq k}\|_{L^6}\leq 1$ and some small $\gamma>0$ (to be chosen later):
		\begin{align*}
			\int_{\T} |Pu|^6 \dif x&=\int_{\T} |(Pu)_{\leq k-1}+(Pu)_{\geq k}|^6 \dif x \leq \sum_{\ell=0}^6  {6 \choose \ell} \int_{\T}|(Pu)_{\leq k-1}|^{6-\ell}|(Pu)_{\geq k}|^{\ell} \dif x \nonumber
			\\ &  \leq 	\int_{\T} |(Pu)_{\leq k-1}|^6 \dif x+ \sum_{\ell=1}^6 {6 \choose \ell} \left(  \frac{6-\ell}{6} \gamma \|(Pu)_{\leq k-1}\|_{L^6}^{6} + 	\frac{\ell}{6}\gamma^{-\frac{6-\ell}{\ell}}  	 \right)
			\\ &\leq \left( 1 +(2^6-1)\gamma \right) 	\int_{\T} |(Pu)_{\leq k-1}|^6 \dif x +C_6(\gamma). \nonumber
		\end{align*}
		Combining this bound with \eqref{def:Ek} yields
		\begin{align}\label{bound:Pu:L6}
				\int_{\gH} &  e^{\frac{1+\varsigma}{6}\|Pu\|_{L^6}^6} \mathds{1}_{\{  \|Pu\|_{L^2} \leq \cK_s , \, E_k \}}  \dif \mu_{0}(u) 
		\\ &	\leq 	\int_{\gH} e^{\frac{1+\varsigma}{6}\|Pu\|_{L^6}^6} \mathds{1}_{\{  \|Pu\|_{L^2} \leq \cK_s , \, \|(Pu)_{\geq k-1}\|_{L^6 } >1 , \|(Pu)_{\geq k}\|_{L^6 } \leq 1 \}}  \dif \mu_{0}(u) \nonumber
			\\ & \leq e^{C_6(\gamma)} \int_{\gH} \mathrm{exp}\Big(\frac{(1+\varsigma)(1+63\gamma)}{6}\int_{\T } |(Pu)_{\leq k-1}|^6 \dif x\Big) \mathds{1}_{\{  \|Pu\|_{L^2} \leq \cK_s , \, \|(Pu)_{\geq k-1}\|_{L^6 } >1  \}}  \dif \mu_{0}(u) .\nonumber
		\end{align}	 
	 Applying \eqref{GNS:ineq} to $(Pu)_{\leq k-1}$ with some small $m>0$ (to be chosen later), and using H\"older's inequality, we deduce from \eqref{bound:Pu:L6} that
	\begin{align}\label{bound:Pu:L6a}
		&	\int_{\gH} e^{\frac{1+\varsigma}{6}\|Pu\|_{L^6}^6} \mathds{1}_{\{  \|Pu\|_{L^2} \leq \cK_s , \, E_k \}}  \dif \mu_{0}(u) 
		\\ & \leq e^{C_6(\gamma)+C(m)\cK_s^6}  \int_{\{   \|(Pu)_{\geq k-1}\|_{L^6 } >1  \}} \mathrm{exp} \Big(\frac{(C_{\mathrm{GNS}}(1,6)+m)\cK_s^4(1+\varsigma)(1+63\gamma)}{6}\int_{\T } |\partial_x(Pu)_{\leq k-1}|^2 \dif x\Big) \dif \mu_{0}(u) \nonumber 
		\\ &\leq e^{C_6(\gamma)+C(m)\cK_s^6}  \left(  \int_{\gH} \mathrm{exp} \Big( \big(3\|Q\|^{-4}_{L^2(\R)}+m\big)\frac{\cK_s^4(1+\varsigma)(1+63\gamma)(1+\iota)}{6}\int_{\T } |\partial_x(Pu)_{\leq k-1}|^2 \dif x\Big)   \dif \mu_{0}(u)  \right)^{\frac{1}{1+\iota}} \nonumber
		\\ &\qquad \times \left[ \mu_0\left(  \|(Pu)_{\geq k-1}\|_{L^6 } >1 \right) \right]^{\frac{\iota}{1+\iota}}. \nonumber
	\end{align}
	This is the only point where the extra factor $1+\varsigma$ enters. Since $\cK_s <\|Q\|_{L^2(\R)}$, we can choose $m,\varsigma,\gamma,\iota>0$ sufficiently small such that 
	\begin{align*}
		\frac{\cK_s^4}{\|Q\|^4_{L^2(\R)}}\left( 1+\frac{m}{3}\|Q\|^4_{L^2(\R)} \right)(1+\varsigma)(1+63\gamma)(1+\iota)  =:2c<1.
	\end{align*}
	Writing $u=\sum_{j\in\mathbb Z}\alpha_j e^{2\pi i j x}$ under $\mu_0$, the first integral on the right-hand side of \eqref{bound:Pu:L6a} is bounded by
		\begin{align}\label{uniform:projection1}
		&\int_{\gH}\exp\left(c \int_{\T}|\partial_x(Pu)_{\leq k-1}|^2\,\dif x\right)\,\dif\mu_0(u) \leq \int_{\gH}\exp\left(c \sum_{1\leq  |j| \leq 2^{k-1}} (2\pi j)^2 |\alpha_j|^2 \right)\,\dif\mu_0(u) \nonumber
		\\&\qquad=\prod_{ 1\leq  |j|\leq 2^{k-1}}\frac{\lambda_j}{\lambda_j-c(2\pi j)^2}
		\leq \prod_{1\leq |j|\leq 2^{k-1}}(1-2c)^{-1}
		\leq (1-2c)^{-2^k},
	\end{align}
	where in the second inequality we used $c(2\pi j)^2 /\lambda_j \leq 2c $. For the second probability on the right-hand side of \eqref{bound:Pu:L6a}, the Gaussian high-frequency $L^6$ tail estimate from \cite[(2.8)]{OST22} gives
	\begin{align*}
			\mu_0\left(\|u_{\geq k-1}\|_{L^6}>1\right)
		\leq C_0 \mathrm{exp} \Big(-C 2^{\frac43 k}  \Big),
	\end{align*}
	for some constants $C_0,C>0$ independent of $k$.
	The same bound remains valid uniformly after inserting the projection $P$. Indeed, the proof of the Gaussian high-frequency tail estimate in \cite[(2.2)--(2.8)]{OST22} is uniform with respect to an additional upper frequency cutoff. Applying that argument to the truncated field $(Pu)_{\geq k-1}$ gives
	\begin{align}\label{uniform:projection}
	\mu_0\left(\|(Pu)_{\geq k-1}\|_{L^6}>1\right)
	\leq C_0 \mathrm{exp} \Big(-C 2^{\frac43 k}  \Big).
\end{align}
Substituting \eqref{uniform:projection1} and \eqref{uniform:projection} into \eqref{bound:Pu:L6a}, we obtain
	\begin{align*}
		\eqref{bound:Pu:L6a} \leq C_{\iota,\gamma,m,\cK_s} \left( 1- 2c \right)^{-\frac{2^k}{1+\iota}}  \mathrm{exp} \Big(-C \frac{\iota}{1+\iota}2^{\frac43 k}  \Big) = C_{\iota,\gamma,m,\cK_s}	\mathrm{exp}\left( -C\frac{\iota}{1+\iota}2^{\frac43 k}+\frac{2^k}{1+\iota} \log \frac{1}{1-2c} \right).
	\end{align*}
	The right-hand side is summable in $k$, uniformly in $P$, because the negative term has order $2^{4k/3}$ whereas the positive term has order $2^k$. Therefore
	\begin{align*}
		&\sum_{k\geq 1} \int_{\gH} e^{\frac{1+\varsigma}{6}\|Pu\|_{L^6}^6} \mathds{1}_{\{  \|Pu\|_{L^2} \leq \cK_s , \, E_k \}}  \dif \mu_{0}(u)  
	\leq C_{\iota,\gamma,m,\cK_s} \sum_{k\geq 1}
		\mathrm{exp}\left( -C\frac{\iota}{1+\iota}2^{\frac43 k}+\frac{2^k}{1+\iota} \log \frac{1}{1-2c} \right)<\infty .
	\end{align*}
This proves \eqref{mian:appendix A}. 	\end{proof}

	\begingroup
	\makeatletter
	\let\addcontentsline\@gobblethree
	\makeatother

	\section*{Acknowledgments}
	The authors would like to thank Jacky J. Chong and Tadahiro Oh for inspiring discussions and helpful comments on an earlier version of this paper.
Part of this work was completed while the first author was visiting LMU Munich, and he would like to thank the members of the Department of Mathematics at LMU for the warm hospitality and support. P.T.N. gratefully acknowledges the financial support from the European Research Council via the ERC Consolidator Grant RAMBAS (Project Nr. 10104424) and the Deutsche Forschungsgemeinschaft (DFG, German Research Foundation) -- TRR 352, Project-ID 470903074. L.L. and R.Z. are grateful to the financial supports by National Key R\&D Program of China (No. 2022YFA1006300) and the financial supports of the NSFC (No. 12271030, No. 12426205) and the financial supports by the Deutsche Forschungsgemeinschaft (DFG, German Research Foundation) --SFB 1283, Project-ID 317210226. 
	
	\endgroup

	\def\cprime{$'$} \def\ocirc#1{\ifmmode\setbox0=\hbox{$#1$}\dimen0=\ht0
		\advance\dimen0 by1pt\rlap{\hbox to\wd0{\hss\raise\dimen0
				\hbox{\hskip.2em$\scriptscriptstyle\circ$}\hss}}#1\else {\accent"17 #1}\fi}

\end{document}